\documentclass[11pt]{article}
\usepackage{amsthm, amsmath, amssymb, amsfonts, url, booktabs, tikz, setspace, fancyhdr, bm}
\usepackage{hyperref}
\usepackage[letterpaper, margin=1in]{geometry}
\usepackage{setspace}
\usepackage{hyperref, enumerate}
\usepackage[shortlabels]{enumitem}
\usepackage[babel]{microtype}
\usepackage[english]{babel}
\usepackage[capitalise]{cleveref}
\usepackage{comment}
\usepackage{bbm}
\usepackage{csquotes}
\usepackage{mathabx}
\usepackage{tikz}
\usepackage{cases}
\usepackage{graphicx}
\usepackage{float}
\usepackage{amsmath}
\usepackage{tabularx}
\usepackage{makecell}
\usepackage{cite}
\usepackage{array}
\usepackage{diagbox}
\usepackage{multirow}
\usepackage{booktabs}   % 专业表格线
\usepackage{caption}    % 子标题支持
\usepackage{lipsum}

\usepackage[linesnumbered,ruled]{algorithm2e}

\usetikzlibrary{positioning}

\counterwithin{figure}{section}

\newtheorem{theorem}{Theorem}[section]
\newtheorem{prop}[theorem]{Proposition}
\newtheorem{lemma}[theorem]{Lemma}
\newtheorem{cor}[theorem]{Corollary}

\newtheorem{claim}[theorem]{Claim}

\theoremstyle{definition}
\newtheorem{defn}[theorem]{Definition}
\newtheorem*{defn-non}{Definition}

\newtheorem{ques}[theorem]{Question}

\newtheorem{rmk}[theorem]{Remark}

\newenvironment{poc}{\begin{proof}[Proof of claim]}{\end{proof}}

\newcommand{\PP}{\mathcal{P}}
\newcommand{\WW}{\mathcal{W}}
\newcommand{\QQ}{\mathcal{Q}}

\newcommand{\CC}{\mathcal{C}}

\newcommand{\floor}[1]{\lfloor #1\rfloor}

\newcommand{\supp}{\mathrm{supp}}
\usepackage{todonotes}

\title{Private Information Retrieval over Graphs}

\author{
Gennian Ge\thanks{School of Mathematical Sciences, Capital Normal University, Beijing 100048, China. Emails: 2220502174@cnu.edu.cn and gnge@zju.edu.cn. Gennian Ge is supported by the National Key Research and Development Program of China under Grant 2020YFA0712100, the National Natural Science Foundation of China under Grant 12231014, and Beijing Scholars Program.}
\and 
Hao Wang\footnotemark[3]
\and
Zixiang Xu\thanks{Extremal Combinatorics and Probability Group (ECOPRO), Institute for Basic Science (IBS), Daejeon 34126, South Korea. Email: zixiangxu@ibs.re.kr. Supported by IBS-R029-C4.}
\and
Yijun Zhang\thanks{School of Mathematical Sciences, University of Science and Technology of China, Hefei, Anhui 230026, China. Email: zyjshuxue@mail.ustc.edu.cn.}
}

\begin{document}
\maketitle

\begin{abstract}

The problem of Private Information Retrieval (PIR) in graph-based replication systems has received significant attention in recent years. A systematic study was conducted by Sadeh, Gu, and Tamo, where each file is replicated across two servers and the storage topology is modeled by a graph. The PIR capacity of a graph \( G \), denoted by \( \mathcal{C}(G) \), is defined as the supremum of retrieval rates achievable by schemes that preserve user privacy, with the rate measured as the ratio between the file size and the total number of bits downloaded. This paper makes the following key contributions.
\begin{itemize}
    \item The complete graph \( K_N \) has emerged as a central benchmark in the study of PIR over graphs. The asymptotic gap between the upper and lower bounds for \(\mathcal{C}(K_{N})\) was previously 2 and was only recently reduced to \( \frac{5}{3} \) by Kong, Meel, Maranzatto, Tamo, and Ulukus. We shrink this gap to \( \frac{3}{4(e-2)} \approx 1.0444 \), bringing it close to resolution. More precisely,
    \begin{enumerate}
 \item[\textup{(1)}]  Sadeh, Gu, and Tamo proved that $\mathcal{C}(K_N)\le \frac{2}{N+1}$ and conjectured this bound to be tight. We refute this conjecture by establishing the strictly stronger bound \(
\mathcal{C}(K_N) \le \frac{1}{e-2}\cdot\frac{1}{N}
 \approx \frac{1.3922}{N}. \)
Our proof combines graph-theoretic insights with information-theoretic analysis. Within the same framework, we also improve the upper bound for the balanced complete bipartite graph \(\mathcal{C}(K_{N/2,N/2})\) from \(\frac{2}{N}\) to \(\frac{1.5415}{N}\) asymptotically.
\item[\textup{(2)}] The first general lower bound on $\mathcal{C}(K_N)$, due to Sadeh, Gu, and Tamo, was
$\frac{2^{N-1}}{2^{N-1}-1}\cdot\frac{1}{N}$, which was recently sharpened to
$\frac{6}{5-2^{3-N}}\cdot\frac{1}{N}$ by Kong, Meel, Maranzatto, Tamo, and Ulukus.
We provide explicit, systematic constructions that further improve this bound, proving \(\mathcal{C}(K_N)\ \ge \Bigl(\frac{4}{3}-o(1)\Bigr)\cdot\frac{1}{N},\)
which in particular implies $\mathcal{C}(G) \ge \Bigl(\frac{4}{3}-o(1)\Bigr)\cdot\frac{1}{|G|}$ for every graph $G$.
       \end{enumerate}
    \item We establish a conceptual bridge between deterministic and probabilistic PIR schemes on graphs. This connection has significant implications for reducing the required subpacketization in practical implementations and is of independent interest. We also design a general probabilistic PIR scheme that performs particularly well on sparse graphs.

\end{itemize}
\end{abstract}

\newpage
\tableofcontents
\cleardoublepage

\section{Introduction}
\subsection{Background}
Private Information Retrieval (PIR), introduced by Chor, Goldreich, Kushilevitz, and Sudan~\cite{Chor1998private}, enables a user to retrieve a specific file from a public (replicated) database without revealing its identity (\emph{privacy requirement}) while guaranteeing successful reconstruction (\emph{reliability requirement}). In the information-theoretic framework established by Sun and Jafar~\cite{Sun2017capacity}, servers are assumed non-colluding and file length can be arbitrarily large, rendering the upload (query) cost negligible so that the download (answer) cost dominates. Under this model, the efficiency of a PIR scheme is measured by the \emph{PIR rate}, which is the number of desired message bits retrieved per bit of download, and the supremum over all schemes is the \emph{PIR capacity}, which is determined exactly in the fully replicated setting~\cite{Sun2017capacity}. Capacity-achieving designs typically fall into two complementary paradigms: (i) \emph{deterministic} coding schemes inspired by blind interference alignment~\cite{Sun2017capacity}, and (ii) \emph{probabilistic} schemes whose randomized queries keep every candidate message equiprobable at each server~\cite{Samy2021asymmetric,Tian2019capacity}.

A crucial systems consideration is \emph{subpacketization}, the segmentation of each file into multiple equal-sized parts during retrieval~\cite{Shah2014subpacketization}. Excessive subpacketization can impose unrealistic constraints on file sizes, metadata overheads, and update complexity, limiting deployability even when rates are information-theoretically optimal. The importance of subpacketization extends beyond PIR, most notably to \emph{coded caching}, where reducing subpacketization is a central practical challenge and has driven extensive research~\cite{Yan2017caching,Yan2017cachingcolor,xu2024caching,shangguan2018caching}. Thus, both capacity and subpacketization must be addressed for implementable protocols.

While coding for storage is powerful and widely studied, large-scale systems still heavily rely on \emph{replication} for simplicity, ease of updates, and high availability, despite lower storage efficiency and fault tolerance compared to coded storage~\cite{Silberstein2015repetitioncode,Sipser1996expandercode,Yohananov2019codegraph}. At the same time, fully replicating every file across all servers becomes increasingly impractical at modern scales. This tension motivates PIR under \emph{partial replication}, naturally modeled by (hyper)graphs: vertices represent servers and (hyper)edges represent files, each stored on the incident servers. When each file is stored on exactly $r$ servers, we obtain an $r$-replication system; in particular, $2$-replication corresponds to simple graphs~\cite{Raviv2020graph}. Within this framework, graph topology directly affects privacy structure and retrieval efficiency.

Building on this viewpoint, a growing body of work has analyzed graph-based PIR. For $2$-replication (simple graphs), Raviv, Tamo and Yaakobi~\cite{Raviv2020graph} initiated the study under colluding servers and showed how topology shapes achievable schemes. Capacity bounds for regular graphs and exact capacity for cycles were obtained in~\cite{Banawan2019graph}. Upper bounds and improved constructions for several canonical topologies, such as complete graphs, stars, and complete bipartite graphs were developed in~\cite{Sadeh2023bound}, while the exact capacity of the $4$-file star was determined in~\cite{Yao2023star}. More recently, Kong, Meel, Maranzatto, Tamo, and Ulukus~\cite{kong2025newcapacityboundspir} extended this line to \emph{multigraph-based} PIR with finite uniform multiplicity, in which each edge of the underlying simple graph is replaced by parallel edges, enabling richer replication patterns and new design possibilities. 

Parallel to these advances, several orthogonal generalizations have been pursued: collusion-resistant PIR and its fundamental limits~\cite{Sun2017colluding,Sun2018colluding,Freij2017colluding}, coded PIR with storage-retrieval tradeoffs via MDS coding~\cite{Banawan2018coded,Kumar2019MDScoded,Tajeddine2018MDScoded,Zhang2019codedcolluding,Zhang2019PIRMDS,Zhang2019PIRarray}, and symmetric PIR that enforces mutual privacy so the user learns nothing beyond the requested file~\cite{Sun2019symmetric,Wang2019symmetric}. Within this framework, Meel and Ulukus~\cite{meel2025symmetric} further investigated graph-based symmetric PIR and established fundamental results in this setting. We refer to~\cite{Ulukus2022survey,Vithana2023survey} for comprehensive surveys.

Despite substantial progress, the capacity of general (hyper)graph-based PIR systems remains largely unresolved even without collusion; this is true already for basic $2$-replication models. Among graph families, \emph{complete graphs} play a pivotal role: as observed in~\cite{Sadeh2023bound} and formalized in our work, any PIR scheme over the complete graph $K_N$ can be converted into a PIR scheme over \emph{any} $N$-vertex simple graph (see~\cref{lemma:graph contain}). Consequently, tightening the capacity bounds for $K_N$ is a key step toward understanding broader classes of graphs. In this paper, we significantly sharpen these bounds for complete (and balanced complete bipartite) graphs, and we build a conceptual bridge between deterministic and probabilistic graph-based PIR that enables \emph{very small subpacketization}, especially on sparse topologies.

\subsection{Problem statement and notations}\label{section:problemstatement}
For a positive integer \( a \), let \([a] := \{1, 2, \ldots, a\}\). For integers \( a < b \), the discrete interval is denoted by \([a, b] := \{a, a+1, \ldots, b\}\). Given a family \( \{A_i\}_{i \in I} \) of objects indexed by a set \( I \), and a subset \( J \subseteq I \), we denote the corresponding subcollection by \( A_J := \{A_i : i \in J\} \). For discrete random variables \( A \) and \( B \), \( H(A) \) denotes the entropy of \( A \), and \( I(A; B) \) the mutual information between \( A \) and \( B \). Furthermore, for a collection of random variables \( \mathcal{A} = \{A_1, A_2, \dots, A_n\} \), we use the shorthand \( H(\mathcal{A}) = H(A_1, A_2, \dots, A_n) \). The notations \( H(B \mid \mathcal{A}) \) and \( I(\mathcal{A}; B) \) are defined analogously.

In this work, we study the PIR problem in the setting of graph-based storage systems. Specifically, we mainly consider a \emph{2-replication system} modeled by a simple undirected graph, where each file is stored on exactly two servers, and any pair of servers shares at most one file. We further assume that servers are \emph{non-colluding}, meaning that they do not exchange information with one another.

Let \( G = (V, E) \) be a simple undirected graph with \( |V| = N \) vertices and \( |E| = K \) edges. The vertices \( V = \{S_1, S_2, \dots, S_N\} \) represent servers, and each edge \( e = \{i, j\} \in E \) corresponds to a unique file \( W_{\{i,j\}} \) stored on both servers \( S_i \) and \( S_j \). For notational simplicity, we denote \( W_{\{i,j\}} \) simply as \( W_{i,j} \) (or equivalently, \( W_{j,i} \)).

Let \( \mathcal{W} = \{W_e : e \in E\} \) denote the set of total files in the system. For each server \( S_i \), we define the set of files stored on it as
\[
\mathcal{W}_{S_i} := \{ W_{i,j} : (i,j) \in E \}.
\]
As standard, we assume all files are of equal size. This model is widely applicable, as it captures a key aspect of many distributed storage systems: their architecture is fundamentally defined by an underlying graph structure that governs file placement.

In a PIR scheme, each file is typically divided into several equal-sized sub-files. For instance, a file $W$ may be partitioned as $W = (W_1, W_2, W_3)$. To ensure privacy, the user privately and uniformly selects a random permutation $\sigma$ of the index set $[3]$, and defines the permuted sequence $(w_1, w_2, w_3) := (W_{\sigma(1)}, W_{\sigma(2)}, W_{\sigma(3)})$. Thus, recovering the original file $W$ is equivalent to retrieving all permuted sub-files $w_1, w_2, w_3$.

We assume that all files are equally subpacketized into $L$ sub-files, where $L$ is referred to as the \emph{subpacketization level} of the PIR scheme. This parameter $L$ serves as a key performance metric. For each file, the user privately generates a permutation of $[L]$ to shuffle its sub-file indices; let $\mathcal{G}$ denote the collection of these permutations. Each sub-file is treated as a single symbol (a dit) in a finite field $\mathbb{F}_q$. Consequently, the entropy of each file is $L$ dits, i.e., $H(W) = L$.

As we have mentioned above, excessive subpacketization might hinder practical deployment, as the number of sub-files could exceed the actual file size. Therefore, among schemes with equal retrieval rates, lower subpacketization is generally preferred for better efficiency and applicability. 

In the PIR problem, a user privately generates $\theta \in E(G)$ and wishes to retrieve $W_{\theta}$ (i.e., retrieve all permuted sub-files of $W_{\theta}$) while keeping $\theta$ hidden from each server. To achieve this privacy goal, various scheme design methodologies have been developed. Broadly speaking, there exist two primary approaches in the literature: the deterministic approach, as used in~\cite{Sun2017capacity}, and the probabilistic approach, as in~\cite{Samy2021asymmetric,Tian2019capacity}.

In the deterministic approach, once \( \theta \) and \( \mathcal{G} \) are fixed, the queries to be sent to the servers are completely determined. That is, the user constructs queries \( Q_1^{[\theta, \mathcal{G}]}, Q_2^{[\theta, \mathcal{G}]}, \dots, Q_N^{[\theta, \mathcal{G}]} \), where \( Q_i^{[\theta, \mathcal{G}]} \) is sent to server \( S_i \). To ensure privacy, the core idea is to enforce \emph{message symmetry} in the queries, so that all files are treated uniformly in the view of each server.

In contrast, the probabilistic approach incorporates additional private randomness. Specifically, it relies on a set of universal queries that, from the perspective of any individual server, are used for all possible requested files with equal probability. The user first constructs a universal query set based on \( \theta \) and \( \mathcal{G} \), and then privately samples a random variable \( R \). The queries sent to each server are selected from the universal query set according to the value of \( R \). That is, the queries \( Q_1^{[\theta, \mathcal{G}, R]}, Q_2^{[\theta, \mathcal{G}, R]}, \dots, Q_N^{[\theta, \mathcal{G}, R]} \) are constructed based on \( \theta \), \( \mathcal{G} \), and \( R \). Privacy is ensured by enforcing \emph{query symmetry}, meaning that the distribution of the query sent to each server is independent of the requested file index.

Throughout the paper, when the dependence on \( \theta \), \( \mathcal{G} \), and possibly \( R \), is clear from context, we omit these superscripts and simply write \( Q_i \) to refer to \( Q_i^{[\theta, \mathcal{G}]} \) or \( Q_i^{[\theta, \mathcal{G}, R]} \). Let \( \QQ := \{Q_i : i \in [N]\} \) denote the set of all queries sent to the servers. We assume that the requested messages in each \( Q_i \) consist of linear combinations of sub-files from the files stored on server \( S_i \), which is a standard assumption satisfied by all known PIR schemes.

Since the user has no prior knowledge about the contents of the files, the queries must be statistically independent of the file data, which can be formally expressed as:
\[
I(\mathcal{W}; \QQ) = 0.
\]

Upon receiving query \( Q_i \), server \( S_i \) returns an answer \( A_i \), which is deterministically computed as a function of \( Q_i \) and the set of files it stores, denoted by \( \WW_{S_i} \). Hence, for every \( i \in [N] \), the response satisfies:
\begin{equation}
    H(A_i \mid Q_i, \WW_{S_i}) = 0. \label{eq:question answer}
\end{equation}

Let \( \Pi \) be a PIR scheme on a graph-based storage system with \( N \) servers and \( K \) files, each of entropy \( L \). The scheme defines, for every desired file index \( \theta \in [K] \), a random query distribution and a deterministic answering mechanism for each server. A PIR scheme must satisfy two fundamental requirements: \emph{reliability} and \emph{privacy}, defined as follows.

\paragraph{Reliability:} The user must be able to recover the requested file \( W_\theta \) with zero probability of error using the answers \( A_1, A_2, \dots, A_N \) received from the servers and the queries \( \QQ \). Formally,
    \begin{equation}
        H(W_\theta \mid A_{[N]}, \QQ) = 0, \label{eq:reliability}
    \end{equation}
    where \( A_{[N]} := \{A_1, A_2, \dots, A_N\} \).

\paragraph{Privacy:} No individual server should gain any information about the identity of the requested file. Formally, this means that for every server \( i \in [N] \) and every requested file \( \theta \in E(G) \), the joint distribution of the query \( Q_i^{[\theta]} \), the answer \( A_i^{[\theta]} \), and the set of files \( \WW_{S_i} \) stored on server \( S_i \) must be independent of \( \theta \). That is,
    \begin{equation}
        (Q_i^{[\theta]},A_i^{[\theta]},\WW_{S_i}) \sim (Q_i^{[1]},A_i^{[1]},\WW_{S_i}), \label{eq:privacy}
    \end{equation}
    where \( X \sim Y \) indicates that the random variables \( X \) and \( Y \) are identically distributed.

    An equivalent formulation in terms of mutual information is given by
    \begin{equation*}
        I(\theta; Q_i,A_i,\WW_{S_i})=0,
    \end{equation*}
    which, using the relationship in \eqref{eq:question answer}, simplifies to the conditional entropy condition
    \begin{equation}
        H(\theta \mid Q_i, \WW_{S_i}) = H(\theta). \label{eq:privacy_1}
    \end{equation}

The central performance measure of a PIR scheme is its \emph{retrieval rate}, which represents the ratio of the total amount of information downloaded from the servers.

\begin{defn}[Retrieval rate of a PIR scheme]
The \emph{retrieval rate} of a PIR scheme \( \Pi \), denoted by \( R(\Pi) \), is defined as the ratio between the retrieved file size in bits, and the total number of bits downloaded as answers, i.e.,
\[
R(\Pi) := \frac{L}{\sum_{i=1}^{N} H(A_i)},
\]
where \( A_i \) is the answer sent by server \( S_i \) under scheme \( \Pi \). This quantity measures the ratio of the desired information retrieved to the total amount of information downloaded from the servers.

\end{defn}

It is standard in the PIR literature to assume that the file size can be made arbitrarily large. Under this assumption, the upload cost (that is, the cost of sending queries) becomes negligible compared to the download cost (that is, the cost of receiving answers). As a result, the above definition coincides with alternative definitions that condition on the queries:
\[
R(\Pi) := \frac{L}{\sum_{i=1}^{N} H(A_i \mid Q_i)}.
\]
We now define the \emph{capacity} of a graph-based PIR system.

\begin{defn}[PIR capacity of a graph]
Let \( G = (V, E) \) be a simple undirected graph with \( N = |V| \) servers and \( K = |E| \) files, where each file is stored on the two servers corresponding to its endpoints. The \emph{PIR capacity} of \( G \), denoted by \( \CC(G) \), is defined as
\[
\CC(G) := \sup_{\Pi} R(\Pi),
\]
where the supremum is taken over all valid PIR schemes \( \Pi \) on \( G \), and \( R(\Pi) \) denotes the retrieval rate achieved by scheme \( \Pi \). A PIR scheme is said to be \emph{capacity-achieving} or \emph{optimal} if it achieves this supremum.
\end{defn}

Then the main theoretical problems in this field can be stated as follows.
\begin{ques}\label{ques:Main}
    For given graph $G$, determine $\mathcal{C}(G)$.
\end{ques}
While this is a central question from a theoretical perspective, a practically motivated problem might be more important.
\begin{ques}\label{ques:Practical}
    For given graph \(G\), construct PIR schemes over $G$ that achieve large capacity and operate with small subpacketization.

\end{ques}

\subsection{Related work and known results}
The study of PIR in graph-based replication systems has garnered growing interest in recent years. This model was first introduced by Raviv, Tamo, and Yaakobi~\cite{Raviv2020graph}, and further explored in works such as~\cite{Banawan2019graph, Jia2020secure}. A more systematic treatment was later given by Sadeh, Gu, and Tamo~\cite{Sadeh2023bound}, who focused specifically on~\cref{ques:Main}.

The \emph{complete graph} on \( N \) vertices, denoted \( K_N \), is the undirected graph in which every pair of distinct vertices is connected by exactly one edge. Among all graph classes, the complete graph \( K_N \) plays a particularly central role in the study of PIR. Any PIR scheme designed for \( K_N \) can be naturally restricted to its subgraphs, making \( K_N \) a universal benchmark. Moreover, its high degree of symmetry renders it both theoretically appealing and technically tractable.

The graph \( K_N \) contains \( \binom{N}{2} \) edges, each corresponding to a distinct file in the PIR setting. In~\cite{Sadeh2023bound,Banawan2019graph}, the authors established the following general bounds on its PIR capacity.
\begin{theorem}[\cite{Sadeh2023bound,Banawan2019graph}]\label{thm:known-Kn}
The PIR capacity \( \mathcal{C}(K_N) \) satisfies the following bounds:
\begin{enumerate}
  \item[\textup{(1)}]For any integer $N\ge 3$, \(
 \frac{2^{N-1}}{2^{N-1} - 1 }\cdot\frac{1}{N}\le \mathcal{C}(K_N) \le \frac{2}{N+1}.\)
  \item[\textup{(2)}] \( \mathcal{C}(K_3) = \frac{1}{2}\ \textup{and\ } \mathcal{C}(K_4) \ge \frac{3}{10}.\)
\end{enumerate}
\end{theorem}
 Beyond the complete graphs, Sadeh, Gu, and Tamo~\cite{Sadeh2023bound} investigated a variety of graph families, including stars, cycles, wheel graphs, and complete bipartite graphs, deriving refined capacity bounds through a blend of combinatorial constructions and information-theoretic techniques. A key general upper bound they established asserts that for any graph \( G \),
\[
\mathcal{C}(G) \le \min\left\{ \frac{\Delta(G)}{|E(G)|}, \, \frac{1}{\nu(G)} \right\},
\]
where \( \Delta(G) \) denotes the maximum degree of \( G \), and \( \nu(G) \) its matching number.

Very recently, Kong, Meel, Maranzatto, Tamo and Ulukus~\cite{kong2025newcapacityboundspir} made significant progress in characterizing the PIR capacity $\mathcal{C}(G)$ for several fundamental graph topologies, including paths, complete bipartite graphs, and complete graphs. Of particular importance is their construction of an efficient PIR scheme over complete graphs $K_N$, which achieves a strictly higher retrieval rate compared to prior approaches in~\cite{Sadeh2023bound,Banawan2019graph}.   
\begin{theorem}[\cite{kong2025newcapacityboundspir}]\label{thm:KongLowerbound}  
For any integer $N \geq 3$, the PIR capacity satisfies  
\[ \mathcal{C}(K_N) \geq \frac{6}{5 - 2^{3-N}} \cdot \frac{1}{N}. \] 
\end{theorem}  
For a general graph \( G \), applying the PIR scheme designed for the complete graph to \( G \) yields a lower bound on the retrieval rate achievable for any arbitrary simple graph. Therefore,~\cref{thm:KongLowerbound} also serves as a lower bound for \( \mathcal{C}(G) \) for any graph \( G \).

Another particularly interesting class of graphs is the family of complete bipartite graphs. Let $K_{M,N}$ denote the complete bipartite graph with two sets of vertices of sizes $M$ and $N$. The best-known bounds on their PIR capacity are given below.
\begin{theorem}[\cite{kong2025newcapacityboundspir,Sadeh2023bound}]\label{thm:BicliquesPrevious}
    For any integers \( M \ge N \), the PIR capacity satisfies
    \begin{equation*}
        \max\left\{\frac{6}{5-2^{3-(M+N)}}\cdot \frac{1}{M+N},\frac{1}{2M\sqrt{N} + M}\right\} \le \mathcal{C}(K_{M,N}) \le \min\left\{\frac{1}{M},\frac{1}{\sqrt{2MN} - \frac{M}{2}}\right\}.
    \end{equation*}
\end{theorem}

Since parts of our results also apply to PIR schemes over multigraphs, we briefly introduce some relevant background and results. Let \( G = (V, E) \) be a simple graph with \( |E| = K' \). We define the \emph{\( r \)-multigraph extension} of \( G \), denoted by \( G^{(r)} \), as the multigraph obtained by replacing each edge in \( G \) with \( r \) parallel edges. Thus, the total number of edges in \( G^{(r)} \) is \( |E(G^{(r)})| = rK' = K \).

PIR schemes over \( r \)-multigraphs can be constructed from schemes over simple graphs, provided the latter satisfy a specific property defined below.

\begin{defn}[Symmetric retrieval property~\cite{kong2025newcapacityboundspir}]\label{def:SRP}
A graph-based PIR scheme is said to satisfy the \emph{symmetric retrieval property} (SRP) if, for every possible file index \( \theta \), the number of bits of \( W_\theta \) retrieved from each of the two servers storing it is equal. For instance, if \( \theta = k \) and \( W_k \) is replicated across servers \( S_1 \) and \( S_2 \), then half of the bits of \( W_k \) must be retrieved from each of \( S_1 \) and \( S_2 \).
\end{defn}

\begin{theorem}[\cite{kong2025newcapacityboundspir}]\label{thm:mutigraph}
Let \( \Pi \) be a PIR scheme over a simple graph \( G \) that satisfies the SRP and achieves rate \( R(\Pi) \). Then there exists a scheme \( \Pi^{(r)} \) for the corresponding \( r \)-multigraph \( G^{(r)} \) that achieves the following lower bound on capacity:
\begin{equation}
    \mathcal{C}(G^{(r)}) \ge R(\Pi) \cdot \left( \frac{1}{2 - \frac{1}{2^{r-1}}} \right). \label{eq:multigrapg}
\end{equation}
\end{theorem}

\subsection{Our main contributions}
The aforementioned results form a foundation for the study of graph-based PIR schemes. Among them, the complete graph \( K_N \) stands out due to its universality and high degree of symmetry. Any PIR scheme designed for \( K_N \) can be naturally adapted to arbitrary simple graphs with \( N \) vertices, making the task of determining tight bounds for \( \mathcal{C}(K_N) \) a central challenge in the field.

In~\cite{Sadeh2023bound}, it was conjectured that the general upper bound \( \frac{2}{N+1} \) might be tight for the complete graph \( K_N \), and that this rate could potentially be achieved via a more sophisticated PIR scheme. In this work, we disprove this conjecture by establishing a strictly improved upper bound valid for all \( N \geq 4 \). Our main result is as follows.

\begin{theorem}\label{thm:UpperBoundMain}
For any integer \( N \ge 3 \),
\[
\mathcal{C}(K_N) \le \frac{1}{\sum_{i=2}^{N} \frac{1}{i!}} \cdot \frac{1}{N}.
\]
\end{theorem}

This new upper bound exactly characterizes \(\mathcal{C}(K_3)\) and strictly improves on the previously conjectured bound \( \frac{2}{N+1} \) for all \( N \ge 4 \). In particular, it asymptotically approaches
\[
\frac{1}{e - 2} \cdot \frac{1}{N} \approx \frac{1.3922}{N},
\]
which is a significant tightening compared to the former \( \frac{2}{N+1} \) bound. This result reveals a new structural limitation of PIR schemes over complete graphs, providing new insights into how symmetry and full connectivity constrain the capacity. Moreover, our techniques extend naturally to yield improved upper bounds for the PIR capacity of some other symmetry graphs, especially complete bipartite graphs.

\begin{theorem}\label{thm:UpperBound CompleteBipartite}
For any even integer \( N \ge 4 \),
\[
\mathcal{C}\left(K_{\frac{N}{2}, \frac{N}{2}}\right) \le \frac{1}{\sum_{i=1}^{N/2} \frac{1}{i! \cdot 2^i}} \cdot \frac{1}{N}.
\]
\end{theorem}

Note that \cref{thm:UpperBound CompleteBipartite} represents a substantial improvement over the previously known bound in~\cref{thm:BicliquesPrevious} for the capacity of \( K_{\frac{N}{2}, \frac{N}{2}} \). Prior work established an upper bound
\[
\mathcal{C}\left(K_{\frac{N}{2}, \frac{N}{2}}\right) \le \frac{2}{N}.
\]
In contrast, our bound is strictly tighter for all even \( N \ge 4 \) since
\[
\sum_{i=1}^{N/2} \frac{1}{i! \cdot 2^i} < e^{1/2} - 1 \approx 0.6487,
\]
implying an upper bound around \(\frac{1.5415}{N}\) when $N$ goes to infinity.

On the other hand, for any \(N \ge 3\), we construct an explicit PIR scheme over \(K_N\). Furthermore, when \(N\) is sufficiently large, this construction yields an improved lower bound on \(\mathcal{C}(K_N)\), as described below.

\begin{theorem}\label{thm:LowerBoundMain}
    For any integer $N\ge 3$, 
\[
\mathcal{C}(K_N) \ge \left( \frac{4}{3}-\epsilon_N \right) \cdot \frac{1}{N},
\]
where $\epsilon_N$ tends to zero as $N$ tends to infinity.
\end{theorem}
One can actually try to show that the coefficient in lower bound is always larger than $\frac{4}{3}$. As an evidence, we can numerically compute the exact values of the lower bound from our construction when $N$ is relatively small. It turns out that these values are already quite close to the known best upper bounds, see~\cref{fig:Compara}. As a direct consequence, we obtain an improved lower bound for any graph $G$.
\begin{cor}\label{cor:anygraph}
    For any graph $G$
    \[
\mathcal{C}(G) \ge \left( \frac{4}{3}-\epsilon_G \right) \cdot \frac{1}{|G|},
\]
where $\epsilon_G$ tends to zero as $|G|$ tends to infinity.
\end{cor}

\begin{table}[ht]
\centering
\begin{tabular}{|c|c|c|c|c|c|c|c|c|c|c|}
\hline
$N$ & $3$ & $4$ & $5$ & $6$ & $7$ & $8$ & $9$ & $10$ \\ \hline
Upper bounds & $0.5$ & $0.35294$ & $0.27907$ & $0.23211$ & $0.19890$ & $0.17403$ & $0.15469$ & $0.13922$ \\ \hline
Lower bounds & $0.5$ & $0.35$ & $0.27541$ & $0.22868$ & $0.19583$ & $0.17111$ & $0.15198$ & $0.13657$ \\ \hline
\end{tabular}
\caption{Upper bound in~\cref{thm:UpperBoundMain} and lower bound in~\cref{thm:LowerBoundMain} when $3\le N\le 10$.}
\label{fig:Compara}
\end{table}

We now establish a connection between deterministic and probabilistic PIR schemes. Specifically, we show that any deterministic PIR scheme satisfying the \emph{independence property} can be transformed into a probabilistic PIR scheme with subpacketization~$1$, without any loss in retrieval rate. The formal definition of the independence property is provided in~\cref{subsection:transform}. This property captures structural features observed in the PIR scheme proposed in~\cite{Sun2017capacity}. The formal statement reads as follows.

\begin{theorem}\label{thm:transform to probabilistic}
    Let \( \Pi \) be a deterministic PIR scheme over a graph \( G \) that satisfies the independence property with subpacketization $L$, achieving retrieval rate \( R(\Pi) \). Assume that for every server $S_i$, the entropy $H(A_i)$ of its answer equals the number of messages in $A_i$. Moreover, assume that for any $i\in [N]$, we have $H(A_i)\le L$. Then, there exists a probabilistic PIR scheme \( \Pi' \) for \( G \) with subpacketization 1 such that \( R(\Pi') = R(\Pi) \).
\end{theorem}

Based on~\cref{thm:transform to probabilistic}, existing PIR schemes from~\cite{Banawan2019graph,kong2025newcapacityboundspir} for various classes of graphs, as well as our constructed PIR scheme over \( K_N \), can be transformed into probabilistic PIR schemes with subpacketization 1. The subpacketization levels of these schemes are summarized in~\cref{table:Subpacketization}. However, the PIR scheme over star graphs proposed in~\cite{Sadeh2023bound} cannot be transformed into a probabilistic scheme, as it fails to satisfy the condition \( H(A_i) \le L \).

\begin{table}[ht]
\centering
\begin{tabular}{|c|c|c|c|c|}
\hline
Graphs       & Path graph $P_N$ & Circle graph $C_{N}$ & Complete graph $K_{N}$ (\cref{thm:LowerBoundMain}) \\ \hline
Subpacketization &   2\cite{kong2025newcapacityboundspir}          &  $4\binom{N}{2}$ \cite{Banawan2019graph}       &       $(N!)^{O(1)}$           \\ \hline
\end{tabular}
\caption{Subpacketization in current graph-based deterministic PIR schemes.}
\label{table:Subpacketization}
\end{table}
In general, we construct a probabilistic PIR scheme that applies to arbitrary (multi-)graphs. The proposed scheme is simple, universally applicable, and requires only subpacketization 1. Compared to existing constructions, it achieves a higher retrieval rate for certain classes of graphs, particularly when the underlying graph is relatively sparse.

\begin{theorem}\label{thm:general}
For any (multi-)graph \( G \) with \( N \) vertices, there exists a probabilistic PIR scheme \( \Pi \) with subpacketization 1 and retrieval rate
\[
R(\Pi) = \frac{1}{\sum_{i=1}^{N} \left(1 - \frac{1}{2^{d(i)}}\right)},
\]
where \( d(i) \) denotes the degree of vertex \( i \) in \( G \).
\end{theorem}

\section{Upper bounds for PIR capacity}\label{section:UpperBound}
\subsection{General framework to upper bound $\mathcal{C}(G)$}
In this section, we establish a general framework for upper bounding \( \mathcal{C}(G) \) in the case of several symmetric graph families. We begin with introducing a fundamental lemma in~\cite{Sadeh2023bound} that formalizes the privacy constraint: the query sent to any individual server must be statistically independent of the identity of the desired file.

\begin{lemma}[\cite{Sadeh2023bound}]\label{lemma:privacy theta}
Let \( G \) be a graph with \( N \) vertices, and let \( \Lambda \subseteq E(G) \) be a subset of its edges. Define \( \mathcal{W}_{\Lambda} := \{ W_e \mid e \in \Lambda \} \). Then, for any PIR scheme over \( G \), for any edge \( \theta \in E(G) \) and any server index \( i \in [N] \), the following holds:
\begin{equation}
    H(A_i \mid \mathcal{W}_{\Lambda}, Q_i, \theta) = H(A_i \mid \mathcal{W}_{\Lambda}, Q_i).
\end{equation}

\end{lemma}

When designing a PIR scheme over a graph, the structure of the graph plays a pivotal role. The following lemma, which will be used throughout the paper, formalizes this dependency. While the core idea was previously discussed in~\cite{Sadeh2023bound}, a complete formal proof was not provided; for completeness, we include it here.

\begin{lemma}\label{lemma:graph contain}
Let \( G_1 \) and \( G_2 \) be two graphs such that \( G_1 \) is a subgraph of \( G_2 \). Then, any PIR scheme over \( G_2 \) can be adapted to construct a valid PIR scheme over \( G_1 \).
\end{lemma}

\begin{proof}[Proof of~\cref{lemma:graph contain}]
For convenience and without loss of generality, we assume that $G_1$ is a spanning subgraph of $G_2$. In fact, if this is not the case, we can extend $V(G_{1})$ by adding isolated vertices to match the larger one, as isolated vertices do not affect the entropy of the server responses. Let us first fix a PIR scheme over \( G_2 \). Let \( S_1, \dots, S_N \) denote the \( N \) vertices of both \( G_1 \) and \( G_2 \). For each edge \( e \in E(G_j) \), let \( W_e^j \) denote the corresponding file in graph \( G_j \), and let \( \WW_{S_i}^j \) denote the set of files stored on server \( S_i \) in graph \( G_j \), for \( j = 1, 2 \). Let \( \QQ^j = \{ Q_i^j \mid i \in [N] \} \) denote the set of all queries sent to the servers in graph \( G_j \), and let \( A_i^j \) be the corresponding answer from server \( S_i \), for \( j = 1, 2 \).  
    
For any \( \theta \in E(G_2) \) chosen by the user, it follows from~\cref{lemma:privacy theta} that each query \( Q_i^2 \) is independent of \( \theta \), which implies
\begin{equation}
    H(\theta \mid Q_i^2, \mathcal{W}_{S_i}^2) = H(\theta). \label{eq:lem:contain privacy}
\end{equation}
Moreover, by the reliability condition~\eqref{eq:reliability} for graph \( G_2 \), we have
\begin{equation}
    H(W_{\theta} \mid A_{[N]}^2, \QQ^2) = 0. \label{eq:lem:contain reliability}
\end{equation}

We now construct a PIR scheme for graph \( G_1 \) by leveraging the existing scheme designed for \( G_2 \). For any \( \theta \in E(G_1) \), the queries in the PIR scheme consist of summations that are linear combinations of sub-files from the files stored at each server. Specifically, we define \( Q_i^1 \) as the query obtained from \( Q_i^2 \) by removing all sub-files belonging to files in \( \mathcal{W}_{S_i}^2 \setminus \mathcal{W}_{S_i}^1 \). In other words, \( Q_i^1 \) is the restriction of \( Q_i^2 \) to the file set \( \mathcal{W}_{S_i}^1 \).
    
   Under this construction, the following identities hold:
\begin{equation}
    H(Q_i^1 \mid Q_i^2) = 0, \label{eq:lem:contain 1}
\end{equation}
\begin{equation}
    H(A_i^2 \mid A_i^1, \mathcal{W}_{S_i}^2 \setminus \mathcal{W}_{S_i}^1, Q_i^2) = 0, \label{eq:lem:contain 2}
\end{equation}
\begin{equation}
    H(A_i^1 \mid A_i^2, \mathcal{W}_{S_i}^2 \setminus \mathcal{W}_{S_i}^1, Q_i^2) = 0. \label{eq:lem:contain 3}
\end{equation}
Equation~\eqref{eq:lem:contain 1} follows from the fact that \( Q_i^1 \) is entirely determined by \( Q_i^2 \). Equations~\eqref{eq:lem:contain 2} and~\eqref{eq:lem:contain 3} follow from the observation that the difference between \( A_i^1 \) and \( A_i^2 \) consists solely of linear combinations of sub-files of files in \( \mathcal{W}_{S_i}^2 \setminus \mathcal{W}_{S_i}^1 \), which are themselves fully determined by \( \mathcal{W}_{S_i}^2 \setminus \mathcal{W}_{S_i}^1 \) and \( Q_i^2 \).

We now verify that the derived PIR scheme over \( G_1 \) satisfies both the reliability and privacy requirements.

\paragraph{Reliability:} For any \( \theta \in E(G_1) \), we analyze the entropy of the desired file conditioned on the available information:
\begin{align}
H(W_{\theta} \mid A_{[N]}^1, \mathcal{Q}^1) 
&= H(W_{\theta} \mid A_{[N]}^1, \mathcal{Q}^1) + H(\mathcal{Q}^2 \mid A_{[N]}^1, \mathcal{Q}^1) - H(\mathcal{Q}^2 \mid A_{[N]}^1, \mathcal{Q}^1) \notag \\
&= H(W_{\theta} \mid A_{[N]}^1, \mathcal{Q}^1) + H(\mathcal{Q}^2 \mid W_{\theta}, A_{[N]}^1, \mathcal{Q}^1) - H(\mathcal{Q}^2 \mid A_{[N]}^1, \mathcal{Q}^1) \label{eq:lem:contain 4} \\
&= H(W_{\theta}, \mathcal{Q}^2 \mid A_{[N]}^1, \mathcal{Q}^1) - H(\mathcal{Q}^2 \mid A_{[N]}^1, \mathcal{Q}^1) \notag \\
&= H(W_{\theta} \mid A_{[N]}^1, \mathcal{Q}^1, \mathcal{Q}^2) \notag \\
&= H(W_{\theta} \mid A_{[N]}^1, \mathcal{W}^2 \setminus \mathcal{W}^1, \mathcal{Q}^2) \label{eq:lem:contain 5} \\
&= H(W_{\theta} \mid A_{[N]}^1, A_{[N]}^2, \mathcal{W}^2 \setminus \mathcal{W}^1, \mathcal{Q}^2) \label{eq:lem:contain 6} \\
&\le H(W_{\theta} \mid A_{[N]}^2, \mathcal{Q}^2), \notag
\end{align}
where~\eqref{eq:lem:contain 4} follows from the fact that the queries are generated independently of the specific file contents and server responses;~\eqref{eq:lem:contain 5} follows from~\eqref{eq:lem:contain 1} and the independence of \( \mathcal{W}^2 \setminus \mathcal{W}^1 \) from \( W_{\theta} \), \( A_{[N]}^1 \), \( \mathcal{Q}^1 \), and \( \mathcal{Q}^2 \); and~\eqref{eq:lem:contain 6} follows from~\eqref{eq:lem:contain 2}. Therefore, applying~\eqref{eq:lem:contain reliability}, we obtain
\begin{equation}
    H(W_{\theta} \mid A_{[N]}^1, \mathcal{Q}^1) \le H(W_{\theta} \mid A_{[N]}^2, \mathcal{Q}^2) = 0. \label{eq:lem:contain 7}
\end{equation}
This confirms that the user can retrieve the desired file in the derived PIR scheme over $G_1$, completing the proof of reliability.

\paragraph{Privacy:} We now show that the queries and stored data reveal no information about the identity of \( \theta \):
\begin{align}
H(\theta) &\ge H(\theta \mid Q_i^1, \mathcal{W}_{S_i}^1) \notag \\
&\ge H(\theta \mid Q_i^1, Q_i^2, \mathcal{W}_{S_i}^2) \notag \\
&= H(\theta \mid Q_i^2, \mathcal{W}_{S_i}^2) \label{eq:lem:contain 8} \\
&= H(\theta), \label{eq:lem:contain 9}
\end{align}
where~\eqref{eq:lem:contain 8} follows from~\eqref{eq:lem:contain 1}, and~\eqref{eq:lem:contain 9} follows from the privacy condition~\eqref{eq:lem:contain privacy}.
Therefore,
\[
H(\theta \mid Q_i^1, \mathcal{W}_{S_i}^1) = H(\theta),
\]
which confirms that the scheme preserves the user's privacy. This finishes the proof.

\end{proof}

We continue to use the notation established in the proof of~\cref{lemma:graph contain}. One of the most important consequences of~\cref{lemma:graph contain} is the following result, which establishes a direct relationship between the capacities of \( G_1 \) and \( G_2 \) when \( G_1 \) is a subgraph of \( G_2 \).

\begin{cor}\label{cor:graph contain}
    Let \( G_1 \subseteq G_2 \) be two graphs and let \( A_i^2 \) denote the response of server \( S_i \) in a PIR scheme over \( G_2 \) with subpacketization \( L \). Then, for any \( \theta \in E(G_1) \), we have
    \begin{equation}
        \sum_{i \in V(G_2)} H\left(A_i^2 \mid \WW^2 \setminus \WW^1, \QQ^2\right) \ge \frac{L}{\CC(G_1)}.
    \end{equation}
    As a consequence, it follows that \( \CC(G_1) \ge \CC(G_2) \).
    
\end{cor}
\begin{proof}[Proof of~\cref{cor:graph contain}]
Without loss of generality, we assume that \( V(G_1) = V(G_2) = S_{[N]} \). For any \( i \in [N] \), let \( A_i^1 \) denote the response of server \( S_i \) in the PIR scheme over \( G_1 \), which is derived from the PIR scheme over \( G_2 \) as described in~\cref{lemma:graph contain}. Let \( \mathcal{Q}^1 \) denote the corresponding set of queries in graph \( G_1 \). Then, by the definition of capacity, we have
\begin{equation}
    \sum_{i \in [N]} H(A_i^1 \mid Q_i^1) \ge \frac{L}{\mathcal{C}(G_1)}.
\end{equation}
Moreover, the following claim plays a crucial role.
\begin{claim} \label{claim:lem:graphcontain}
For any \( i \in [N] \) and any \( \theta \in E(G_1) \), we have
\begin{equation}
    H(A_i^2 \mid \mathcal{W}^2 \setminus \mathcal{W}^1, \mathcal{Q}^2) = H(A_i^1 \mid Q_i^1). \label{eq:cor:contain claim}
\end{equation}
\end{claim}

    \begin{poc}
        Note that
        \begin{align}
            H(A_i^2|\WW^2\setminus \WW^1,\QQ^2)
            &= H(A_i^2|\WW_{S_i}^2\setminus \WW_{S_i}^1,Q_i^2) \notag \\
            &= H(A_i^2|\WW_{S_i}^2\setminus \WW_{S_i}^1,Q_i^2) - H(A_i^2|A_i^1,\WW_{S_i}^2\setminus \WW_{S_i}^1,Q_i^2) \label{eq:cor:contain 2} \\
            &= I(A_i^1;A_i^2|\WW_{S_i}^2\setminus \WW_{S_i}^1,Q_i^2) \notag \\
            &= H(A_i^1|\WW_{S_i}^2\setminus \WW_{S_i}^1,Q_i^2) - H(A_i^1|A_i^2,\WW_{S_i}^2\setminus \WW_{S_i}^1,Q_i^2) \notag \\
            &= H(A_i^1|\WW_{S_i}^2\setminus \WW_{S_i}^1,Q_i^2) \label{eq:cor:contain 4} \\
            &= H(A_i^1|\WW_{S_i}^2\setminus \WW_{S_i}^1,Q_i^1,Q_i^2) \label{eq:cor:contain 3} \\
            &= H(A_i^1|Q_i^1,Q_i^2) \label{eq:cor:contain 5} \\
            &= H(A_i^1|Q_i^1) \label{eq:cor:contain 6} 
        \end{align}
        where~\eqref{eq:cor:contain 2} follows from~\eqref{eq:lem:contain 2};~\eqref{eq:cor:contain 4} follows from~\eqref{eq:lem:contain 3};~\eqref{eq:cor:contain 3} follows from~\eqref{eq:lem:contain 1};~\eqref{eq:cor:contain 5} follows from the fact that \( A_i^1 \) is independent of \( \mathcal{W}_{S_i}^2 \setminus \mathcal{W}_{S_i}^1 \), since these files are not stored on server \( S_i \) in the graph \( G_1 \); and~\eqref{eq:cor:contain 6} follows from the fact that \( A_i^1 \) depends only on \( Q_i^1 \) and \( \mathcal{W}_{S_i}^1 \), so knowledge of \( Q_i^2 \) provides no additional information about \( A_i^1 \).

    \end{poc}
Therefore, by~\cref{claim:lem:graphcontain}, we obtain
\[
\sum_{i\in[N]}H(A_i^2\mid Q_i^2)\ge \sum_{i \in [N]} H(A_i^2 \mid \mathcal{W}^2 \setminus \mathcal{W}^1, \mathcal{Q}^2) = \sum_{i \in [N]} H(A_i^1 \mid Q_i^1) \ge \frac{L}{\mathcal{C}(G_1)}.
\]
Moreover, by the definition of capacity, we have
\begin{equation}
    \frac{L}{\mathcal{C}(G_2)} = \inf \left\{ \sum_{i \in [N]} H(A_i^2 \mid Q_i^2) \right\} \ge \frac{L}{\mathcal{C}(G_1)},
\end{equation}
which implies that \( \mathcal{C}(G_1) \ge \mathcal{C}(G_2) \). This finishes the proof of~\cref{cor:graph contain}.

\end{proof}

With the above preparations in place, we now develop a framework for estimating the information entropy of the server responses. For convenience, for any \( 1 \le k \le N - 2 \), we set \( \mathcal{W}_{[0]} := \varnothing \) and define
\[
\mathcal{W}_{[k]} := \bigcup_{\ell = 1}^{k} \{ W_e \mid \ell \in e \} = \bigcup_{\ell = 1}^{k} \mathcal{W}_{S_{\ell}}.
\]

\begin{lemma}\label{lemma:splitting to small graph}
   For any PIR scheme over an arbitrary graph \( G \), the following inequality holds:
\begin{equation} \label{eq:splitting to small graph}
    \sum_{i = 1}^{N} H(A_i \mid \mathcal{Q}) \ge L + H(A_1 \mid \mathcal{W}_{S_N}, \mathcal{Q}) + \sum_{i = 2}^{N - 1} H(A_i \mid \mathcal{W}_{[i - 1]}, \mathcal{Q}).
\end{equation}

\end{lemma}
\begin{proof}[Proof of~\cref{lemma:splitting to small graph}]
First, we observe that 
\begin{align}
    L &= H(W_{1,N}) \notag \\
    &= H(W_{1,N}|\QQ,\theta=(1,N)) \label{eq:thm:upperbound 1} \\
    &= H(W_{1,N}|\QQ,\theta=(1,N)) - H(W_{1,N}|A_{[N]},\QQ,\theta=(1,N)) \label{eq:thm:upperbound 2} \\
    &= I(W_{1,N};A_{[N]}|\QQ,\theta=(1,N)) \notag \\
    &= H(A_{[N]}|\QQ,\theta=(1,N)) - H(A_{[N]}|W_{1,N},\QQ,\theta=(1,N)) \notag \\
    &\le \sum_{i=1}^{N}H(A_i|\QQ,\theta=(1,N)) - H(A_{[N]}|W_{1,N},\QQ,\theta=(1,N)) \notag \\
    &= \sum_{i=1}^{N}H(A_i|\QQ) - H(A_{[N]}|W_{1,N},\QQ,\theta=(1,N)), \label{eq:thm:upperbound 4}
\end{align}
where \eqref{eq:thm:upperbound 1} follows since the queries $\QQ$ and the identity of the desired file are independent of the specific file content; \eqref{eq:thm:upperbound 2} follows from the reliability requirement of the PIR scheme; and \eqref{eq:thm:upperbound 4} follows from~\cref{lemma:privacy theta}. 

Rearranging the above inequality yields
\begin{align}
    \sum_{i=1}^{N}H(A_i|\QQ) &\ge L + H(A_{[N]}|W_{1,N},\QQ,\theta=(1,N)) \notag \\
    &= L + \sum_{i=1}^{N}H(A_i|A_{[i-1]},W_{1,N},\QQ,\theta=(1,N)) \label{eq:thm:upperbound 5} \\
    &= L + H(A_1|W_{1,N},\QQ,\theta=(1,N)) + \sum_{i=2}^{N}H(A_i|A_{[i-1]},W_{1,N},\QQ,\theta=(1,N)) \notag \\
    &\ge L + H(A_1|\WW_{S_N},\QQ,\theta=(1,N)) + \sum_{i=2}^{N}H(A_i|A_{[i-1]},\WW_{[i-1]},\QQ,\theta=(1,N)) \label{eq:thm:upperbound not tight} \\
    &= L + H(A_1|\WW_{S_N},\QQ,\theta=(1,N)) + \sum_{i=2}^{N}H(A_i|\WW_{[i-1]},\QQ,\theta=(1,N)) \label{eq:thm:upperbound 6} \\
    &\ge L + H(A_1|\WW_{S_N},\QQ,\theta=(1,N)) + \sum_{i=2}^{N-1}H(A_i|\WW_{[i-1]},\QQ,\theta=(1,N)) \notag \\
    &= L + H(A_1|\WW_{S_N},\QQ) + \sum_{i=2}^{N-1}H(A_i|\WW_{[i-1]},\QQ), \label{eq:thm:upperbound 7}
\end{align}
where \eqref{eq:thm:upperbound 5} follows from the chain rule of joint entropy function; \eqref{eq:thm:upperbound 6} follows from $A_{[i-1]}$ can be totally determined by $\WW_{[i-1]}$ and $\QQ$ for any $i\ge 2$; and \eqref{eq:thm:upperbound 7} follows from~\cref{lemma:privacy theta}. This finishes the proof.
\end{proof}

\subsection{Complete graphs: Proof of~\cref{thm:UpperBoundMain}}\label{section:upperbound}
In this section, we derive an upper bound on the PIR capacity of the complete graph \( K_N \), leveraging~\cref{lemma:splitting to small graph}. To that end, we begin by recalling a key structural property of \( K_N \) that will play an essential role in our analysis. A graph \( G \) is said to be \emph{vertex-transitive} if its automorphism group \( \operatorname{Aut}(G) \) acts transitively on the vertex set. That is, for any pair of vertices \( u, v \in V(G) \), there exists an automorphism \( f \in \operatorname{Aut}(G) \) such that \( f(v) = u \). Before presenting the formal proof of~\cref{thm:UpperBoundMain}, we recall the following two important lemmas from~\cite{Sadeh2023bound}.
\begin{lemma}[Lemma 3,~\cite{Sadeh2023bound}]\label{lem:A_i+A_j le L}
    Let $i, j \in[N]$ be two distinct servers that share the file $W_{i,j}\in \WW$, then
    \[
    L\le H(A_i|\WW\setminus \{W_{i,j}\},\QQ) + H(A_j|\WW\setminus \{W_{i,j}\},\QQ).
    \]
\end{lemma}

\begin{lemma}[Lemma 11,~\cite{Sadeh2023bound}]\label{lem:symmetric origin}
In a vertex-transitive graph-based replication system, for any achievable rate \( R \), and for any \( i \ne j \in [N] \), there exists a PIR scheme with rate \( R \) such that
\[
H(A_i \mid \QQ) = H(A_j \mid \QQ),
\]
and
\[
H(A_i \mid \WW \setminus \{W_{i,j}\}, \QQ) \ge \frac{L}{2}.
\]
\end{lemma}

Since the complete graph \( K_N \) is vertex-transitive, the conclusion of~\cref{lem:symmetric origin} applies to PIR schemes constructed over \( K_N \). Furthermore, owing to the high degree of symmetry in \( K_N \), we establish the following stronger result tailored specifically to the complete graph. The proof of~\cref{lem:symmetric complete} follows a similar line of reasoning as that of~\cref{lem:symmetric origin}. The proof of the following lemma is deferred to~\cref{section:proof of symmetric of complete graph}.

\begin{lemma}\label{lem:symmetric complete}
Let $N\ge 3$ be an integer. For any achievable rate \( R \) and any \( 1 \le k \le N - 2 \), there exists a PIR scheme over $K_N$ with rate $R$ such that
\begin{equation}
    H(A_i \mid \WW_{[k]}, \QQ) = H(A_j \mid \WW_{[k]}, \QQ), \label{eq:lem:symmetric complete 1}
\end{equation}
for any \( k<i \ne j \le N \); and
\begin{equation}
    H(A_i \mid \WW_{S_N}, \QQ) = H(A_j \mid \WW_{S_N}, \QQ), \label{eq:lem:symmetric complete 2}
\end{equation}
for any \( i \ne j < N \).
\end{lemma}

With the above preparations in place, we now formally prove~\cref{thm:UpperBoundMain}. The proof proceeds by induction on~\( N \). For the base case \( N = 3 \), it was shown in~\cite{Banawan2019graph} that \( \CC(K_3) = \frac{1}{2} \). For the induction step, assume that for every \( n \le N-1 \), we have 
\[
\CC(K_n) \le \frac{1}{\left( \sum_{i=2}^{n} \frac{1}{i!} \right) \cdot n}.
\]
Now, suppose \( R \) is an achievable rate for \( K_N \). Then there exists a PIR scheme \( \Pi_N \) with rate~\( R \) that satisfies the equations in~\cref{lem:symmetric complete}.

For any \( 2 \le i \le N-2 \), let \( G_i = K_{N - i + 1} \) be the complete graph whose vertex set consists of \( S_i, S_{i+1}, \dots, S_N \). Then, the set of files corresponding to the edges in \( E(K_N) \setminus E(G_i) \) is precisely the collection \( \WW_{[i-1]} \). Therefore, for any \( 2 \le i \le N-2 \), we have
\begin{align}
    (N - i + 1) \cdot H(A_i \mid \WW_{[i-1]}, \QQ) 
    &= \sum_{j=i}^{N} H(A_j \mid \WW_{[i-1]}, \QQ) \label{eq:thm:upperbound 8} \\
    &\ge \frac{L}{\CC(G_i)} \label{eq:thm:upperbound 9} \\
    &= \frac{L}{\CC(K_{N - i + 1})} \notag \\
    &\ge \left( \sum_{j=2}^{N - i + 1} \frac{1}{j!} \right) \cdot (N - i + 1) \cdot L, \label{eq:thm:upperbound 10}
\end{align}
where~\eqref{eq:thm:upperbound 8} follows from~\eqref{eq:lem:symmetric complete 1}; \eqref{eq:thm:upperbound 9} follows from~\cref{cor:graph contain}; and~\eqref{eq:thm:upperbound 10} follows from the inductive hypothesis. Thus, for \( 2 \le i \le N - 2 \), we have
\begin{equation}
    H(A_i \mid \WW_{[i-1]}, \QQ) 
    \ge \left( \sum_{j=2}^{N - i + 1} \frac{1}{j!} \right) \cdot L. 
    \label{eq:thm:upperbound 11}
\end{equation}
Similarly, we obtain
\begin{equation}
    H(A_1 \mid \WW_{S_N}, \QQ) 
    \ge \left( \sum_{j=2}^{N - 1} \frac{1}{j!} \right) \cdot L. 
    \label{eq:thm:upperbound 12}
\end{equation}
Combining \eqref{eq:lem:symmetric complete 1} for $k=N-2$ and~\cref{lem:A_i+A_j le L}, we have
\begin{equation}
    H(A_{N-1}|\WW_{[N-2]},\QQ) = H(A_{N-1}|\WW \setminus \{W_{N-1,N}\},\QQ)\ge \frac{L}{2}. \label{eq:thm:upperbound 13}
\end{equation}
Combining~\eqref{eq:thm:upperbound 7}, \eqref{eq:thm:upperbound 11}, \eqref{eq:thm:upperbound 12}, and~\eqref{eq:thm:upperbound 13}, we obtain
\begin{align}
    \sum_{i=1}^{N} H(A_i \mid \QQ) 
    &\ge L + H(A_1 \mid \WW_{S_N}, \QQ) 
    + \sum_{i=2}^{N-2} H(A_i \mid \WW_{[i-1]}, \QQ) 
    + H(A_{N-1} \mid \WW_{[N-2]}, \QQ) \notag \\ 
    &\ge L + \left( \sum_{j=2}^{N-1} \frac{1}{j!} \right) \cdot L 
    + \sum_{i=2}^{N-2} \left( \sum_{j=2}^{N-i+1} \frac{1}{j!} \right) \cdot L 
    + \frac{L}{2} \notag \\
    &= \left( 1 + \sum_{j=2}^{N-1} \frac{N - j + 1}{j!} \right) \cdot L \label{eq:thm:upperbound 14} \\
    &= \left( 1 + \left( \sum_{j=2}^{N-1} \frac{1}{j!} \right) \cdot (N + 1) 
    - \sum_{j=2}^{N-2} \frac{1}{(j - 1)!} \right) \cdot L \notag \\
    &= \left( \left( \sum_{j=2}^{N-1} \frac{1}{j!} \right) \cdot N 
    + \frac{1}{(N - 1)!} \right) \cdot L \notag \\
    &= \left( \sum_{j=2}^{N} \frac{1}{j!} \right) \cdot N L, \notag
\end{align}
where~\eqref{eq:thm:upperbound 14} follows by counting how many times each term \( \frac{1}{j!} \) appears in the summation. 
Thus, the rate of the scheme \( \Pi_N \) satisfies
\[
R(\Pi_N) = \frac{L}{\sum_{i=1}^{N} H(A_i \mid \QQ)} 
\le \frac{1}{\left( \sum_{i=2}^{N} \frac{1}{i!} \right) \cdot N},
\]
which completes the proof of~\cref{thm:UpperBoundMain}.

\subsection{Complete bipartite graphs: Proof of~\cref{thm:UpperBound CompleteBipartite}}
In this section, we derive an upper bound on the PIR capacity of the complete bipartite graph \( K_{\frac{N}{2},\frac{N}{2}} \) with parts \( U \) and \( V \), where \( |U| = |V| = \frac{N}{2} \). We consider a PIR scheme operating over this bipartition, with servers \( S_1, S_3, \dots, S_{N-1} \) assigned to \( U \), and servers \( S_2, S_4, \dots, S_N \) assigned to \( V \).

The following lemma relies on the symmetry of the complete bipartite graph \( K_{\frac{N}{2},\frac{N}{2}} \). We omit its proof, as it closely parallels the argument in~\cref{lem:symmetric complete}.
\begin{lemma}\label{lemma:symmetric bipartite}
Let \(N\ge 4\) be an even number. For any achievable rate \( R \) and any \( 1 \le k \le \frac{N}{2} - 1 \), there exists a PIR scheme over \( K_{\frac{N}{2},\frac{N}{2}} \) with rate \( R \) such that the following hold:
\begin{equation}
    H(A_i \mid \WW_{[2k]}, \QQ) = H(A_j \mid \WW_{[2k]}, \QQ), \label{eq:lem:symmetric bipartite 1}
\end{equation}
for any \( 2k < i \ne j \le N \);
\begin{equation}
    H(A_i \mid \WW_{[2k+1]}, \WW_{S_N}, \QQ) = H(A_j \mid \WW_{[2k+1]}, \WW_{S_N}, \QQ), \label{eq:lem:symmetric bipartite 2}
\end{equation}
for any \( 2k+1 < i \ne j \le N \); and
\begin{equation}
    H(A_i \mid \WW_{S_{N-1}}, \WW_{S_N}, \QQ) = H(A_j \mid \WW_{S_{N-1}}, \WW_{S_N}, \QQ), \label{eq:lem:symmetric bipartite 3}
\end{equation}
for any \( 1 \le i \ne j < N - 1 \).
\end{lemma}
    We now proceed to prove~\cref{thm:UpperBound CompleteBipartite} by induction on \( N \). For the base case \( N = 4 \), note that \( K_{2,2} \) is isomorphic to the 4-cycle, and it has been known in~\cite{Sadeh2023bound} that \( \mathcal{C}(K_{2,2}) = \mathcal{C}(C_4) = \frac{2}{5} \), which satisfies the claimed upper bound. Now assume as the induction hypothesis that for all even integers \( n \le N - 2 \), the following holds:
\[
\mathcal{C}\left(K_{\frac{n}{2}, \frac{n}{2}}\right) \le \frac{1}{\sum_{i=1}^{n/2} \frac{1}{i! \cdot 2^i}} \cdot \frac{1}{n}.
\]
We aim to prove the bound for \( \mathcal{C}(K_{\frac{N}{2}, \frac{N}{2}}) \). Suppose \( R \) is an achievable rate for this graph. Then, there exists a PIR scheme \( \Pi_N \) of rate \( R \) that satisfies the symmetric conditions stated in~\cref{lemma:symmetric bipartite}.

For any \( 1 \le k \le \frac{N}{2} - 1 \), define the subgraph \( G_k = K_{\frac{N}{2} - k, \frac{N}{2} - k} \), whose vertex set consists of the servers \( S_{2k+1}, S_{2k+2}, \dots, S_N \). Observe that the set of files corresponding to the edge set \( E(K_{\frac{N}{2}, \frac{N}{2}}) \setminus E(G_k) \) is precisely \( \WW_{[2k]} \). Therefore, for any \( 1 \le k \le \frac{N}{2} - 1 \), we obtain:
\begin{align}
    (N - 2k) \cdot H(A_{2k+1} \mid \WW_{[2k]}, \QQ) 
    &= \sum_{j=2k+1}^{N} H(A_j \mid \WW_{[2k]}, \QQ) \label{eq:thm:UpperBound CompleteBipartite 1} \\
    &\ge \frac{L}{\mathcal{C}(G_k)} \label{eq:thm:UpperBound CompleteBipartite 2} \\
    &= \frac{L}{\mathcal{C}(K_{\frac{N}{2} - k, \frac{N}{2} - k})} \notag \\
    &\ge \left( \sum_{i=1}^{\frac{N}{2} - k} \frac{1}{i! \cdot 2^i} \right) \cdot (N - 2k) \cdot L, \label{eq:thm:UpperBound CompleteBipartite 3}
\end{align}
where~\eqref{eq:thm:UpperBound CompleteBipartite 1} follows from the symmetry condition~\eqref{eq:lem:symmetric bipartite 1};~\eqref{eq:thm:UpperBound CompleteBipartite 2} follows from~\cref{cor:graph contain}; and~\eqref{eq:thm:UpperBound CompleteBipartite 3} follows from the induction hypothesis. Thus, for any \( 1 \le k \le \frac{N}{2} - 2 \), we obtain
\begin{equation}
    H(A_{2k+1} \mid \WW_{[2k]}, \QQ) \ge \left( \sum_{i=1}^{\frac{N}{2} - k} \frac{1}{i! \cdot 2^i} \right) \cdot L. \label{eq:thm:UpperBound CompleteBipartite 4}
\end{equation}
Similarly, for any \( 1 \le k \le \frac{N}{2} - 2 \), we have
\begin{equation}
    H(A_{2k} \mid \WW_{[2k-1]}, \WW_{S_N}, \QQ) \ge \left( \sum_{i=1}^{\frac{N}{2} - k} \frac{1}{i! \cdot 2^i} \right) \cdot L, \label{eq:thm:UpperBound CompleteBipartite 5}
\end{equation}
and
\begin{equation}
    H(A_1 \mid \WW_{S_{N-1}}, \WW_{S_N}, \QQ) \ge \left( \sum_{i=1}^{\frac{N}{2} - 1} \frac{1}{i! \cdot 2^i} \right) \cdot L. \label{eq:thm:UpperBound CompleteBipartite 6}
\end{equation}
Combining~\eqref{eq:lem:symmetric bipartite 1} with \( k = \frac{N}{2} - 1 \) and applying~\cref{lem:A_i+A_j le L}, we obtain
\begin{equation}
    H(A_{N-1} \mid \WW_{[N-2]}, \QQ) = H(A_{N-1} \mid \WW \setminus \{ \WW_{N-1,N} \}, \QQ) \ge \frac{L}{2}. \label{eq:thm:UpperBound CompleteBipartite 7}
\end{equation}
Similarly, combining~\eqref{eq:lem:symmetric bipartite 2} with \( k = \frac{N}{2} - 1 \) and using~\cref{lem:A_i+A_j le L}, we have
\begin{equation}
    H(A_{N-2} \mid \WW_{[N-3]}, \WW_{S_N}, \QQ) = H(A_{N-2} \mid \WW \setminus \{ \WW_{N-2,N-1} \}, \QQ) \ge \frac{L}{2}. \label{eq:thm:UpperBound CompleteBipartite 8}
\end{equation}
According to~\cref{lemma:splitting to small graph}, we have
\begin{align}
    \sum_{i=1}^N H(A_i \mid \QQ)
    &\ge L + H(A_1 \mid \WW_{S_N}, \QQ) + \sum_{i=2}^{N-1} H(A_i \mid \WW_{[i-1]}, \QQ) \notag \\
    &= L + H(A_1 \mid \WW_{S_N}, \QQ) 
    + \sum_{k=1}^{\frac{N}{2}-2} H(A_{2k} \mid \WW_{[2k-1]}, \QQ) 
    + \sum_{k=1}^{\frac{N}{2}-2} H(A_{2k+1} \mid \WW_{[2k]}, \QQ) \notag \\
    &\quad + H(A_{N-2} \mid \WW_{[N-3]}, \QQ) + H(A_{N-1} \mid \WW_{[N-2]}, \QQ) \notag \\
    &= L + H(A_1 \mid \WW_{S_N}, \QQ) 
    + \sum_{k=1}^{\frac{N}{2}-2} H(A_{2k} \mid \WW_{[2k-1]}, \WW_{S_N}, \QQ) 
    + \sum_{k=1}^{\frac{N}{2}-2} H(A_{2k+1} \mid \WW_{[2k]}, \QQ) \notag \\
    &\quad + H(A_{N-2} \mid \WW_{[N-3]}, \WW_{S_N}, \QQ) 
    + H(A_{N-1} \mid \WW_{[N-2]}, \QQ) \label{eq:thm:UpperBound CompleteBipartite 9} \\
    &\ge L + \left( \sum_{i=1}^{\frac{N}{2}-1} \frac{1}{i!2^i} \right) L 
    + \sum_{k=1}^{\frac{N}{2}-2} \left( \sum_{i=1}^{\frac{N}{2}-k} \frac{1}{i!2^i} \right) L 
    + \sum_{k=1}^{\frac{N}{2}-2} \left( \sum_{i=1}^{\frac{N}{2}-k} \frac{1}{i!2^i} \right) L 
    + \frac{L}{2} + \frac{L}{2} \label{eq:thm:UpperBound CompleteBipartite 10} \\
    &= \left( 1 + \sum_{i=1}^{\frac{N}{2}-1} \frac{1}{i!2^i} 
    + 2 \sum_{k=1}^{\frac{N}{2}-2} \sum_{i=1}^{\frac{N}{2}-k} \frac{1}{i!2^i} \right) L \notag \\
    &= \left( 1 + \sum_{i=1}^{\frac{N}{2}-1} \frac{N - 2i + 1}{i!2^i} \right) L \label{eq:thm:UpperBound CompleteBipartite 11} \\
    &= \left( (N+1) \sum_{i=1}^{\frac{N}{2}-1} \frac{1}{i!2^i} 
    - \sum_{i=1}^{\frac{N}{2}-1} \frac{1}{(i-1)!2^{i-1}} + 1 \right) L \notag \\
    &= \left( N \sum_{i=1}^{\frac{N}{2}-1} \frac{1}{i!2^i} 
    + \frac{1}{\left(\frac{N}{2}-1\right)! 2^{\frac{N}{2}-1}} \right) L \notag \\
    &= \left( \sum_{i=1}^{\frac{N}{2}} \frac{1}{i!2^i} \right) \cdot N L. \notag
\end{align}
Here,~\eqref{eq:thm:UpperBound CompleteBipartite 9} uses the fact that server \( S_{2k} \) and server \( S_N \) share no common files for any \( 1 \le k \le \frac{N}{2}-1 \);~\eqref{eq:thm:UpperBound CompleteBipartite 10} follows from~\eqref{eq:thm:UpperBound CompleteBipartite 4},~\eqref{eq:thm:UpperBound CompleteBipartite 5},~\eqref{eq:thm:UpperBound CompleteBipartite 6},~\eqref{eq:thm:UpperBound CompleteBipartite 7}, and~\eqref{eq:thm:UpperBound CompleteBipartite 8}; and~\eqref{eq:thm:UpperBound CompleteBipartite 11} follows by counting how many times each \( \frac{1}{i!2^i} \) appears in the summation. Therefore, the rate \( R(\Pi_N) = \frac{L}{\sum_{i=1}^N H(A_i \mid \QQ)} \) satisfies
\[
R(\Pi_N) \le \frac{1}{\sum_{i=1}^{N/2} \frac{1}{i!2^i}} \cdot \frac{1}{N},
\]
which completes the proof of~\cref{thm:UpperBound CompleteBipartite}.

\section{New construction of PIR scheme over $K_{N}$ and improved lower bound for $\mathcal{C}(K_{N})$}\label{section:lowerbound}

In this section, we present explicit constructions of PIR schemes over the complete graph \( K_N \). To build intuition, we proceed incrementally: we begin by revisiting the classical PIR scheme over \( K_3 \), introduce our new construction over \( K_4 \) as a warm-up, and then extend the construction to a general scheme over \( K_N \) for arbitrary \( N \).

For simplicity, we always assume that each file is a binary vector. More generally, if the files are vectors over an additive abelian group (such like \( \mathbb{Z}_q \)), our construction remains applicable by incorporating random signs before each sub-file in the queries. A similar technique is also employed in~\cref{subsection:Probabilistic General Graph}.

\subsection{PIR scheme over triangle: Revisited and a new framework}\label{subsection:PIR scheme k3 revisited}

In~\cite{Banawan2019graph}, Banawan and Ulukus presented an explicit capacity-achieving PIR scheme over the cycle \( C_N \) for any integer \( N \ge 3 \). Since \( K_3 = C_3 \), this result directly yields an optimal scheme over the complete graph $K_3$. More recently, Kong, Meel, Maranzatto, Tamo, and Ulukus~\cite{kong2025newcapacityboundspir} introduced a simplified variant specifically tailored to \( K_3 \). In this section, we reinterpret the optimal scheme over \( K_3 \) from a novel and arguably more natural perspective.

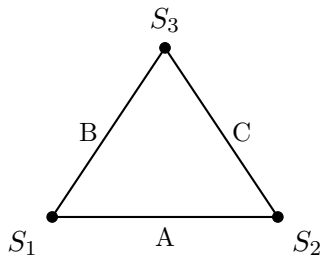
\begin{figure}[h]
\centering
\begin{tikzpicture}[
    vertex/.style={circle, draw, fill=black, inner sep=0pt, minimum size=4pt},
    edge/.style={midway, font=\small}
]

% 定义顶点坐标
\coordinate (S1) at (0,0);
\coordinate (S2) at (3,0);
\coordinate (S3) at (1.5,2.25);

% 绘制边并标注
\draw[thick] (S1) -- node[edge, below] {A} (S2);
\draw[thick] (S1) -- node[edge, left] {B} (S3);
\draw[thick] (S2) -- node[edge, right] {C} (S3);

% 标注顶点（使用数学模式 $S_1$, $S_2$, $S_3$）
\node[vertex, label=below left:$S_1$] at (S1) {};
\node[vertex, label=below right:$S_2$] at (S2) {};
\node[vertex, label=above:$S_3$] at (S3) {};
\end{tikzpicture}
\caption{The corresponding PIR system over $K_3$} % 添加标题
\label{fig:K_3} % 添加标签
\end{figure}

The replication system over \( K_3 \) consists of three servers \( S_1, S_2, S_3 \), and three files \( A, B, C \), corresponding to the edges \( S_1S_2 \), \( S_1S_3 \), and \( S_2S_3 \), respectively; see~\cref{fig:K_3}. Without loss of generality, we assume the user wishes to retrieve file \( A \), with demand index \( \theta = A \). We begin by presenting an illustrative PIR scheme over \( K_3 \) that achieves the optimal retrieval rate of \( \frac{1}{2} \).

In this scheme, each file is divided into \( L = 6 \) equal-sized sub-files. As described in~\cref{section:problemstatement}, the user privately generates three independent random permutations \( \sigma_1, \sigma_2, \sigma_3 \) of \([6]\). For each \( i \in [6] \), define
\[
a_i := A(\sigma_1(i)), \quad b_i := B(\sigma_2(i)), \quad c_i := C(\sigma_3(i)),
\]
where \( a_i \) denotes the \( \sigma_1(i) \)-th sub-file of \( A \), and \( b_i, c_i \) are defined analogously. The user then requests and downloads four messages from each server \( S_1 \), \( S_2 \), and \( S_3 \), as summarized in~\cref{table:PIR scheme over k3}, which lists the exact sub-files requested from each server. Upon receiving the queries, the servers respond with the corresponding answers. For instance, from server \( S_1 \), the user downloads the sub-files \( a_1 \), \( b_3 \), \( a_3 + b_1 \), and \( a_5 + b_2 \).

\begin{table}[ht]
\centering
\begin{tabular}{|c|c|c|}
\hline
$S_1$ & $S_2$ & $S_3$ \\ \hline
$a_1$ & $a_2$ & $b_1$ \\ 
$b_3$ & $c_2$ & $c_1$ \\ 
$a_3+b_1$ & $a_4+c_1$ & $b_2+c_2$ \\ 
$a_5+b_2$ & $a_6+c_3$ & $b_3+c_3$ \\ \hline
\end{tabular}
\caption{Answer table of the PIR scheme over $K_3$ in~\cite{kong2025newcapacityboundspir}, for $\theta =A$. }
\label{table:PIR scheme over k3}
\end{table}
Based on the responses, the user can recover all six sub-files of \( A \) as follows:
\begin{itemize}
    \item \( a_1 \) is retrieved directly from \( S_1 \); \( a_2 \) is retrieved directly from \( S_2 \);
    \item \( a_3 \) is recovered by adding \( a_3 + b_1 \) (from \( S_1 \)) and \( b_1 \) (from \( S_3 \)); similarly, \( a_4 \) is recovered by adding \( a_4 + c_1 \) (from \( S_2 \)) and \( c_1 \) (from \( S_3 \));
    \item \( a_5 \) is recovered by adding \( a_5 + b_2 \) (from \( S_1 \)), \( c_2 \) (from \( S_2 \)) and \( b_2 + c_2 \) (from \( S_3 \));
    \item \( a_6 \) is recovered by adding \( b_3 \) (from \( S_1 \)), \( a_6+c_3 \) (from \( S_2 \)) and \( b_3 + c_3 \) (from \( S_3 \)).
\end{itemize}

Since the bit permutations are privately chosen and unknown to the servers, the queries appear uniformly random and independent of the desired file index \( \theta \). This guarantees the privacy of the retrieval. The scheme downloads a total of \( 4 \times 3 = 12 \) sub-files to recover 6 desired sub-files, resulting in a PIR rate of \( \frac{6}{12} = \frac{1}{2} \).

To gain a deeper understanding of the internal structure of PIR schemes over \( K_N \) for larger \( N \), it is helpful to categorize the types of summations that appear in the servers' responses. Each such summation involves sub-files stored on a given server and can be classified based on the number of distinct files it includes. This motivates the following definition.

\begin{defn}[\( k \)-sum]
    Let \( G \) be a graph, and let \( \Pi \) be a PIR scheme defined over \( G \). Fix a desired file index \( \theta \in E(G) \). For any server \( S_i \) with \( i \in [N] \) and any positive integer \( k \), suppose \( W_{i,j_1}, \dots, W_{i,j_k} \) are \( k \) distinct files stored on \( S_i \). A summation that consists of exactly one sub-file from each of \( W_{i,j_1}, \dots, W_{i,j_k} \) is called a \emph{\( k \)-sum} of type \( W_{i,j_1} + \dots + W_{i,j_k} \). Two such \( k \)-sums are said to be of the same type if they involve the same \( k \) files. The \emph{\( k \)-block} of the answer \( A_i \) refers to the collection of all \( k \)-sums of various types that appear in the answer returned by server \( S_i \).
\end{defn}

As an example, in the PIR scheme described in~\cref{table:PIR scheme over k3}, the sub-file \( a_2 \) from server \( S_2 \) is a 1-sum of type \( A \); the summation \( b_2 + c_2 \) from server \( S_3 \) is a 2-sum of type \( B + C \); the sub-files \( a_1 \) and \( b_3 \) from server \( S_1 \) are 1-sums of types \( A \) and \( B \), respectively; and the summations \( a_3 + b_1 \) and \( a_5 + b_2 \) from \( S_1 \) are 2-sums of type \( A + B \).

By definition, for any server \( S_i \) with degree \( \deg(S_i) \), the number of possible distinct types of \( k \)-sums is exactly \( \binom{\deg(S_i)}{k} \). In~\cref{table:PIR scheme over k3}, we observe that for \( k = 1 \) and \( k = 2 \), the number of \( k \)-sums appearing in the answer from each server is identical across all servers. Moreover, the answer \( A_i \) from each server can be decomposed into the union of its \( k \)-blocks for \( k = 1, 2 \). To further quantify the frequency of each distinct \( k \)-sum, we introduce the following concept.

\begin{defn}[Multiplicity of \( k \)-sums]\label{def:Multiplicity}
    Let \( \Pi \) be a PIR scheme over a graph \( G \), and fix a desired file index \( \theta \in E(G) \). For any positive integer \( k \), we say that the \( k \)-sums in the scheme have \emph{multiplicity} \( x_k \) if the following holds: for every server \( S_i \) with \( i \in [N] \), and for every choice of \( k \) distinct files \( W_{i,j_1}, \dots, W_{i,j_k} \) stored on \( S_i \), there exist exactly \( x_k \) distinct \( k \)-sums of type \( W_{i,j_1} + \dots + W_{i,j_k} \) in the answer \( A_i \).
\end{defn}

By applying~\cref{def:Multiplicity} to the scheme in~\cref{table:PIR scheme over k3}, one can verify that the \( 1 \)-sums have multiplicity \( 1 \), while the \( 2 \)-sums have multiplicity \( 2 \). Observe that by appropriately combining selected summations from the answers of different servers, one can recover individual sub-files of the desired file. This observation motivates the following definition.
\begin{defn}[Recovery Pattern]\label{def:Recovery Pattern}
    Let \( G \) be a graph, and let \( \Pi \) be a PIR scheme over \( G \). Fix a desired file index \( \theta \in E(G) \). A collection consisting of at most one requested summation from the query of each server is called a \emph{recovery pattern} for \( W_{\theta} \) if the sum of these selected summations equals a sub-file of \( W_{\theta} \). 
    
\end{defn}

For example, in~\cref{table:PIR scheme over k3}, the summation \( a_3 + b_1 \) from \( S_1 \), together with \( b_1 \) from \( S_3 \), constitutes a recovery pattern for \( A \), as their sum yields \( a_3 \).

To simplify notation and enhance readability, we omit sub-file indices when listing or analyzing recovery patterns. Instead, we use underlined symbols to indicate that a summand belongs to a particular file, without specifying its exact sub-file. For instance, a \( 2 \)-sum of type \( A + B \) is denoted by \( \underline{a} + \underline{b} \), where each underlined term represents some (possibly distinct) sub-file from the corresponding file. This abstraction allows us to concentrate on the structural aspects of recovery patterns without being distracted by indexing details. Throughout the rest of this section, we will consistently adopt this simplified notation in both our constructions and discussions.

Using this convention, we summarize all recovery patterns employed in the PIR scheme over \( K_3 \) in~\cref{table:recovery patterns in k3}. To avoid redundancy due to symmetry, equivalent patterns are grouped into equivalence classes. Each summation in~\cref{table:PIR scheme over k3} can then be classified according to the recovery pattern it participated in. Based on this classification, we present an adjusted version of the PIR scheme over \( K_3 \) in~\cref{table:adjusted scheme over k3}, which explicitly highlights the structure and functional role of each summation. Notably, in~\cref{table:adjusted scheme over k3}, each recovery pattern listed in~\cref{table:recovery patterns in k3} is used exactly once.

\begin{table}[htbp]
\centering
\begin{tabular}{|c|c|c|c|}
\hline
                  & $S_1$ & $S_2$ & $S_3$ \\ \hline
\multirow{2}{*}{recovery pattern 1} & $\underline{a}$ &  &  \\ \cline{2-4} 
                  &  & $\underline{a}$ &  \\ \hline
\multirow{2}{*}{recovery pattern 2} & $\underline{a}+\underline{b}$ &  & $\underline{b}$ \\ \cline{2-4} 
                  &  & $\underline{a}+\underline{c}$ & $\underline{c}$ \\ \hline
\multirow{2}{*}{recovery pattern 3} & $\underline{a}+\underline{b}$ & $\underline{c}$ & $\underline{b}+\underline{c}$ \\ \cline{2-4} 
                  & $\underline{b}$ & $\underline{a}+\underline{c}$ & $\underline{b}+\underline{c}$ \\ \hline
\end{tabular}
\caption{The table of the whole recovery patterns used in the PIR scheme in~\cref{table:PIR scheme over k3}.}
\label{table:recovery patterns in k3}
\end{table}

\begin{table}[ht]
\centering
\begin{tabular}{|c|c|c|c|}
\hline
                  & $S_1$ & $S_2$ & $S_3$ \\ \hline
\multirow{2}{*}{recovery pattern 1} & $a_1$ &  &  \\ \cline{2-4} 
                  &  & $a_2$ &  \\ \hline
\multirow{2}{*}{recovery pattern 2} & $a_3+b_1$ &  & $b_1$ \\ \cline{2-4} 
                  &  & $a_4+c_1$ & $c_1$ \\ \hline
\multirow{2}{*}{recovery pattern 3} & $a_5+b_2$ & $c_2$ & $b_2+c_2$ \\ \cline{2-4} 
                  & $b_3$ & $a_6+c_3$ & $b_3+c_3$ \\ \hline
\end{tabular}
\caption{The adjusted scheme over $K_3$ according to recovery pattern, for $\theta =A$.}
\label{table:adjusted scheme over k3}
\end{table}

Moreover, if the recovery patterns used in the PIR scheme and their frequencies are known a priori, the corresponding PIR scheme can be directly reconstructed. The reconstruction process involves generating the sub-file indices according to the following rules:
\begin{itemize}
    \item The subscripts of sub-files belonging to the desired file must be pairwise distinct;
    \item Within each recovery pattern, the subscripts of sub-files from the same non-desired file must be identical;
    \item In the query sent to any individual server, the subscripts of sub-files from the same non-desired file must be pairwise distinct.
\end{itemize}
By applying these three rules, the PIR scheme presented in~\cref{table:adjusted scheme over k3} can be reconstructed, given that each recovery pattern in~\cref{table:recovery patterns in k3} is used exactly once.

\subsection{Warm-up: A new construction of PIR scheme over \( K_4 \)}\label{subsection:PIR scheme K4}
\subsubsection{Construction of PIR scheme over $K_{4}$ with rate $\frac{7}{20}$}
Building on the PIR scheme over \( K_3 \), we now present our construction of a PIR scheme over \( K_4 \). In the spirit of~\cite{Sun2017capacity}, our approach follows a myopic (greedy) paradigm: it begins by focusing on the retrieval of specific sub-files from the desired message and gradually evolves into a complete PIR scheme through iterative refinements. This construction is guided by the following three fundamental design principles:
\begin{itemize}
    \item[\textup{(1)}] enforcing symmetry across servers;
    \item[\textup{(2)}] ensuring message symmetry within queries to each server;
    \item[\textup{(3)}] leveraging side information from undesired messages to recover additional desired sub-files whenever possible.
\end{itemize}
While the high-level strategy is reminiscent of that in~\cite{Sun2017capacity}, our scheme departs from prior approaches in a key aspect: the recovery of desired sub-files is orchestrated through a carefully selected collection of recovery patterns. The methodology for identifying and deploying these patterns is elaborated in~\cref{subsubsection:Selection of Recovery Patterns}.

In this subsection, we illustrate how these principles can be systematically applied to extend the PIR scheme over \( K_3 \) to a new construction over \( K_4 \). The resulting scheme achieves a retrieval rate of \( \frac{7}{20} \). We model the underlying replication system as the complete graph \( K_4 \), consisting of four servers \( S_1, S_2, S_3, S_4 \), where each file is represented as a binary vector. Without loss of generality, we assume that the desired file, denoted by \( A \), is uniquely stored on servers \( S_1 \) and \( S_2 \). The remaining files in the complete graph \( K_4 \) are labeled \( B, C, D, E, F \), as illustrated in~\cref{fig:K_4}. These files can be naturally categorized into two groups: those that are stored on either \( S_1 \) or \( S_2 \), and those that are stored on neither.

To clarify this classification, we employ a color-coded representation in~\cref{fig:K_4}. The edge corresponding to the desired file \( A \), which is shared between \( S_1 \) and \( S_2 \), is shown in black. Edges representing files stored on exactly one of \( S_1 \) or \( S_2 \) are depicted in blue. Lastly, edges corresponding to files stored on neither \( S_1 \) nor \( S_2 \) are drawn in red.

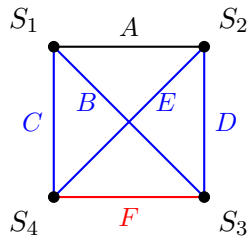
\begin{figure}[h]
\centering
\begin{tikzpicture}[
    vertex/.style={circle, draw, fill=black, inner sep=0pt, minimum size=4pt},
    edge/.style={midway, font=\small}
]

% 定义顶点坐标
\coordinate (S1) at (0,2);
\coordinate (S2) at (2,2);
\coordinate (S3) at (2,0);
\coordinate (S4) at (0,0);
% 绘制边并标注
\draw[thick, black] (S1) -- node[edge, above] {$A$} (S2);
\draw[thick, blue] (S1) -- node[edge, above left, xshift=-8pt] {$B$} (S3);
\draw[thick, blue] (S1) -- node[edge, left] {$C$} (S4);
\draw[thick, blue] (S2) -- node[edge, right] {$D$} (S3);
\draw[thick, blue] (S2) -- node[edge, above right, xshift=6pt] {$E$} (S4);
\draw[thick, red] (S3) -- node[edge, below] {$F$} (S4);

% 标注顶点（使用数学模式 $S_1$, $S_2$, $S_3$）
\node[vertex, label=above left:$S_1$] at (S1) {};
\node[vertex, label=above right:$S_2$] at (S2) {};
\node[vertex, label=below right:$S_3$] at (S3) {};
\node[vertex, label=below left:$S_4$] at (S4) {};
\end{tikzpicture}
\caption{The corresponding PIR system to $K_4$} % 添加标题
\label{fig:K_4} % 添加标签
\end{figure}

As before, we omit the subscripts of sub-files when listing and analyzing recovery patterns, focusing instead on the structural role of each summation. The selected recovery patterns are carefully designed to combine known side information with a summation involving exactly one sub-file from the desired message. More precisely, our PIR scheme over the complete graph \( K_N \) is constructed in a sequential manner, proceeding through Steps \( 1 \) to \( N \! - \! 1 \). In each Step \( k \), the user recovers certain sub-files of the desired message by applying specific recovery patterns. These patterns consist exclusively of $(k\!-\!1)$-sums and $k$-sums composed of sub-files from undesired files, along with a $k$-sum that includes one sub-file of the desired file. Each recovery pattern is applied several times and is confined to its designated step in the construction.

Suppose a complete set of recovery patterns is predetermined, as shown in~\cref{table:recovery patterns in K4}. The table lists all recovery patterns for retrieving the desired file \( A \), where subscripts are omitted for simplicity. To eliminate redundancy, patterns are grouped into classes based on structural equivalence.

\begin{table}[ht]
\small
\centering
\begin{tabular}{|c|c|c|c|c|c|}
\hline
Steps & Recovery patterns & $S_1$ & $S_2$ & $S_3$ & $S_4$ \\ \hline
\multirow{2}{*}{1} & \multirow{2}{*}{1} & $\underline{a}$ &  &  &  \\ \cline{3-6}
            & &  & $\underline{a}$ &  &  \\ \hline
\multirow{12}{*}{2} & \multirow{4}{*}{2} & $\underline{a}+\underline{b}$ &  & $\underline{b}$ &  \\ \cline{3-6}
            & & $\underline{a}+\underline{c}$ &  &  & $\underline{c}$ \\ \cline{3-6}
            & &  & $\underline{a}+\underline{d}$ & $\underline{d}$ &  \\ \cline{3-6}
            & &  & $\underline{a}+\underline{e}$ &  & $\underline{e}$ \\ \cline{2-6}
 & \multirow{4}{*}{3} & $\underline{a}+\underline{b}$ & $\underline{d}$ & $\underline{b}+\underline{d}$ &  \\ \cline{3-6}
            & & $\underline{a}+\underline{c}$ & $\underline{e}$ &  & $\underline{c}+\underline{e}$ \\ \cline{3-6}
            & & $\underline{b}$ & $\underline{a}+\underline{d}$ & $\underline{b}+\underline{d}$ &  \\ \cline{3-6}
            & & $\underline{c}$ & $\underline{a}+\underline{e}$ &  & $\underline{c}+\underline{e}$ \\ \cline{2-6}
 & \multirow{4}{*}{4} & $\underline{a}+\underline{b}$ &  & $\underline{b}+\underline{f}$ & $\underline{f}$ \\ \cline{3-6}
            & & $\underline{a}+\underline{c}$ &  & $\underline{f}$ & $\underline{c}+\underline{f}$ \\ \cline{3-6}
            & &  & $\underline{a}+\underline{d}$ & $\underline{d}+\underline{f}$ & $\underline{f}$ \\ \cline{3-6}
            & &  & $\underline{a}+\underline{e}$ & $\underline{f}$ & $\underline{e}+\underline{f}$ \\ \hline
\multirow{6}{*}{3} & \multirow{2}{*}{5} & $\underline{a}+\underline{b}+\underline{c}$ &  & $\underline{b}+\underline{f}$ & $\underline{c}+\underline{f}$ \\ \cline{3-6}
            & &  & $\underline{a}+\underline{d}+\underline{e}$ & $\underline{d}+\underline{f}$ & $\underline{e}+\underline{f}$ \\ \cline{2-6}
 & \multirow{2}{*}{6} & $\underline{a}+\underline{b}+\underline{c}$ & $\underline{d}+\underline{e}$ & $\underline{b}+\underline{d}$ & $\underline{c}+\underline{e}$ \\ \cline{3-6}
            & & $\underline{b}+\underline{c}$ & $\underline{a}+\underline{d}+\underline{e}$ & $\underline{b}+\underline{d}$ & $\underline{c}+\underline{e}$ \\ \cline{2-6}
 & \multirow{2}{*}{7} & $\underline{a}+\underline{b}+\underline{c}$ & $\underline{d}+\underline{e}$ & $\underline{b}+\underline{d}+\underline{f}$ & $\underline{c}+\underline{e}+\underline{f}$ \\ \cline{3-6}
            & & $\underline{b}+\underline{c}$ & $\underline{a}+\underline{d}+\underline{e}$ & $\underline{b}+\underline{d}+\underline{f}$ & $\underline{c}+\underline{e}+\underline{f}$ \\ \hline
\end{tabular}
\caption{The table of the whole recovery patterns used in our PIR scheme over $K_4$.}
\label{table:recovery patterns in K4}
\end{table}

Our objective is to construct a PIR scheme over the complete graph \( K_4 \) such that, for each \( 1 \leq k \leq 3 \), the \( k \)-sums employed in the scheme appear with nonnegative multiplicities \( x_k \geq 0 \). The construction proceeds sequentially, with each \( x_k \) determined based on the residual contribution of \( (k\!-\!1) \)-sums selected in the preceding step.

The scheme is built over three steps. In each step \( 1 \leq k \leq 3 \), we select a tailored collection of recovery patterns for the desired file \( A \), where each pattern consists solely of \( (k\!-\!1) \)-sums and \( k \)-sums extracted from server responses. After step \( k\!-\!1 \), certain \( (k\!-\!1) \)-sums may remain unused. In step \( k \), these leftover side information terms are combined with carefully chosen \( k \)-sums to facilitate the recovery of additional sub-files of \( A \), thereby determining the value of \( x_k \). Throughout this process, the construction aims to make optimal use of the residual \( (k\!-\!1) \)-sums to maximize efficiency.

We now formally describe the procedure for constructing the PIR scheme. The process begins by setting the multiplicity \( x_1 \) of 1-sums to a fixed positive integer \( M \), which is assumed to be sufficiently large. Another key parameter is the sub-packetization level \( L \), which remains unspecified at this stage. Both \( M \) and \( L \) will be explicitly determined after the completion of Steps 1 through 3.

Let \( [a_1,\ldots, a_L] \), \( [b_1, \dots, b_L] \), \( [c_1, \dots, c_L] \), \( [d_1, \dots, d_L] \), \( [e_1, \dots, e_L] \), and \( [f_1, \dots, f_L] \) represent independent, uniformly random permutations of the \( L \) sub-files of files \( A, B, C, D, E, \) and \( F \), respectively. These permutations are generated privately by the user and serve to randomize the ordering of sub-files across the scheme.

In the following, we formally describe the operations performed in each of the three steps. Although the arrangement of sub-file indices within these steps may initially seem unintuitive, the reasoning behind these choices will become evident upon completion of the full construction.

\paragraph{Step 1.} The user initiates the scheme by directly querying sub-files \( a_1, \dots, a_M \) from server \( S_1 \), and \( a_{M+1}, \dots, a_{2M} \) from server \( S_2 \), thereby retrieving \( 2M \) portions of the desired file \( A \). To maintain message symmetry, an equal number of sub-file portions must be requested from the undesired files stored at \( S_1 \) and \( S_2 \). Accordingly, the user queries \( b_{2M+1}, \dots, b_{3M} \) and \( c_{2M+1}, \dots, c_{3M} \) from \( S_1 \), and \( d_{M+1}, \dots, d_{2M} \) and \( e_{M+1}, \dots, e_{2M} \) from \( S_2 \).

To enforce symmetry across servers, the user also issues queries to \( S_3 \) and \( S_4 \). From server \( S_3 \), the user requests \( b_1, \dots, b_M \), \( d_1, \dots, d_M \), \( f_{\frac{M}{2}+1}, \dots, f_M \), and \( f_{\frac{3M}{2}+1}, \dots, f_{2M} \); from server \( S_4 \), the user queries \( c_1, \dots, c_M \), \( e_1, \dots, e_M \), \( f_1, \dots, f_{\frac{M}{2}} \), and \( f_{M+1}, \dots, f_{\frac{3M}{2}} \). This ensures that all servers contribute equally to the initial stage of the scheme.

At the conclusion of Step 1, the user has collected a pool of side information comprising:
\begin{itemize}
  \item \( b_{2M+1}, \dots, b_{3M} \) and \( c_{2M+1}, \dots, c_{3M} \) from \( S_1 \),
  \item \( d_{M+1}, \dots, d_{2M} \) and \( e_{M+1}, \dots, e_{2M} \) from \( S_2 \),
  \item all sub-files queried from \( S_3 \) and \( S_4 \).
\end{itemize}

These portions constitute the available 1-sum side information, which will be utilized in subsequent steps. A summary of the remaining portions of each relevant 1-sum is provided in~\cref{table:side information K4}.

\paragraph{Step 2.} In this step, we exploit the remaining side information left in block~1 after Step~1 by applying the recovery patterns outlined in~\cref{table:recovery patterns in K4}.

We begin with Recovery Pattern~2. For all \( j = 1, 2, \dots, M \), the user queries \( a_{2M+j} + b_j \) from \( S_1 \). Since \( b_j \) is known from block~1 of \( S_3 \), the user can recover \( a_{2M+j} \). Similarly, the user obtains \( a_{3M+j} + c_j \) from \( S_1 \), and \( a_{4M+j} + d_j \), \( a_{5M+j} + e_j \) from \( S_2 \). With \( c_j \), \( d_j \), and \( e_j \) available from block~1 of \( S_4 \) and \( S_3 \), the user successfully recovers \( a_{3M+1}, \dots, a_{6M} \). This process fully utilizes the side information \( b_1, \dots, b_M \) and \( d_1, \dots, d_M \) from \( S_3 \), as well as \( c_1, \dots, c_M \) and \( e_1, \dots, e_M \) from \( S_4 \).

Next, we apply Recovery Pattern~3. For each \( j = 1, 2, \dots, M \), the user performs the following:
\begin{itemize}
  \item Requests \( a_{6M+j} + b_{M+j} \) from \( S_1 \) and \( b_{M+j} + d_{M+j} \) from \( S_3 \) to recover \( a_{6M+j} \), using \( d_{M+j} \) from \( S_2 \).
  \item Requests \( a_{7M+j} + c_{M+j} \) from \( S_1 \) and \( c_{M+j} + e_{M+j} \) from \( S_4 \) to recover \( a_{7M+j} \), using \( e_{M+j} \) from \( S_2 \).
  \item Requests \( a_{8M+j} + d_{2M+j} \) from \( S_2 \) and \( b_{2M+j} + d_{2M+j} \) from \( S_3 \) to recover \( a_{8M+j} \), using \( b_{2M+j} \) from \( S_1 \).
  \item Requests \( a_{9M+j} + e_{2M+j} \) from \( S_2 \) and \( c_{2M+j} + e_{2M+j} \) from \( S_4 \) to recover \( a_{9M+j} \), using \( c_{2M+j} \) from \( S_1 \).
\end{itemize}
This exhausts the remaining side information \( b_{2M+1}, \dots, b_{3M} \), \( c_{2M+1}, \dots, c_{3M} \) from \( S_1 \), and exhausts \( d_{M+1}, \dots, d_{2M} \), \( e_{M+1}, \dots, e_{2M} \) from \( S_2 \).

We now turn to the final portion of side information from block~1, consisting of:
\begin{itemize}
  \item \( f_{\frac{M}{2}+1}, \dots, f_M \), \( f_{\frac{3M}{2}+1}, \dots, f_{2M} \) from \( S_3 \),
  \item \( f_1, \dots, f_{\frac{M}{2}} \), \( f_{M+1}, \dots, f_{\frac{3M}{2}} \) from \( S_4 \).
\end{itemize}
To exploit this, the user applies recovery pattern~4 for all \( j = 1, 2, \dots, \frac{M}{2} \), as follows:
\begin{itemize}
  \item Requests \( a_{10M+j} + b_{3M+j} \) from \( S_1 \) and \( b_{3M+j} + f_j \) from \( S_4 \) to recover \( a_{10M+j} \), using \( f_j \) from \( S_4 \).
  \item Requests \( a_{\frac{21M}{2}+j} + c_{3M+j} \) from \( S_1 \) and \( c_{3M+j} + f_{\frac{M}{2}+j} \) from \( S_3 \) to recover \( a_{\frac{21M}{2}+j} \), using \( f_{\frac{M}{2}+j} \) from \( S_3 \).
  \item Requests \( a_{11M+j} + d_{3M+j} \) from \( S_2 \) and \( d_{3M+j} + f_{M+j} \) from \( S_4 \) to recover \( a_{11M+j} \), using \( f_{M+j} \) from \( S_4 \).
  \item Requests \( a_{\frac{23M}{2}+j} + e_{3M+j} \) from \( S_2 \) and \( e_{3M+j} + f_{\frac{3M}{2}+j} \) from \( S_3 \) to recover \( a_{\frac{23M}{2}+j} \), using \( f_{\frac{3M}{2}+j} \) from \( S_3 \).
\end{itemize}
As a result, the user successfully recovers \( a_{10M+1}, \dots, a_{12M} \), fully utilizing the remaining side information in block~1.

We now analyze the usage frequencies of the various 2-sum types. The most frequently used are \( A+B \), \( A+C \) from \( S_1 \), and \( A+D \), \( A+E \) from \( S_2 \), each appearing \( \frac{5M}{2} \) times. Denote this maximal multiplicity by \( x_2 = \frac{5M}{2} \).

To balance the number of occurrences across all 2-sum types and prepare new side information for subsequent steps, we introduce the following additional 2-sum queries (see~\cref{table:side information K4} for a summary):
\begin{itemize}
  \item From \( S_1 \): request \( b_{\frac{23M}{4}+j} + c_{\frac{23M}{4}+j} \) for \( j = 1, \dots, \frac{M}{4} \), and \( b_{\frac{33M}{4}+j} + c_{\frac{33M}{4}+j} \) for \( j = 1, \dots, \frac{9M}{4} \).
  \item From \( S_2 \): request \( d_{\frac{11M}{2}+j} + e_{\frac{11M}{2}+j} \) for \( j = 1, \dots, \frac{M}{4} \), and \( d_{6M+j} + e_{6M+j} \) for \( j = 1, \dots, \frac{9M}{4} \).
  \item From \( S_3 \): request \( b_{\frac{7M}{2}+j} + f_{2M+j} \) and \( d_{\frac{7M}{2}+j} + f_{4M+j} \) for \( j = 1, \dots, 2M \), and \( b_{\frac{11M}{2}+j} + d_{\frac{11M}{2}+j} \) for \( j = 1, \dots, \frac{M}{2} \).
  \item From \( S_4 \): request \( c_{\frac{7M}{2}+j} + f_{2M+j} \) and \( e_{\frac{7M}{2}+j} + f_{4M+j} \) for \( j = 1, \dots, 2M \), and \( c_{\frac{11M}{2}+j} + e_{\frac{11M}{2}+j} \) for \( j = 1, \dots, \frac{M}{2} \).
\end{itemize}

A summary of the remaining portions of each relevant 2-sum is provided in~\cref{table:side information K4}.

\paragraph{Step 3.} Following a similar approach to Step~2, we apply recovery patterns~5, 6, and 7 in sequence to exploit the remaining side information in block~2.

We begin with recovery pattern~5. For all \( j = 1, 2, \dots, 2M \), the user requests \( a_{12M+j} + b_{\frac{7M}{2}+j} + c_{\frac{7M}{2}+j} \) from \( S_1 \). Since \( b_{\frac{7M}{2}+j} + f_{2M+j} \) and \( c_{\frac{7M}{2}+j} + f_{2M+j} \) are available in block~2 of \( S_3 \) and \( S_4 \), respectively, the user can recover \( a_{12M+j} \). Similarly, the user requests \( a_{14M+j} + d_{\frac{7M}{2}+j} + e_{\frac{7M}{2}+j} \) from \( S_2 \), and uses the side information \( d_{\frac{7M}{2}+j} + f_{4M+j} \) from \( S_3 \) and \( e_{\frac{7M}{2}+j} + f_{4M+j} \) from \( S_4 \) to recover \( a_{14M+j} \). Thus, recovery pattern~5 exploits the side information \( b_{\frac{7M}{2}+j} + f_{2M+j} \), \( d_{\frac{7M}{2}+j} + f_{4M+j} \) from \( S_3 \), and \( c_{\frac{7M}{2}+j} + f_{2M+j} \), \( e_{\frac{7M}{2}+j} + f_{4M+j} \) from \( S_4 \).

We next apply recovery pattern~6. For all \( j = 1, 2, \dots, \frac{M}{4} \), the user proceeds as follows:
\begin{itemize}
    \item Requests \( a_{16M+j} + b_{\frac{11M}{2}+j} + c_{\frac{11M}{2}+j} \) from \( S_1 \) to recover \( a_{16M+j} \), leveraging the side information \( d_{\frac{11M}{2}+j} + e_{\frac{11M}{2}+j} \) from \( S_2 \), \( b_{\frac{11M}{2}+j} + d_{\frac{11M}{2}+j} \) from \( S_3 \), and \( c_{\frac{11M}{2}+j} + e_{\frac{11M}{2}+j} \) from \( S_4 \).
    \item Requests \( a_{\frac{65M}{4}+j} + d_{\frac{23M}{4}+j} + e_{\frac{23M}{4}+j} \) from \( S_2 \) to recover \( a_{\frac{65M}{4}+j} \), using \( b_{\frac{23M}{4}+j} + c_{\frac{23M}{4}+j} \) from \( S_1 \), \( b_{\frac{23M}{4}+j} + d_{\frac{23M}{4}+j} \) from \( S_3 \), and \( c_{\frac{23M}{4}+j} + e_{\frac{23M}{4}+j} \) from \( S_4 \).
\end{itemize}
At this point, all remaining side information in block~2 from \( S_3 \) and \( S_4 \) has been utilized. What remains in block~2 is:
\begin{itemize}
    \item \( b_{\frac{75M}{4}+j} + c_{\frac{75M}{4}+j} \), for \( j = 1, \dots, \frac{9M}{4} \), from \( S_1 \);
    \item \( d_{\frac{33M}{2}+j} + e_{\frac{33M}{2}+j} \), for \( j = 1, \dots, \frac{9M}{4} \), from \( S_2 \).
\end{itemize}
To exploit the remaining side information, the user applies recovery pattern~7 for all \( j = 1, 2, \dots, \frac{9M}{4} \):
\begin{itemize}
    \item Requests \( a_{\frac{33M}{2}+j} + b_{6M+j} + c_{6M+j} \) from \( S_1 \), \( b_{6M+j} + d_{6M+j} + f_{6M+j} \) from \( S_3 \), and \( c_{6M+j} + e_{6M+j} + f_{6M+j} \) from \( S_4 \) to recover \( a_{\frac{33M}{2}+j} \), using \( d_{6M+j} + e_{6M+j} \) from \( S_2 \).
    \item Requests \( a_{\frac{75M}{4}+j} + d_{\frac{33M}{4}+j} + e_{\frac{33M}{4}+j} \) from \( S_2 \), \( b_{\frac{33M}{4}+j} + d_{\frac{33M}{4}+j} + f_{\frac{33M}{4}+j} \) from \( S_3 \), and \( c_{\frac{33M}{4}+j} + e_{\frac{33M}{4}+j} + f_{\frac{33M}{4}+j} \) from \( S_4 \) to recover \( a_{\frac{75M}{4}+j} \), using \( b_{\frac{33M}{4}+j} + c_{\frac{33M}{4}+j} \) from \( S_1 \).
\end{itemize}

As a result, the user successfully recovers all symbols \( a_{\frac{33M}{2}+1}, \dots, a_{21M} \), thereby fully exhausting the side information in block~2, as shown in~\cref{table:side information K4}.

Finally, we analyze the usage frequency of each 3-sum in this step. Each type of 3-sum from each server is used exactly \( \frac{9M}{2} \) times. Therefore, we set the multiplicity of the 3-sum, denoted by \( x_3 \), to this value. No side information remains after Step~3.

\begin{table}[ht]
\centering
\begin{tabular}{|c|c|c|c|c|c|c|}
\hline
             Steps & $S_1$ & $S_2$ & Remaining & $S_3$ & $S_4$ & Remaining \\ \hline
\multirow{2}{*}{Step $1$} & $\underline{a}$ & $\underline{a}$ & $0$ & $\underline{b},\underline{d}$ & $\underline{c},\underline{e}$ & $M$ \\ \cline{2-7} 
            & $\underline{b},\underline{c}$ & $\underline{d},\underline{e}$ & $M$ & $\underline{f}$ & $\underline{f}$ & $M$ \\ \hline
            \hline
\multirow{2}{*}{Step $2$} & $\underline{a}+\underline{b},\underline{a}+\underline{c}$ & $\underline{a}+\underline{d},\underline{a}+\underline{e}$ & $0$ & $\underline{b}+\underline{d}$ & $\underline{c}+\underline{e}$ & 
            $\frac{M}{2}$ \\ \cline{2-7} 
            & $\underline{b}+\underline{c}$ & $\underline{d}+\underline{e}$ & $\frac{5M}{2}$ & $\underline{b}+\underline{f},\underline{d}+\underline{f}$ & $\underline{c}+\underline{f},\underline{e}+\underline{f}$ & $2M$ \\ \hline \hline
   Step $3$ & $\underline{a}+\underline{b}+\underline{c}$ & $\underline{a}+\underline{d}+\underline{e}$ & $0$ & $\underline{b}+\underline{d}+\underline{f}$ & $\underline{c}+\underline{e}+\underline{f}$ & $0$ \\ \hline
\end{tabular}
\caption{The number of remaining summations of each type in block \( k \) after completing step \( k \). }
\label{table:side information K4}
\end{table}

The parameter \( M \) is chosen to ensure that all index ranges in Steps~1--3 are valid integers. For instance, recovery pattern~6 involves an index range \( j = 1, 2, \dots, M/4 \), which requires \( M \) to be divisible by 4. To satisfy this and similar constraints, we set \( M = 4 \).

The subpacketization level \( L \) is defined as the total number of distinct subfiles of \( A \) recovered throughout all three steps. Since this total equals \( 21M \), we obtain \( L = 21 \times 4 = 84 \).

Upon completing Steps~1--3, we obtain a specific sequence of recovery patterns for \( A \), which can be grouped according to the general recovery patterns presented in~\cref{table:recovery patterns in K4} (suppressing subscripts for simplicity). Recall that the subpacketization corresponds to the entropy of each file. In our PIR scheme, each server returns the same number of symbols, and thus the entropy of the answer \( A_i \), conditioned on the query \( \QQ_i \), is determined by the multiplicities \( x_k \) of the \( k \)-sum patterns used:
\[
H(A_i \mid \QQ_i) = x_1 \binom{3}{1} + x_2 \binom{3}{2} + x_3 \binom{3}{3} = 60.
\]
This leads to a PIR scheme over \( K_4 \) with rate
\[
\frac{H(A)}{\sum_{i=1}^4 H(A_i \mid \QQ_i)} = \frac{84}{4 \times 60} = \frac{7}{20},
\]
as summarized in~\cref{table:specific K4} in~\cref{section:Table for PIR scheme over k4}.

The reliability of this PIR scheme over \( K_4 \) is ensured since for each \( j \in [84] \), we provide a specific recovery pattern to retrieve the sub-file \( a_j \). Furthermore, privacy is achieved through two key mechanisms.

First, each sub-file in the queries is indexed using a random permutation, independently chosen for each file. These permutations are generated privately by the user and remain unknown to the servers.

Second, the values of \( x_1, x_2, x_3 \)—which determine the number of summations of each type—are independent of the desired file. Consequently, for any server and any summation type, the frequency of occurrence in the queries is also independent of the desired file, being entirely determined by the fixed parameters \( x_1, x_2, x_3 \).

Since the permutations are chosen uniformly at random and independently across files, and the summation frequencies are independent of the desired file, all possible query realizations are equally likely regardless of which file is requested. This ensures that the scheme satisfies the privacy condition.

Moreover, from~\cref{table:specific K4}, we observe that the subscripts of each sub-file are assigned according to the order in which the recovery patterns are constructed. Specifically, our scheme sequentially generates a series of recovery patterns, each corresponding to a row in~\cref{table:specific K4}, and assigns sub-file subscripts based on the following rules:
\begin{itemize}
    \item[\textup{(1)}] For the desired file, the subscripts of its sub-files are unique across all server responses.
    \item[\textup{(2)}] For each non-desired file, the subscripts of its sub-files are unique within the response of each individual server.
    \item[\textup{(3)}] Within any given recovery pattern, all sub-files of the same file share the same subscript.
\end{itemize}

As an example, consider the first row of Step~2 in~\cref{table:specific K4}:
\begin{itemize}
    \item[\textup{(1)}] The sub-file \(a_9\) appears in the answer from \(S_1\) in this recovery pattern and does not appear in any other recovery pattern.
    \item[\textup{(2)}] The sub-file \(b_1\) appears in the responses involved in this recovery pattern and is not reused elsewhere.
    \item[\textup{(3)}] Since both \(S_1\) and \(S_3\) include sub-files of file \(B\) in this recovery pattern, those sub-files share the same subscript, namely \(b_1\).
\end{itemize}

This illustrates a general principle: in the \(n\)-th recovery pattern of our scheme, the subscript of a sub-file from a particular file depends on how many prior recovery patterns (from the 1st to the \((n\!-\!1)\)-th) have involved sub-files from that file. For instance, in the 9th recovery pattern in~\cref{table:specific K4} (i.e., the first row of Step~2), the sub-file of file \(A\) is assigned subscript 9 because sub-files of \(A\) have appeared in each of the first 8 recovery patterns. On the other hand, sub-files of file \(B\) receive subscript 1 since this is their first appearance, and within the same recovery pattern, all sub-files of \(B\) must share this subscript.

This observation implies that it suffices to generate a sequence of recovery patterns without specifying sub-file indices. The subscripts of the sub-files can then be naturally assigned based on the order in which these recovery patterns are constructed.

\subsubsection{A remark on tiny gap for $\mathcal{C}(K_{4})$}
We compare the upper bound given in~\cref{thm:UpperBoundMain} for \( K_4 \) with the rate achieved by the PIR scheme constructed above. Recall from~\cref{thm:UpperBoundMain} that the capacity of \( K_4 \) satisfies \( \CC(K_4) \le \frac{6}{17} \). However, a small gap remains between this upper bound and the rate achieved by our proposed scheme. To understand the source of this gap, we examine the tightness of each inequality used in the proof of~\cref{thm:UpperBoundMain}. Notably, when \( N = 4 \), the only inequality that is not tight occurs in~\eqref{eq:thm:upperbound not tight} when \( i = 3 \).

In particular, consider the case where the desired file is stored at servers \( S_1 \) and \( S_2 \). In the proof of~\cref{thm:UpperBoundMain}, the following chain of inequalities appears:
\begin{align}
    H(A_4 \mid A_2, A_3, W_{1,2} = A, \QQ, \theta = (1,2)) 
    &\ge H(A_4 \mid A_2, A_3, \WW_{S_2}, \WW_{S_3}, \QQ, \theta = (1,2)) \label{eq:compare 1} \\
    &= H(A_4 \mid \WW_{S_2}, \WW_{S_3}, \QQ, \theta = (1,2)) \label{eq:compare 2} \\
    &\ge \frac{1}{2}L. \label{eq:compare 3}
\end{align}

In our PIR scheme shown in~\cref{table:specific K4}, the term \( H(A_4 \mid A_2, A_3, W_{1,2} = A, \QQ, \theta = (1,2)) \) evaluates to \( 44 \), whereas the corresponding lower bounds in~\eqref{eq:compare 1}--\eqref{eq:compare 3} all evaluate to \( 42 \). This discrepancy indicates that our PIR scheme saturates all other inequalities and that the looseness in~\eqref{eq:compare 3} is the sole reason the upper bound \( \CC(K_4) \le \frac{6}{17} \) might be not tight. Actually, we believe that the scheme in~\cref{table:specific K4} is optimal.

Furthermore, since the proof of~\cref{thm:UpperBoundMain} proceeds recursively, any slackness in the bound for \( \CC(K_4) \) propagates and affects the upper bound on \( \CC(K_N) \) for general \( N \).

An additional noteworthy observation is that the PIR scheme over \( K_3 \), obtained by applying~\cref{lemma:graph contain} to the scheme in~\cref{table:specific K4}, coincides with a repeated version of the optimal PIR scheme over \( K_3 \) presented in~\cref{table:PIR scheme over k3}.

\subsection{General PIR scheme over $K_N$}
\subsubsection{Overview of the scheme}
We now present a recursive construction of a PIR scheme over the complete graph \( K_N \), applicable for any \( N \geq 3 \). In this construction, the answer of each server consists of carefully structured summations selected from \( k \)-sums, where \( 1 \leq k \leq N - 1 \). The scheme is built iteratively over \( N - 1 \) steps, indexed by \( k = 1 \) to \( N - 1 \).

At each step \( k \), we focus exclusively on the \((k\!-\!1)\)-sums and \( k \)-sums appearing in the server responses \( A_{[N]} \). These summations are strategically combined to recover portions of sub-files of the desired file. The key design principle is to define the summations used in step \( k \) based on the set of unused \((k\!-\!1)\)-sums remaining after steps \( 1 \) through \( k\!-\!1 \), which we refer to as the \emph{side information} after step \( k\!-\!1 \).

Throughout the construction, we aim to utilize the available side information in each step as efficiently as possible, ensuring that the scheme fully leverages the combinatorial structure of summations across steps.

Our construction can be viewed as a recursive generalization of the PIR scheme introduced in~\cite{Sun2017capacity}, which serves as a foundational component of our approach. The full construction procedure is described in~\cref{subsubsection:Description of Our Construction}. In~\cref{subsubsection:Verification the Feasibility of Our Construction}, we rigorously establish the feasibility of the proposed scheme by verifying both reliability and privacy. Subsequently,~\cref{subsubsection:Estimation of the Rate} analyzes the resulting PIR rate, highlighting the structural complexity and key challenges involved.

However, a direct implementation of this recursive scheme leads to prohibitively large subpacketization. To address this issue, we propose a probabilistic variant of the scheme in~\cref{subsection:transform}, which significantly reduces the subpacketization to 1.

Inspired by the methodology used to derive our upper bound, we conjecture that this recursive framework, arguably the core innovation of this work, can be extended to construct efficient PIR schemes for certain broader classes of graphs.

\subsubsection{Selection of recovery patterns}\label{subsubsection:Selection of Recovery Patterns}
In~\cref{subsection:PIR scheme K4}, we presented the PIR scheme over \( K_4 \) as an illustrative example, demonstrating how a scheme can be constructed based on a predefined sequence of recovery patterns. In this subsection, we extend this framework to the general complete graph \( K_N \), and describe the strategy for selecting recovery patterns in the general case.

As before, we denote the \( N \) servers in the \( K_N \)-based replication system by \( S_1, S_2, \dots, S_N \). Without loss of generality, we assume that the desired file, denoted by \( A \), is the unique file stored on both \( S_1 \) and \( S_2 \). The remaining files in the system can be partitioned into two categories: those stored on either \( S_1 \) or \( S_2 \), and those stored on neither. We label files in the first category as \( B_{i,j} \), representing files stored on servers \( S_i \) and \( S_j \), where exactly one of \( i \) or \( j \) belongs to \( \{1,2\} \). These are referred to as \emph{blue files}. Files in the second category are labeled as \( R_{i,j} \), corresponding to files stored on servers \( S_i \) and \( S_j \) with \( i, j \notin \{1,2\} \), and are referred to as \emph{red files}. As an illustration, consider the case \( N = 5 \); see~\cref{figure:colored K_5}. In this setting, the edge corresponding to the desired file \( A \) is depicted in black, those corresponding to blue files are shown in blue, and those corresponding to red files are shown in red.

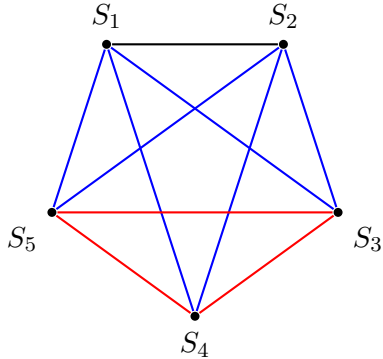
\begin{figure}[hbtp]
\centering
\begin{tikzpicture}[thick]
  % 定义五个顶点，均匀分布在圆周上，只用小黑点表示
  \foreach \i/\angle/\pos in {
      1/126/above,
      2/54/above,
      3/342/below right,
      4/270/below,
      5/198/below left
  } {
    \node[fill=black, circle, inner sep=1.2pt, label=\pos:$S_\i$] (S\i) at (\angle:2cm) {};
  }

  % 绘制所有边（完全图 K5）部分上色
  \draw[black] (S1) -- (S2);
  
  \foreach \i in {1,2} {
    \foreach \j in {3,4,5} {
      \draw[blue] (S\i) -- (S\j);
    }
  }

  \foreach \i in {3,4,5} {
    \foreach \j in {\i,...,5} {
      \ifnum\i=\j
        \relax % 跳过自环
      \else
        \draw[red] (S\i) -- (S\j);
      \fi
    }
  }
\end{tikzpicture}
\caption{The corresponding PIR system for $K_5$, where blue (red) files are indicated by blue (red) edges.}
\label{figure:colored K_5}
\end{figure}

More generally, in the complete graph \( K_N \), server \( S_1 \) (and similarly \( S_2 \)) stores the desired file \( A \) along with \( N - 2 \) blue files. Each server \( S_j \), for \( 3 \leq j \leq N \), stores two blue files and \( N - 3 \) red files. Due to the inherent symmetry of the complete graph, all files of the same color stored on the same server are indistinguishable in their role within the scheme. This observation motivates the following definition.
\begin{defn}\label{def:simpleNotation4.4}
    Let \( k_0 \in \{0,1\} \) and \( k_1, k_2 \in \mathbb{Z}_{\ge 0} \), and set \( k = k_0 + k_1 + k_2 \). Suppose that server \( S_i \), for some \( i \in [N] \), stores \( k_1 \) distinct blue files \( B_{i,j_1}, \dots, B_{i,j_{k_1}} \) and \( k_2 \) distinct red files \( R_{i,j_1'}, \dots, R_{i,j_{k_2}'} \). If \( k_0 = 1 \), assume that the desired file \( A \) is also stored on \( S_i \). Then, regardless of the specific indices \( j_1, \dots, j_{k_1} \) and \( j_1', \dots, j_{k_2}' \), a \( k \)-sum of the form
    \[
        k_0 A + B_{i,j_1} + \cdots + B_{i,j_{k_1}} + R_{i,j_1'} + \cdots + R_{i,j_{k_2}'}
    \]
    is called a $k$-sum with label \( k_0 \ast A + k_1 \ast B + k_2 \ast R \).
\end{defn}

As an example, consider a replication system based on \( K_5 \), and suppose we examine server \( S_3 \). A 3-sum such as \( B_{3,1} + B_{3,2} + R_{3,5} \) is labeled as \( 0 \ast A + 2 \ast B + 1 \ast R \), which we abbreviate as a 3-sum with label \( 2 \ast B + R \) for simplicity. In general, for each block \( k \) in the replication system based on \( K_N \),~\cref{table:type of summation} lists all possible \( k \)-sum labels that may appear from each server, together with the number of distinct types associated with each label. For instance, consider server \( S_1 \). The number of distinct types of \( k \)-sums with label \( A + (k\!-\!1) \ast B \) equals the number of ways to choose \( k - 1 \) distinct blue files from those stored on \( S_1 \), which is given by \( \binom{N - 2}{k - 1} \).

\begin{table}[ht]
\centering
\renewcommand\arraystretch{1.2}
\begin{tabular}{|c|c|c|}
\hline
Server & Label of Summation & Number of Distinct Types \\ \hline
\multirow{2}{*}{\(S_1, S_2\)} 
& \( A + (k\!-\!1) \ast B\) & \(\binom{N - 2}{k - 1}\) \\ \cline{2-3}
& \(k \ast B\) & \(\binom{N - 2}{k}\) \\ \hline
\multirow{3}{*}{\(S_3, \dots, S_N\)} 
& \(2 \ast B + (k\!-\!2) \ast R\) & \(\binom{N - 3}{k - 2}\) \\ \cline{2-3}
& \( B + (k\!-\!1) \ast R\) & \(2\binom{N - 3}{k - 1}\) \\ \cline{2-3}
& \(k \ast R\) & \(\binom{N - 3}{k}\) \\ \hline
\end{tabular}
\caption{All labels of \(k\)-sum in block $k$ for replication system based on $K_N$.}
\label{table:type of summation}
\end{table}
With the aid of the simplified notation introduced in~\cref{def:simpleNotation4.4}, we now proceed to present the recovery patterns for the complete graph \( K_N \) in a direct and structured manner. Similar to the PIR schemes described in~\cref{subsection:PIR scheme K4}, the construction proceeds in \( N - 1 \) sequential steps.

For each \( 1 \leq k \leq N - 1 \), the recovery patterns in step \( k \) involve only \((k\!-\!1)\)-sums and \( k \)-sums extracted from the server responses. More precisely, the patterns at step \( k \) are designed to combine the side information provided by the \((k\!-\!1)\)-sums left in step \( k\!-\!1 \) with carefully selected \( k \)-sums, thereby enabling the recovery of specific sub-files of the desired file \( A \).

As before, structurally equivalent recovery patterns are grouped into equivalence classes to avoid redundancy. Within each class, side information with the same label is utilized by the same recovery pattern. Note that after step \( k\!-\!1 \), the available side information in step \( k \) may include all \((k\!-\!1)\)-sums with labels listed in~\cref{table:type of summation}, except for those with label \( A + (k\!-\!1) \ast B \) stored on servers \( S_1 \) and \( S_2 \). A more detailed description of the overall construction process is provided in~\cref{subsubsection:Description of Our Construction}; in this subsection, we focus solely on the selection of recovery patterns.

To systematically consume all available side information at each step, we define four types of recovery patterns, denoted by \( \alpha \), \( \beta \), \( \gamma \), and \( \zeta \). Each recovery pattern is responsible for utilizing side information with specific labels and recovering portions of the desired file. We now describe each of these patterns in detail.

\paragraph{Recovery Pattern \(\alpha\):} This pattern is designed to utilize side information consisting of \((k\!-\!1)\)-sums with label \( B + (k\!-\!2) \ast R \), obtained from servers \( S_3,\dots,S_N \). For any \( i \in \{1, 2\} \) and any distinct indices \( 3 \le j_1, j_2, \dots, j_{k-1} \le N \), the user sends the following queries:
\begin{itemize}
    \item From each server \( S_{j_\ell} \), for \( \ell \in [k - 1] \), request a \((k\!-\!1)\)-sum of type
    \[
    B_{i,j_\ell} + \sum_{s \in [k - 1] \setminus \{\ell\}} R_{j_\ell, j_s}.
    \]
    
    \item From server \( S_i \), request a \( k \)-sum of type
    \[
    A + \sum_{\ell \in [k - 1]} B_{i,j_\ell}.
    \]
\end{itemize}

Observe that
\[
A = \sum_{\ell \in [k - 1]} \left( B_{i,j_\ell} + \sum_{s \in [k - 1] \setminus \{\ell\}} R_{j_\ell, j_s} \right)
+ \left( A + \sum_{\ell \in [k - 1]} B_{i,j_\ell} \right),
\]
which implies that the user can recover a portion of \( A \) by summing all the responses obtained above.

We refer to the above construction as \emph{Recovery Pattern~\(\alpha\)} in step \( k \). In our PIR scheme over \( K_N \), this pattern is used to consume all available side information with label \( B + (k\!-\!2) \ast R \) in block \( k\!-\!1 \).

We now turn to the second class of recovery patterns, which is designed to utilize different types of available side information. Specifically, the next recovery pattern consumes \((k\!-\!1)\)-sums with label \( 2 \ast B + (k\!-\!3) \ast R \) from servers \( S_3,\dots,S_N \), as well as \((k\!-\!1)\)-sums with label \( (k\!-\!1) \ast B \) from servers \( S_1 \) and \( S_2 \). The construction of this pattern depends on the parity of \( k \).

We now introduce the second class of recovery patterns, which are designed to utilize different types of available side information. Specifically, the following construction consumes \((k\!-\!1)\)-sums with label \( 2 \ast B + (k\!-\!3) \ast R \) from servers \( S_3, \dots, S_N \), as well as a portion of the \((k\!-\!1)\)-sums with label \( (k\!-\!1) \ast B \) from servers \( S_1 \) and \( S_2 \). The construction depends on the parity of \( k \).

\paragraph{Recovery Pattern~\(\beta\):} Let \( i \in \{1, 2\} \), and let \( i' \) denote the other index in \( \{1, 2\} \). Let \( 3 \le j_1, j_2, \dots, j_{k - 1} \le N \) be distinct indices. The user proceeds as follows:

\paragraph{When \( k \) is odd:}
\begin{itemize}
    \item From each server \( S_{j_\ell} \) (for \( \ell \in [k - 1] \)), request a \((k\!-\!1)\)-sum of type
    \[
    B_{i,j_\ell} + B_{i',j_\ell} + \sum_{s \in [k - 1] \setminus \{\ell, k - \ell\}} R_{j_\ell, j_s}.
    \]

    \item From server \( S_{i'} \), request a \((k\!-\!1)\)-sum of type
    \(
    \sum_{\ell \in [k - 1]} B_{i',j_\ell}.
    \)

    \item From server \( S_i \), request a \( k \)-sum of type
    \(
    A + \sum_{\ell \in [k - 1]} B_{i,j_\ell}.
    \)
\end{itemize}

Observe that
\[
A = \sum_{\ell \in [k - 1]} \left( B_{i,j_\ell} + B_{i',j_\ell} + \sum_{s \in [k - 1] \setminus \{\ell, k - \ell\}} R_{j_\ell, j_s} \right)
+ \sum_{\ell \in [k - 1]} B_{i',j_\ell}
+ \left( A + \sum_{\ell \in [k - 1]} B_{i,j_\ell} \right),
\]
which enables the user to recover a portion of \( A \) by summing all the responses.

\paragraph{When \( k \) is even:}
\begin{itemize}
    \item From each server \( S_{j_\ell} \) (for \( \ell \in \{2, \dots, k - 1\} \)), request a \((k\!-\!1)\)-sum of type
    \[
    B_{i,j_\ell} + B_{i',j_\ell} + \sum_{s \in [k - 1] \setminus \{\ell, k - \ell + 1\}} R_{j_\ell, j_s}.
    \]

    \item From server \( S_{j_1} \), request a \( k \)-sum of type
    \[
    B_{i,j_1} + B_{i',j_1} + \sum_{s \in \{2, \dots, k - 1\}} R_{j_1, j_s}.
    \]

    \item From server \( S_{i'} \), request a \((k\!-\!1)\)-sum of type
    \(
    \sum_{\ell \in [k - 1]} B_{i',j_\ell}.
    \)

    \item From server \( S_i \), request a \( k \)-sum of type
    \(
    A + \sum_{\ell \in [k - 1]} B_{i,j_\ell}.
    \)
\end{itemize}
Summing all the responses allows the user to recover a portion of \( A \). This parity-dependent construction is referred to as \emph{Recovery Pattern~\(\beta\)} in step \( k \). It consumes all side information with label \( 2 \ast B + (k\!-\!3) \ast R \) in block \( k\!-\!1 \) from servers \( S_3, \dots, S_N \), and partially consumes side information with label \( (k\!-\!1) \ast B \) from servers \( S_1 \) and \( S_2 \).

To consume the \((k\!-\!1)\)-sums with label \( (k\!-\!1) \ast B \) from servers \( S_1 \) and \( S_2 \), we introduce the following recovery pattern.

\paragraph{Recovery Pattern~\(\gamma\):} For any \( i \in \{1, 2\} \), let \( i' \) denote the other element in \( \{1, 2\} \), and let \( 3 \le j_1, j_2, \dots, j_{k - 1} \le N \) be some distinct indices. The user proceeds as follows:

\begin{itemize}
    \item From server \( S_{i'} \), request a \((k\!-\!1)\)-sum of type
    \(
    \sum_{\ell \in [k - 1]} B_{i', j_\ell}.
    \)

    \item From each server \( S_{j_\ell} \), for \( \ell \in [k - 1] \), request a \( k \)-sum of type
    \[
    B_{i, j_\ell} + B_{i', j_\ell} + \sum_{s \in [k - 1] \setminus \{\ell\}} R_{j_\ell, j_s}.
    \]

    \item From server \( S_i \), request a \( k \)-sum of type
    \(
    A + \sum_{\ell \in [k - 1]} B_{i, j_\ell}.
    \)
\end{itemize}

Summing all of these responses enables the user to recover a portion of the desired file \( A \). This pattern guarantees full utilization of all \((k\!-\!1)\)-sum side information with label \( (k\!-\!1) \ast B \) in block \( k\!-\!1 \) stored on servers \( S_1 \) and \( S_2 \). We refer to this as \emph{recovery pattern~\(\gamma\)} in step \( k \).

To consume the remaining side information consisting of \((k\!-\!1)\)-sums with label \( (k\!-\!1) \ast R \) from servers \( S_3, \dots, S_N \), we define the following recovery pattern.

\paragraph{Recovery Pattern~\(\zeta\):} For any \( i \in \{1, 2\} \) and any distinct indices \( 3 \le j_0, j_1, j_2, \dots, j_{k - 1} \le N \), the user performs the following:

\begin{itemize}
    \item From server \( S_{j_0} \), request a \((k\!-\!1)\)-sum of type
    \(
    \sum_{\ell \in [k - 1]} R_{j_0, j_\ell}.
    \)

    \item From each server \( S_{j_\ell} \), for \( \ell \in [k - 1] \), request a \( k \)-sum of type
    \[
    B_{i, j_\ell} + \sum_{s \in [k - 1] \setminus \{\ell\}} R_{j_\ell, j_s}.
    \]

    \item From server \( S_i \), request a \( k \)-sum of type
    \(
    A + \sum_{\ell \in [k - 1]} B_{i, j_\ell}.
    \)
\end{itemize}

Summing all of these responses enables the user to recover a portion of \( A \). This recovery pattern consumes part of the side information with label \( (k\!-\!1) \ast R \) in block \( k\!-\!1 \), obtained from servers \( S_3, \dots, S_N \). We refer to this as \emph{recovery pattern~\(\zeta\)} in step \( k \).

\begin{table}[ht]
\centering
\begin{tabular}{|c|c|c|c|c|c|c|}
\hline
\multirow{2}{*}{} & \multirow{2}{*}{Servers} & \multirow{2}{*}{Label} & \multicolumn{4}{c|}{Recovery Patterns} \\ \cline{4-7}
 & & & $\alpha$ & $\beta$ & $\gamma$ & $\zeta$ \\ \hline
\multirow{4}{*}{$(k-1)$-sum} & $S_1,S_2$ & $(k-1)\ast B$ &  & \checkmark & $\bigstar$ &  \\ \cline{2-7} 
                  & \multirow{3}{*}{$S_3,\dots,S_N$} & $2\ast B+(k-3)\ast R$ &  & $\bigstar$ &  &  \\ \cline{3-7} 
                  &  & $B+(k-2)\ast R$ & $\bigstar$ &  &  &  \\ \cline{3-7} 
                  &  & $(k-1)\ast R$ &  &  &  & $\bigstar$ \\ \hline
\multirow{3}{*}{$k$-sum} & $S_1,S_2$ & $A+(k-1)\ast B$ & \checkmark & \checkmark & \checkmark & \checkmark \\ \cline{2-7} 
                  & \multirow{2}{*}{$S_3,\dots,S_N$} & $2\ast B+(k-2)\ast R$ &  & \checkmark & \checkmark &  \\ \cline{3-7} 
                  &  & $B+(k-1)\ast R$ &  &  &  & \checkmark \\ \hline
\end{tabular}
\caption{Label references in recovery patterns. The symbol $\bigstar$ indicates the target label consumed by each recovery pattern, while \checkmark denotes other labels utilized in the recovery pattern.}
\label{table:recovery pattern refer}
\end{table}

For easier understanding, we list the target label consumed and the other labels utilized by each recovery pattern in~\cref{table:recovery pattern refer}.

So far, we have presented all the recovery patterns which will be used in the construction. In fact, each recovery pattern listed in~\cref{table:recovery patterns in k3} and~\cref{table:recovery patterns in K4} can be viewed as a specific instantiation of one of the general patterns introduced above, determined by particular choices of parameters. It is important to note that, with the aid of recovery patterns~\(\alpha\), \(\beta\), \(\gamma\), and \(\zeta\), the user is able to consume side information with all labels in block \( k\!-\!1 \), except for \((k\!-\!1)\)-sums with label \( A + (k\!-\!2) \ast B \) stored on servers \( S_1 \) and \( S_2 \). However, our PIR scheme is constructed in such a way that this type of side information is never generated.

To ensure uniform utilization of all available side information and to avoid any imbalance in the use of different summation types, the user applies each recovery pattern uniformly over all valid combinations of parameters \( i, i', j_0, \dots, j_{k - 1} \), maintaining equal usage frequency for all \( k \)-sums and \((k\!-\!1)\)-sums with the same label.

\subsubsection{Description of general construction}\label{subsubsection:Description of Our Construction}

In this subsection, we focus on constructing a PIR scheme \( \Pi_N \) over the complete graph \( K_N \). Our goal is to demonstrate how recovery patterns~\(\alpha\), \(\beta\), \(\gamma\), and \(\zeta\), introduced in~\cref{subsubsection:Selection of Recovery Patterns}, can be systematically employed to generate the full PIR scheme over \( K_N \).

As discussed earlier, the construction proceeds in \( N - 1 \) sequential steps. In step \( k \), the user consumes as much of the remaining side information in block \( k - 1 \) as possible by applying recovery patterns~\(\alpha\), \(\beta\), \(\gamma\), and \(\zeta\). The next definition is used to describe the amount of side information in block $k$ after Step $k$.

\begin{defn}[Residual multiplicity]\label{def:residualMulti}
    For distinct $i,j_1,j_2,\dots,j_k\in [N]$, let $W_{i,j_{\ell}}$ denote the file shared by servers $S_i$ and $S_{j_{\ell}}$, for any $\ell\in [k]$. We say that the \emph{residual multiplicity} of a \(k\)-sum of type $W_{i,j_1}+\dots+W_{i,j_k}$ from server \(S_i\) is equal to \(r\) for some non-negative integer \(r\), if, after step \(k\), there are exactly \(r\) distinct \(k\)-sums of this type, from $S_i$, present as side information.
\end{defn}
According to~\cref{def:residualMulti}, it holds that, for any $k$, the residual multiplicity of $k$-sum with some label must be no more than the multiplicity of $k$-sum. Recall that for any \( k \), the recovery patterns \( \alpha, \beta, \gamma, \zeta \) each represents a class of symmetric, structured, specific recovery patterns. We emphasize that when the user consumes side information, all realizations of these recovery patterns are utilized uniformly. This uniform usage will be illustrated in detail using Step 2 as a representative example. Therefore, in our scheme, the residual multiplicity
of a \( k \)-sum depends only on its label and is uniform across all servers. For instance, on server \( S_1 \), the 2-sums \( B_{1,3} + B_{1,4} \) and \( B_{1,3} + B_{1,5} \) share the same label \( 2 \ast B \), and therefore have identical residual multiplicities. Similarly, the 2-sums \( B_{1,3} + R_{3,4} \) from \( S_3 \) and \( B_{1,4} + R_{4,5} \) from \( S_4 \), both labeled \( B + R \), have the same residual multiplicity. Hence, when all \( k \)-sums with a given label share the same residual multiplicity \( r \), we refer to \( r \) as the residual multiplicity of that label, without distinguishing specific types or servers. 

The following proposition provides a complete specification of our PIR scheme construction over the complete graph \( K_N \).

\begin{prop}\label{prop:construction of KN}
Let \( \{x_k\}_{k=1}^{N-1} \), \( \{y_k\}_{k=2}^{N-1} \), and \( \{z_k\}_{k=1}^{N-2} \) be sequences of non-negative numbers, with \( x_1 = z_1 = 1 \), \( y_2 = \frac{N-3}{2} \), and \( z_k, y_k \le x_k \), satisfying the following recursive relations:

If \( k \le \left\lfloor \frac{N}{2} \right\rfloor + 1 \),
\begin{align}
\left\{
\begin{aligned}
x_k &= \frac{N - k + 1}{2}\cdot x_{k-1} + z_{k-1}, \\
y_k &= x_k - 2x_{k-1} + \frac{k - 2}{N - k}\cdot y_{k-1}, \\
z_k &= x_k - \frac{k - 1}{2}\cdot x_{k-1}.
\end{aligned}
\right. \label{eq:case 1}
\end{align}

If \( k \ge \left\lfloor \frac{N}{2} \right\rfloor + 2 \),
\begin{align}
\left\{
\begin{aligned}
x_k &= \frac{k - 1}{2k - N} \cdot (x_{k-1} + z_{k-1}), \\
y_k &= x_k - 2x_{k-1} + \frac{k - 2}{N - k} \cdot y_{k-1}, \\
z_k &= 0.
\end{aligned}
\right. \label{eq:case 2}
\end{align}

These sequences must satisfy the following inequality for all \( k \):
\begin{equation}
    x_k \geq \frac{k}{2(N - k - 1)} \cdot y_k. \label{eq:construction in K_N}
\end{equation}

Let \( M \) be the smallest positive integer such that \( x_kM, y_kM, z_kM \) are all non-negative integers for all relevant \( k \). Then, there exists a PIR scheme \( \Pi_N \) over \( K_N \) satisfying the following properties for all available \( k \):

\begin{itemize}
    \item[\textup{(a)}] The multiplicity of \( k \)-sums is \( x_k M \),
    \item[\textup{(b)}] The residual multiplicity of \( k \)-sums with label \( A + (k\!-\!1) \ast B \) from \( S_1, S_2 \) is zero,
    \item[\textup{(c)}] The residual multiplicity of \( k \)-sums with label \( k \ast B \) from \( S_1, S_2 \) is \( x_k M \),
    \item[\textup{(d)}] The residual multiplicity of \( k \)-sums with label \( 2 \ast B + (k\!-\!2) \ast R \) from \( S_3, \dots, S_N \) is \( y_k M \),
    \item[\textup{(e)}] The residual multiplicity of \( k \)-sums with label \( B + (k\!-\!1) \ast R \) from \( S_3, \dots, S_N \) is \( z_k M \),
    \item[\textup{(f)}] The residual multiplicity of \( k \)-sums with label \( k \ast R \) from \( S_3, \dots, S_N \) is \( x_k M \).
\end{itemize}
\end{prop}

\noindent
Note that \( k \)-sums with label \( 2 \ast B + (k\!-\!2) \ast R \) exist only for \( 2 \le k \le N - 1 \), and those with label \( B + (k\!-\!1) \ast R \) exist only for \( 1 \le k \le N - 2 \). These constraints restrict the valid ranges of the sequences \( \{y_k\}_{k=2}^{N-1} \) and \( \{z_k\}_{k=1}^{N-2} \).

\vspace{0.5em}
\noindent
\paragraph{Structure of the proof of~\cref{thm:LowerBoundMain}} The remainder of this section is devoted to establishing~\cref{prop:construction of KN}. Specifically,~\cref{subsubsection:Description of Our Construction} provides a detailed recursive procedure for constructing the PIR scheme \( \Pi_N \), while~\cref{subsubsection:Verification the Feasibility of Our Construction} rigorously verifies that the constructed scheme satisfies all the required properties, including reliability, privacy, and the recurrence and inequality conditions~\eqref{eq:construction in K_N}. The construction proceeds step-by-step from \( k = 1 \) to \( k = N - 1 \), using the residual multiplicities of \((k\!-\!1)\)-sums as the input to generate the queries and update the parameters at step \( k \).

We first illustrate the scheme through steps \( 1 \) and \( 2 \) as concrete examples. Let us start from Step 1. We assume that the multiplicity of $1$-sum is $M$, where \( M \) is a sufficiently large positive integer. As described in~\cref{prop:construction of KN}, the value of $M$ should be chosen so that \(x_kM,y_kM,z_kM\) are integers for all \(k\). Thus, the multiplicity of $1$-sum satisfies the property (a) in~\cref{prop:construction of KN}.

\paragraph{Step 1.} When \( k = 1 \), the only available recovery pattern is to directly request sub-files of the desired file \( A \) from servers \( S_1 \) and \( S_2 \). Hence, the user retrieves \( M \) portions of \( A \) from each of \( S_1 \) and \( S_2 \). Once this query pattern is specified, Step 1 is complete.

\begin{claim}\label{claim:complete step 1}
    The residual multiplicity of all \( 1 \)-sums after Step 1 satisfies the recurrence in~\eqref{eq:case 1} and properties~\textup{(b)} through~\textup{(f)} in~\cref{prop:construction of KN}.
\end{claim}

\begin{poc}
    Since the user requests \( M \) portions of the \( 1 \)-sum with label \( A \) from both \( S_1 \) and \( S_2 \), property~\textup{(b)} is satisfied. Moreover, as no other \( 1 \)-sums are involved in Step 1 and \( z_1 = 1 \), properties~\textup{(c)} through~\textup{(f)} hold trivially.
\end{poc}

\paragraph{Step 2.} When \( k = 2 \), recovery pattern~\(\beta\) is not applicable, as there are no \( 1 \)-sums with label \( 2 \ast B + (k\!-\!2) \ast R = 2 \ast B \) from servers \( S_3, \dots, S_N \). Thus, we apply recovery patterns~\(\alpha\),~\(\gamma\), and~\(\zeta\) to consume the side information left in Step 1.

To begin, we consume the \( 1 \)-sums with label \( B \) from servers \( S_3, \dots, S_N \), each with residual multiplicity \( M \), by applying recovery pattern~\(\alpha\). According to its definition, for each \( i \in \{1,2\} \) and \( j \in [3,N] \), the user:
\begin{itemize}
    \item Requests \( M \) portions of the \( 1 \)-sum of type \( B_{i,j} \) from \( S_j \),
    \item Requests \( M \) portions of the \( 2 \)-sum of type \( A + B_{i,j} \) from \( S_i \).
\end{itemize}

We apply each realization of recovery pattern~\(\alpha\), that is, each combination of \( i \in \{1,2\} \) and \( j \in [3,N] \), the same number of times. This ensures that all summation types with the same label are used uniformly; specifically, each type is used exactly \( M \) times. This observation motivates the way we compute the total number of summations with a given label when determining residual multiplicity.

As an example, we re-calculate the total number of requested portions of \( 2 \)-sums with label \( A + B \) under recovery pattern~\(\alpha\). Each consumed portion of a \( 1 \)-sum labeled \( B \) from servers \( S_3, \dots, S_N \) corresponds to one request of a \( 2 \)-sum labeled \( A + B \). Since the residual multiplicity of such \( 1 \)-sums is \( M \), the total number of consumed \( 1 \)-sums is
\[
( N - 2 ) \cdot 2 \cdot \binom{N - 3}{0} \cdot M = 2(N - 2)M,
\]
where the multiplier counts: the number of servers, the number of distinct summation types with label \( B \) per server, and the residual multiplicity. Consequently, the user must request \( 2(N - 2)M \) portions of \( 2 \)-sums with label \( A + B \) from \( S_1 \) and \( S_2 \). Each of these servers contains \( N - 2 \) distinct types of \( 2 \)-sums with this label, so the user must request \( M \) portions of each type, consistent with the uniformity requirement above.

Next, to consume the \( 1 \)-sums with label \( B \) from servers \( S_1 \) and \( S_2 \), each with residual multiplicity \( M \), we apply recovery pattern~\(\gamma\). By maintaining uniform use across all realizations of this pattern, each consumed \( 1 \)-sum with label \( B \) from \( S_1 \) or \( S_2 \) results in one request for a \( 2 \)-sum with label \( A + B \) from \( S_1 \) or \( S_2 \), and one request for a \( 2 \)-sum with label \( 2 \ast B \) from servers \( S_3, \dots, S_N \). The total number of consumed \( 1 \)-sums with label \( B \) is:
\[
2 \cdot \binom{N - 2}{1} \cdot M = 2(N - 2)M.
\]
Thus, the user also requests \( 2(N - 2)M \) portions of \( 2 \)-sums with label \( A + B \), and the same number with label \( 2 \ast B \).

To consume the \( 1 \)-sums with label \( R \) from servers \( S_3, \dots, S_N \), each with residual multiplicity \( M \), we apply recovery pattern~\(\zeta\). Under uniform application, each such \( 1 \)-sum requires one request of a \( 2 \)-sum with label \( A + B \), and one with label \( B + R \). From~\cref{table:type of summation}, the total number of consumed \( 1 \)-sums with label \( R \) is:
\[
(N - 2) \cdot \binom{N - 3}{1} \cdot M = (N - 3)(N - 2)M.
\]
Hence, the user additionally requests \( (N - 3)(N - 2)M \) portions of \( 2 \)-sums with labels \( A + B \) and \( B + R \), respectively.

At this point, all side information remaining after Step 1 has been consumed, and the number of times each recovery pattern applied in Step 2 is fully determined. We now compute the multiplicities and residual multiplicities of the \( 2 \)-sums. Note that the total number of requested portions of \( 2 \)-sums with label \( A + B \) is:
\[
2(N - 2)M + 2(N - 2)M + (N - 3)(N - 2)M = (N - 2)(N + 1)M.
\]
Each unit of multiplicity of \( 2 \)-sums contributes \( 2 \cdot \binom{N - 2}{1} = 2(N - 2) \) such summations. Therefore, to exactly provide the required portions of sum with label \( A + B \), we set the multiplicity of \( 2 \)-sums to be
\[
\frac{(N - 2)(N + 1)M}{2(N - 2)} = \frac{N + 1}{2}M.
\]
This ensures that the multiplicity provides the exact number of \( 2 \)-sums with label \( A + B \) required by recovery patterns~\(\alpha\),~\(\gamma\), and~\(\zeta\).

According to the recurrence in~\eqref{eq:case 1}, we have \( x_2 = \frac{N + 1}{2} \). Hence, the multiplicity of \( 2 \)-sums equals \( x_2 M \), satisfying property~\textup{(a)} in~\cref{prop:construction of KN}.

\begin{claim}
    The residual multiplicity of \( 2 \)-sums with each label after Step 2 satisfies the recurrence in~\eqref{eq:case 1} and properties~\textup{(b)} through~\textup{(f)} in~\cref{prop:construction of KN}.
\end{claim}

\begin{poc}
    The choice of the multiplicity \( x_2 M \) for \( 2 \)-sums guarantees that all \( 2 \)-sums with label \( A + B \) are fully consumed, satisfying property~\textup{(b)}.

    When applying recovery pattern~\(\gamma\), the user requests \( 2(N - 2)M \) portions of \( 2 \)-sums with label \( 2 \ast B \) from servers \( S_3, \dots, S_N \). These summations are not used in any other pattern. Since each unit of multiplicity provides
    \[
    (N - 2) \cdot \binom{N - 3}{0} = N - 2
    \]
    such summations, the residual multiplicity of this label is:
    \[
    x_2 M - \frac{2(N - 2)M}{N - 2} = x_2 M - 2M = \left( \frac{N + 1}{2} - 2 \right) M = \frac{N - 3}{2}M.
    \]
    Since \( y_2 = \frac{N - 3}{2} \), property~\textup{(d)} is satisfied.

    Likewise, under recovery pattern~\(\zeta\), the user requests \( (N - 3)(N - 2)M \) portions of \( 2 \)-sums with label \( B + R \). Each unit of multiplicity provides
    \[
   2 \cdot (N - 2) \cdot \binom{N - 3}{1} = 2(N - 2)(N - 3)
    \]
    such summations, so the residual multiplicity is:
    \[
    x_2 M - \frac{(N - 3)(N - 2)M}{2(N - 3)} = x_2 M - \frac{(N - 2)M}{2} = \left( \frac{N + 1}{2} - \frac{N - 2}{2} \right) M = \frac{3M}{2}.
    \]
    By \eqref{eq:case 1}, \( z_2 = \frac{3}{2} \), which  satisfies the property~\textup{(e)}.

    Finally, observe that \( 2 \)-sums with label \( 2 \ast B \) from \( S_1, S_2 \), and \( 2 \ast R \) from \( S_3, \dots, S_N \), are not involved in any recovery pattern at Step 2. Therefore, their residual multiplicities remain equal to the full multiplicity \( x_2 M \), which establishes property~\textup{(c)} and~\textup{(f)}.
\end{poc}

\vspace{0.5em}
So far, we have fully determined the multiplicities of all recovery patterns used in Step 2, and computed the resulting residual multiplicities for \( 2 \)-sums with each label. This completes the verification of Step 2.

We proceed by induction on \( k \) to construct and analyze the PIR scheme. The base cases \( k = 1 \) and \( k = 2 \) have already been established. Now, we consider the inductive step for a general \( 2 \le k \le N\!-\!1 \), and assume as the inductive hypothesis that Step \( k\!-\!1 \) satisfies properties~\textup{(a)} through~\textup{(f)} in~\cref{prop:construction of KN}. We also assume that the parameters \( x_k, y_k, z_k \) defined in the proposition satisfy the inequality~\eqref{eq:construction in K_N}.

As in previous steps, Step \( k \) consists of two parts: first, determining the multiplicities of the recovery patterns~\(\alpha\), \(\beta\), \(\gamma\), and \(\zeta\); and second, verifying that Step \( k \) satisfies the desired properties.

\paragraph{Step \( k \).} As in Step 2, the user successively applies recovery patterns~\(\alpha\), \(\beta\), \(\gamma\), and \(\zeta\) to consume the side information remaining from Step \( k\!-\!1 \). Each realization of every recovery pattern is applied uniformly, that is, the same number of times across all valid parameter choices, so it suffices to compute the total number of portions of \( k \)-sums with each label required by these patterns.

For a given \( k \), the number of portions of \( k \)-sums with each label contributed by one unit of multiplicity (that is, one complete copy of each summation type) is as follows, as derived from~\cref{table:type of summation}:

\begin{itemize}
    \item \( 2 \cdot \binom{N - 2}{k - 1} \) portions of \( k \)-sums with label \( A + (k\!-\!1) \ast B \) from \( S_1, S_2 \),
    \item \( 2 \cdot \binom{N - 2}{k} \) portions of \( k \)-sums with label \( k \ast B \) from \( S_1, S_2 \),
    \item \( (N - 2) \cdot \binom{N - 3}{k - 2} \) portions of \( k \)-sums with label \( 2 \ast B + (k\!-\!2) \ast R \) from \( S_3, \dots, S_N \),
    \item \(2 \cdot  (N - 2) \cdot \binom{N - 3}{k - 1} \) portions of \( k \)-sums with label \( B + (k\!-\!1) \ast R \) from \( S_3, \dots, S_N \),
    \item \( (N - 2) \cdot \binom{N - 3}{k} \) portions of \( k \)-sums with label \( k \ast R \) from \( S_3, \dots, S_N \).
\end{itemize}
Each expression above is obtained by multiplying the number of servers storing summations of the corresponding label by the number of distinct summation types of that label per server. To determine the total number of summations that must be consumed in Step \( k \), we simply multiply each of these values by the corresponding residual multiplicity inherited from Step \( k\!-\!1 \).

\paragraph{Recovery Pattern~\(\alpha\).}
By applying recovery pattern~\(\alpha\), the user consumes \((k\!-\!1)\)-sums with label \( B + (k\!-\!2) \ast R \) from servers \( S_3, \dots, S_N \). By the inductive hypothesis, the residual multiplicity of this label is \( z_{k - 1}M \). To consume \( k\!-\!1 \) such portions, the user must request one portion of a \( k \)-sum with label \( A + (k\!-\!1) \ast B \) from either \( S_1 \) or \( S_2 \).

The total number of consumed \((k\!-\!1)\)-sums with label \( B + (k\!-\!2) \ast R \) is
\[
(N - 2) \cdot 2 \cdot \binom{N - 3}{k - 2} \cdot z_{k - 1}M.
\]
Hence, the number of requested portions of \( k \)-sums with label \( A + (k\!-\!1) \ast B \) is
\begin{equation}
    \frac{(N - 2) \cdot 2 \cdot \binom{N - 3}{k - 2} \cdot z_{k - 1}M}{k - 1} = 2\binom{N - 2}{k - 1} z_{k - 1}M. \label{eq:recovery pattern alpha xk}
\end{equation}

\paragraph{Recovery patterns~\(\beta\) and~\(\gamma\).}
We now jointly consider recovery patterns~\(\beta\) and~\(\gamma\), as they interact through shared side information. These two patterns consume the following types of \((k\!-\!1)\)-sums:
\begin{itemize}
    \item \((k\!-\!1) \ast B\) from servers \( S_1, S_2 \), and
    \item \( 2 \ast B + (k\!-\!3) \ast R \) from servers \( S_3, \dots, S_N \).
\end{itemize}

By the inductive hypothesis, the residual multiplicities of these two labels are \( x_{k - 1}M \) and \( y_{k - 1}M \), respectively. The corresponding \( k \)-sums used in these patterns are:
\begin{itemize}
    \item \( A + (k\!-\!1) \ast B \) from \( S_1, S_2 \), and
    \item \( 2 \ast B + (k\!-\!2) \ast R \) from \( S_3, \dots, S_N \).
\end{itemize}

In recovery pattern~\(\beta\), to consume \( k\!-\!1 \) (or \( k\!-\!2 \)) portions of the \((k\!-\!1)\)-sum with label \( 2 \ast B + (k\!-\!3) \ast R \), the user consumes one portion of \((k\!-\!1) \ast B\). The assumption
\[
x_{k - 1} \geq \frac{k - 1}{2(N - k)} y_{k - 1}
\]
implies the inequality
\[
2 \cdot(k - 2) \cdot \binom{N - 2}{k - 1} \cdot x_{k - 1}M \geq (N - 2) \cdot \binom{N - 3}{k - 3} \cdot y_{k - 1}M,
\]
which ensures that the total number of available \((k\!-\!1)\)-sums with label \( (k\!-\!1) \ast B \) suffices to consume all \((k\!-\!1)\)-sums with label \( 2 \ast B + (k\!-\!3) \ast R \). Thus, recovery patterns~\(\beta\) and~\(\gamma\) can fully exhaust the side information with these two labels.

Moreover, each application of these patterns that consumes a portion of \( (k\!-\!1) \ast B \) also triggers a request for a \( k \)-sum with label \( A + (k\!-\!1) \ast B \). Since the total number of such \((k\!-\!1)\)-sums is
\[
2 \cdot \binom{N - 2}{k - 1} \cdot x_{k - 1}M,
\]
the total number of requested \( k \)-sums with label \( A + (k\!-\!1) \ast B \) from this part is
\begin{equation}
    2\binom{N - 2}{k - 1} \cdot x_{k - 1}M. \label{eq:recovery pattern beta xk}
\end{equation}

In addition, each consumed portion of \( (k\!-\!1) \ast B \) requires the user to request \( k\!-\!1 \) portions of \( k \)-sums with label \( 2 \ast B + (k\!-\!2) \ast R \), unless some can be reused from previous steps. Subtracting the requested summations originating from the label \( 2 \ast B + (k\!-\!3) \ast R \), the total number of requested portions of \( k \)-sums with label \( 2 \ast B + (k\!-\!2) \ast R \) is
\begin{equation}
    (k - 1) \cdot 2 \cdot \binom{N - 2}{k - 1} \cdot x_{k - 1}M - (N - 2) \cdot \binom{N - 3}{k - 3} \cdot y_{k - 1}M. \label{eq:recovery pattern beta yk}
\end{equation}

\paragraph{Recovery Pattern~\(\zeta\).}
Recovery Pattern~\(\zeta\) is used to consume \((k\!-\!1)\)-sums with label \( (k\!-\!1) \ast R \) from servers \( S_3, \dots, S_N \), whose residual multiplicity is \( x_{k-1}M \) by the inductive hypothesis. To consume one such portion, the user must request:
\begin{itemize}
    \item One portion of a \( k \)-sum with label \( A + (k\!-\!1) \ast B \) from \( S_1 \) or \( S_2 \),
    \item \( k - 1 \) portions of \( k \)-sums with label \( B + (k\!-\!1) \ast R \) from \( S_3, \dots, S_N \).
\end{itemize}

However, the available types of \( k \)-sums provide these two labels in a fixed ratio:
\[
\frac{(N - 2) \cdot 2 \cdot \binom{N - 3}{k - 1}}{2 \cdot \binom{N - 2}{k - 1}} = N - k - 1.
\]
This may be smaller than the required ratio \( k - 1 \), especially when \( k \) is close to \( N \), potentially limiting the feasibility of consuming all residual \((k\!-\!1)\)-sums with label \( (k\!-\!1) \ast R \).

The outcome depends on whether \( k < \left\lfloor \frac{N}{2} \right\rfloor + 2 \). We distinguish two cases.

\paragraph{Case 1: \( k \le \left\lfloor \frac{N}{2} \right\rfloor + 1 \).} 

In this case, all side information with label \( (k\!-\!1) \ast R \) can be consumed. The number of requested portions of \( k \)-sums with label \( B + (k\!-\!1) \ast R \) is
\begin{equation}
    (k - 1)(N - 2)\binom{N - 3}{k - 1}x_{k - 1}M. \label{eq:recovery pattern zeta zk 1}
\end{equation}
Correspondingly, the number of required portions of \( k \)-sums with label \( A + (k\!-\!1) \ast B \) is
\begin{equation}
    (N - 2)\binom{N - 3}{k - 1}x_{k - 1}M. \label{eq:recovery pattern zeta xk 1}
\end{equation}

\paragraph{Case 2: \( k \ge \left\lfloor \frac{N}{2} \right\rfloor + 2 \).} 

In this case, only part of the side information with label \( (k\!-\!1) \ast R \) can be consumed. The total number of consumed multiplicities is
\[
\frac{2}{2k - N}\cdot (x_{k - 1} + z_{k - 1})M,
\]
which corresponds to
\begin{equation}
    \frac{2(N - 2)}{2k - N}\cdot\binom{N - 3}{k - 1}(x_{k - 1} + z_{k - 1})M \label{eq:recovery pattern zeta consumed}
\end{equation}
portions of \((k\!-\!1)\)-sums with label \( (k\!-\!1) \ast R \). To consume this amount, the user must request
\begin{equation}
    (k - 1) \cdot \frac{2(N - 2)}{2k - N}\binom{N - 3}{k - 1}(x_{k - 1} + z_{k - 1})M \label{eq:recovery pattern zeta zk 2}
\end{equation}
portions of \( k \)-sums with label \( B + (k\!-\!1) \ast R \), and
\begin{equation}
    \frac{2(N - 2)}{2k - N}\binom{N - 3}{k - 1}(x_{k - 1} + z_{k - 1})M \label{eq:recovery pattern zeta xk 2}
\end{equation}
portions of \( k \)-sums with label \( A + (k\!-\!1) \ast B \).

\medskip
In both cases, we define the multiplicity of \( k \)-sums so that it suffices to provide the total required portions of \( k \)-sums with label \( A + (k\!-\!1) \ast B \) across recovery patterns \(\alpha,\beta,\gamma,\zeta\). More precisely, when \( k \le \left\lfloor \frac{N}{2} \right\rfloor + 1 \), the total number of requested portions with label \( A + (k\!-\!1) \ast B \) is:
\[
2\binom{N-2}{k-1}z_{k-1}M + 2\binom{N-2}{k-1}x_{k-1}M + (N - 2)\binom{N - 3}{k - 1}x_{k - 1}M.
\]
Dividing by \( 2\binom{N - 2}{k - 1} \), by~\eqref{eq:case 1}, the required multiplicity is
\[
\left( \frac{N - k + 1}{2}x_{k-1} + z_{k-1} \right)M = x_kM.
\]
When \( k \ge \left\lfloor \frac{N}{2} \right\rfloor + 2 \), the total number of required portions is:
\[
2\binom{N-2}{k-1}z_{k-1}M + 2\binom{N-2}{k-1}x_{k-1}M + \frac{2}{2k - N}\binom{N - 3}{k - 1}(x_{k - 1} + z_{k - 1})M.
\]
Dividing by \( 2\binom{N - 2}{k - 1} \), the required multiplicity is
\[
\frac{k - 1}{2k - N}(x_{k - 1} + z_{k - 1})M = x_kM,
\]
by~\eqref{eq:case 2}. Thus, in either case, the multiplicity of \( k \)-sums is set to \( x_k M \), satisfying property~\textup{(a)} in~\cref{prop:construction of KN}.

\begin{claim}
    The residual multiplicity of $k$-sums with each label after Step~$k$ satisfies~\eqref{eq:case 1} or~\eqref{eq:case 2}, as well as properties~\textup{(b)} through~\textup{(f)} in~\cref{prop:construction of KN}.
\end{claim}

\begin{poc}
    The definition of the $k$-sum multiplicity directly guarantees property~\textup{(b)}. Additionally, $k$-sums with label $k \ast B$ from \( S_1, S_2 \) and $k \ast R$ from \( S_3, \dots, S_N \) are never used in any of the recovery patterns $\alpha$, $\beta$, $\gamma$, or $\zeta$. Thus, their residual multiplicities remain equal to the full multiplicity $x_k M$, verifying properties~\textup{(c)} and~\textup{(f)}.

    We now compute the residual multiplicities for the other labels.

    \paragraph{Case 1:} \( k \le \left\lfloor \frac{N}{2} \right\rfloor + 1 \). From~\eqref{eq:recovery pattern beta yk}, the residual multiplicity of $k$-sums with label \( 2 \ast B + (k - 2) \ast R \) from \( S_{3},\dots,S_N \) is
    \[
    x_k M - \frac{(k - 1) \cdot 2 \cdot \binom{N - 2}{k - 1} \cdot x_{k-1} M - (N - 2) \cdot \binom{N - 3}{k - 3} \cdot y_{k-1} M}{(N - 2) \cdot \binom{N - 3}{k - 2}} = x_k M - 2x_{k-1}M + \frac{k - 2}{N - k} y_{k-1} M,
    \]
    which equals \( y_k M \) by~\eqref{eq:case 1}.

    Similarly, from~\eqref{eq:recovery pattern zeta zk 1}, the residual multiplicity of $k$-sums with label \( B + (k\!-\!1) \ast R \) from \( S_3,\dots,S_N \) is
    \[
    x_k M - \frac{(k - 1)(N - 2)\binom{N - 3}{k - 1} \cdot x_{k-1} M}{(N - 2) \cdot 2\binom{N - 3}{k - 1}} = \left(x_k  - \frac{k - 1}{2} x_{k-1}\right) M,
    \]
    which matches \( z_k M \) by~\eqref{eq:case 1}.

    \paragraph{Case 2:} \( k \ge \left\lfloor \frac{N}{2} \right\rfloor + 2 \). By the same computation from~\eqref{eq:recovery pattern beta yk}, the residual multiplicity of $k$-sums with label \( 2 \ast B + (k - 2) \ast R \) is
    \[
    x_k M - 2x_{k-1}M + \frac{k - 2}{N - k} y_{k-1} M = y_k M,
    \]
    consistent with~\eqref{eq:case 2}.

    From~\eqref{eq:recovery pattern zeta zk 2}, the residual multiplicity of $k$-sums with label \( B + (k\!-\!1) \ast R \) is
    \[
    x_k M - \frac{(k - 1) \cdot \frac{2(N - 2)}{2k - N} \cdot \binom{N - 3}{k - 1} \cdot (x_{k-1} + z_{k-1}) M}{(N - 2) \cdot 2 \binom{N - 3}{k - 1}} = \left(x_k - \frac{k - 1}{2k - N}(x_{k-1} + z_{k-1})\right) M,
    \]
    which equals \( z_k M \) by~\eqref{eq:case 2}.

    Hence, Step~$k$ satisfies properties~\textup{(d)} and~\textup{(e)}, completing the proof.
\end{poc}

\medskip

It is worth emphasizing that when \( k \ge \left\lfloor \frac{N}{2} \right\rfloor + 2 \), we have \( z_k = 0 \), which implies that no additional \( k \)-sums with label \( B + (k\!-\!1) \ast R \) are available for applying recovery pattern~\(\zeta\). In other words, the residual \((k\!-\!1)\)-sums with label \( (k\!-\!1) \ast R \) cannot be consumed further beyond this point.

So far, we have completed the construction of Step~$k$ for all \( 1 \le k \le N - 1 \). The union of all these steps yields a complete description of the scheme \( \Pi_N \), which (though not yet verified as a PIR scheme) satisfies properties~\textup{(a)} through~\textup{(f)} from~\cref{prop:construction of KN}, with parameters \( x_k, y_k, z_k \) satisfying the recursive relations~\eqref{eq:case 1} and~\eqref{eq:case 2}.

Observe that the values \( x_k, y_k, z_k \) may not be integers. However, the residual multiplicities \( x_k M, y_k M, z_k M \) must be integers. For this reason, we introduce a scaling factor \( M \), defined to be the smallest positive integer such that all of \( x_k M, y_k M, z_k M \) are integers for all \( k \).

Assuming that \( \Pi_N \) is indeed a valid PIR scheme, an assumption we will justify in the next subsection, we define the subpacketization of the scheme as
\begin{equation}
    L := 2 \sum_{k=1}^{N-1} \binom{N-2}{k-1} x_k M. \label{eq:expression subpacketization}
\end{equation}
This expression reflects the total number of recovery patterns involving sub-files of the desired file \( A \) used throughout \( \Pi_N \). Specifically, Step~\( k \) contributes \( 2 \binom{N-2}{k-1} x_k M \) such recovery patterns, each corresponding to one distinct sub-file of \( A \).

Consequently, every file in the system is divided into \( L \) equal-sized sub-files. Following the conventions in~\cref{subsection:PIR scheme k3 revisited} and~\cref{subsection:PIR scheme K4}, the user privately samples \( \binom{N}{2} \) independent random permutations \( \sigma_{i,j} \) over \( [L] \) for each file \( W_{i,j} \) stored on servers \( S_i \) and \( S_j \). For \( s \in [L] \), define the \( s \)-th sub-file of \( W_{i,j} \) to be
\[
(w_{i,j})_s := W_{i,j}(\sigma_{i,j}(s)).
\]

Using this subpacketization, we enumerate the \( L \) recovery patterns and their corresponding side information. Subscripts are assigned by recording, for each file \( W_{i,j} \), the number of times it has already appeared in earlier patterns. Specifically, if the sub-file of \( W_{i,j} \) appears in the \( n \)-th recovery pattern, and has appeared \( r \) times previously, then it is indexed as \( (w_{i,j})_{r+1} \).

This indexing method guarantees that, within any single recovery pattern, all sub-files belonging to the same file share the same subscript. As a result, when the summation corresponding to the recovery pattern is computed, the user can correctly isolate and recover the targeted sub-file from the corresponding recovery pattern.

As seen in~\cref{table:specific K4}, this rule ensures that for each server, all queried sub-files are distinct. In cases where a sub-file appears only as redundant side information, its index can be chosen arbitrarily, provided it does not conflict with previously assigned indices. This indexing scheme, together with the use of independent permutations \( \sigma_{i,j} \), guarantees the uniform randomness of all transmitted sub-files and ensures the structural correctness of the scheme \( \Pi_N \).

\subsubsection{Verification of the feasibility of the general construction}\label{subsubsection:Verification the Feasibility of Our Construction}

In~\cref{subsubsection:Description of Our Construction}, we presented a recursive construction of the scheme \( \Pi_N \), which satisfies properties~\textup{(a)} through~\textup{(f)} in~\cref{prop:construction of KN}. In this subsection, we aim to verify the feasibility of this construction.

Our verification proceeds in three parts. First, we show that the sequences \( \{x_k\}_{k=1}^{N-1} \), \( \{y_k\}_{k=2}^{N-1} \), and \( \{z_k\}_{k=1}^{N-2} \), defined by the recurrence relations in~\eqref{eq:case 1} and~\eqref{eq:case 2}, satisfy the constraints \( 0 \le y_k, z_k \le x_k \) and the key inequality in~\eqref{eq:construction in K_N}. This confirms the internal consistency and validity of the multiplicity parameters used in the construction.

Second, we verify that the only remaining side information in the scheme \( \Pi_N \) occurs in the form of \( k \)-sums with label \( k \ast R \) from servers \( S_3, \dots, S_N \), and this happens only for \( k \ge \left\lfloor \frac{N}{2} \right\rfloor + 2 \). This observation ensures that no further recovery patterns can be applied and justifies the termination of our recursive process at Step~\( N - 1 \).

Finally, we establish that \( \Pi_N \) is a valid PIR scheme by proving its correctness (reliability) and privacy. This confirms that our construction not only satisfies the structural properties but also achieves the desired PIR functionality.

We begin by verifying the nonnegativity and boundedness of the sequence \( \{z_k\} \) with respect to \( \{x_k\} \).

\begin{lemma}\label{lemma:zk xk}
    For any \( 1 \le k \le N - 1 \), the sequence \( \{x_k\}_{k=1}^{N-1} \) is strictly positive and monotonically non-decreasing. Moreover, for all \( 1 \le k \le N - 2 \), we have \( 0 \le z_k \le x_k \).
\end{lemma}

\begin{proof}[Proof of~\cref{lemma:zk xk}]
    We prove the two assertions in turn.

    \paragraph{ Behavior when $k\le \floor{\frac{N}{2}}+1$.} We proceed by induction to show that \( 0 < z_k \le x_k \).

    \emph{Base case:} When \( k = 1 \), we have \( x_1 = z_1 = 1 \), so the claim holds.

    \emph{Inductive step:} Suppose the statement holds for some \( k \le \left\lfloor \frac{N}{2} \right\rfloor + 1 \), i.e., \( 0 < z_{k-1} \le x_{k-1} \). Then by the recurrence~\eqref{eq:case 1},
    \begin{align*}
        z_k &= x_k - \frac{k - 1}{2}x_{k - 1} \\
            &= \left( \frac{N - k + 1}{2}x_{k - 1} + z_{k - 1} \right) - \frac{k - 1}{2}x_{k - 1} \\
            &= \frac{N - 2k + 2}{2}x_{k - 1} + z_{k - 1} \\
            &\ge z_{k - 1} > 0,
    \end{align*}
    where we used the assumption \( k \le \left\lfloor \frac{N}{2} \right\rfloor + 1 \), which implies \( N - 2k + 2 \ge 0 \), and the inductive hypothesis \( x_{k-1}, z_{k-1} > 0 \).

    Since \( z_k = x_k - \frac{k - 1}{2}x_{k - 1} \), it immediately follows that \( z_k \le x_k \). Thus, \( 0 < z_k \le x_k \) holds for all \( k \le \left\lfloor \frac{N}{2} \right\rfloor + 1 \) by induction.

    \paragraph{ Behavior when \( k \ge \left\lfloor \frac{N}{2} \right\rfloor + 2 \).} From~\eqref{eq:case 2}, we define
    \[
    x_k = \frac{k - 1}{2k - N}(x_{k - 1} + z_{k - 1}), \qquad z_k = 0.
    \]
    Since \( x_{k - 1} > 0 \) and \( z_{k - 1} \ge 0 \) by Step 1, this implies \( x_k > 0 \) and clearly \( z_k = 0 < x_k \), so the claim holds in this range as well.

    \paragraph{ Monotonicity of \( \{x_k\} \).} We show that \( x_k \ge x_{k - 1} \) for all \( 2 \le k \le N - 1 \).

    \emph{Case 1:} \( 2 \le k \le \left\lfloor \frac{N}{2} \right\rfloor + 1 \). From~\eqref{eq:case 1},
    \[
    x_k = \frac{N - k + 1}{2}x_{k - 1} + z_{k - 1} \ge \frac{N - k + 1}{2}x_{k - 1} \ge x_{k - 1},
    \]
    since \( \frac{N - k + 1}{2} \ge 1 \) and \( z_{k - 1} \ge 0 \).

    \emph{Case 2:} \( k \ge \left\lfloor \frac{N}{2} \right\rfloor + 2 \). Then from~\eqref{eq:case 2},
    \[
    x_k = \frac{k - 1}{2k - N}(x_{k - 1} + z_{k - 1}) \ge \frac{k - 1}{2k - N}x_{k - 1} \ge x_{k - 1},
    \]
    since \( \frac{k - 1}{2k - N} \ge 1 \) in this range and \( z_{k - 1} \ge 0 \). This completes the proof.
\end{proof}

Having verified the properties of the sequence \( \{z_k\} \) in~\cref{lemma:zk xk}, we now turn our attention to the sequence \( \{y_k\} \). Our goal is to show that each \( y_k \) lies within the range \([0, x_k]\), and that the inequality~\eqref{eq:construction in K_N} is satisfied. To this end, we derive an explicit expression for \( y_k \) in terms of \( x_{k-1} \) and \( x_k \), and use it to confirm all required properties.

\begin{lemma}\label{lemma:yk}
    For any $2\le k\le N-1$, we have $0\le y_k\le x_k$ and 
    \begin{equation}
        x_k\ge \frac{k\cdot y_{k}}{2(N-k-1)}. \label{eq:lemma:yk construction equation}
    \end{equation}
    Moreover, we have 
    \begin{equation}\label{eq:relation y_k x_k}
        y_k = 
        \begin{cases}
            \frac{N - k - 1}{2} x_{k-1}, & \text{if  } 2\le k \leq \floor{\frac{N}{2}} + 1, \\
            \frac{N - k - 1}{k - 1} x_k,  & \text{if  } \floor{\frac{N}{2}} + 2\le k\le N-1.
        \end{cases}
    \end{equation}
\end{lemma}
\begin{proof}[Proof of~\cref{lemma:yk}]
We now proceed to prove \eqref{eq:relation y_k x_k} by induction.

When $k = 2$, we have 
$$
y_2 = x_2 - 2x_1 = \frac{N - 3}{2}x_1,
$$
as desired. Now assume that for some $m$ with $3 \le m \le \left\lfloor \frac{N}{2} \right\rfloor+1$, the identity \eqref{eq:relation y_k x_k} holds for $k = m - 1$. We consider the case $k = m$:
\begin{align}
y_{m} &= x_m - 2x_{m-1} + \frac{m - 2}{N - m} y_{m-1} \notag \\
&= x_m - 2x_{m-1} + \frac{m - 2}{N - m} \cdot \frac{N - m}{2} x_{m-2} \label{eq:lemma:rationality 1} \\
&= \frac{N - m + 1}{2} x_{m-1} + z_{m-1} - 2x_{m-1} + \frac{m - 2}{2} x_{m-2} \label{eq:lemma:rationality 2} \\
&= \frac{N - m + 1}{2} x_{m-1} + x_{m-1} - \frac{m - 2}{2} x_{m-2} - 2x_{m-1} + \frac{m - 2}{2} x_{m-2} \label{eq:lemma:rationality 3} \\
&= \frac{N - m - 1}{2} x_{m-1},
\end{align}
where \eqref{eq:lemma:rationality 1} follows from the induction hypothesis, and \eqref{eq:lemma:rationality 2} and \eqref{eq:lemma:rationality 3} follow from \eqref{eq:case 1}. This completes the proof for the case $2 \le k \le \left\lfloor \frac{N}{2} \right\rfloor + 1$.

Now consider the case $k = \left\lfloor \frac{N}{2} \right\rfloor + 2$. Then,
\begin{align}
y_k &= x_k - 2x_{k-1} + \frac{k - 2}{N - k} y_{k-1} \notag \\
&= x_k - 2x_{k-1} + \frac{k - 2}{N - k} \cdot \frac{N - k}{2} x_{k-2} \label{eq:lemma:rationality 4} \\
&= \frac{k - 1}{2k - N} (x_{k-1} + z_{k-1}) - x_{k-1} - \left(x_{k-1} - \frac{k - 2}{2} x_{k-2} \right) \label{eq:lemma:rationality 5} \\
&= \frac{k - 1}{2k - N} (x_{k-1} + z_{k-1}) - x_{k-1} - z_{k-1} \label{eq:lemma:rationality 6} \\
&= \frac{N - k - 1}{2k - N} (x_{k-1} + z_{k-1}) \notag \\
&= \frac{N - k - 1}{k - 1} x_k, \label{eq:lemma:rationality 7}
\end{align}
where \eqref{eq:lemma:rationality 4} uses the fact that $k - 1 \le \left\lfloor \frac{N}{2} \right\rfloor + 1$; \eqref{eq:lemma:rationality 5} and \eqref{eq:lemma:rationality 7} follow from \eqref{eq:case 2}; and \eqref{eq:lemma:rationality 6} follows from \eqref{eq:case 1}. This completes the proof for the case $k = \left\lfloor \frac{N}{2} \right\rfloor + 2$.

Finally, assume that for some $m$ with $\left\lfloor \frac{N}{2} \right\rfloor + 2 \le m - 1 \le N - 1$, the identity \eqref{eq:relation y_k x_k} holds for $k = m - 1$. We now consider the case $k = m$:
\begin{align}
y_m &= x_m - 2x_{m-1} + \frac{m - 2}{N - m}\cdot y_{m-1} \notag \\
&= x_m - 2x_{m-1} + \frac{m - 2}{N - m} \cdot \frac{N - m}{m - 2} \cdot x_{m-1} \label{eq:lemma:rationality 8} \\
&= x_m - (x_{m-1} + z_{m-1}) + z_{m-1} \notag \\
&= x_m - (x_{m-1} + z_{m-1}) \label{eq:lemma:rationality 9} \\
&= x_m - \frac{2m - N}{m - 1}\cdot x_m \label{eq:lemma:rationality 10} \\
&= \frac{N - m - 1}{m - 1}\cdot x_m, \notag
\end{align}
where \eqref{eq:lemma:rationality 8} follows from the induction hypothesis, and \eqref{eq:lemma:rationality 9} and \eqref{eq:lemma:rationality 10} follow from \eqref{eq:case 2}. This completes the proof for the case $\left\lfloor \frac{N}{2} \right\rfloor + 2 \le k \le N - 1$.

Therefore, the proof of \eqref{eq:relation y_k x_k} is complete. This identity further implies that $y_k \ge 0$ for all $k$, since we have $x_k\ge 0$. For $2 \le k \le \left\lfloor \frac{N}{2} \right\rfloor + 1$, we have
\begin{align}
y_k &= \frac{N - k - 1}{2} x_{k-1} \notag \\
&= x_k - x_{k-1} + z_{k-1} \label{eq:lemma:rationality 11} \\
&\le x_k \le \frac{2(N - k-1)}{k} x_k, \notag
\end{align}
where \eqref{eq:lemma:rationality 11} follows from \eqref{eq:case 1}.

For $\left\lfloor \frac{N}{2} \right\rfloor + 2 \le k \le N - 1$, observe that $\frac{N - k - 1}{k - 1} \le 1$ and $\frac{N - k - 1}{k - 1} \le \frac{2(N - k-1)}{k}$. Since $y_k = \frac{N - k - 1}{k - 1} x_k$, this implies $y_k \le x_k$ and hence $x_k \ge \frac{k - 1}{2(N - k)} y_k$. This completes the proof.
\end{proof}

Observe that~\eqref{eq:relation y_k x_k} implies \( y_{N-1} = 0 \), which confirms that the only remaining unused side information in our scheme consists of \( k \)-sums with label \( k \ast R \) for all \( k \ge \left\lfloor \frac{N}{2} \right\rfloor + 2 \). These sub-files are never involved in any recovery pattern, and thus do not affect the correctness of the scheme.

We now verify that the proposed construction indeed constitutes a valid PIR scheme over the complete graph \( K_N \), by proving both privacy and reliability.

\begin{lemma}\label{lemma:KN privacy reliability}
    The scheme \( \Pi_N \) satisfies the privacy and reliability requirements of PIR.
\end{lemma}

\begin{proof}[Proof of~\cref{lemma:KN privacy reliability}]
We verify those properties as follows.

\paragraph{Privacy:} The privacy argument follows directly from~\cite[Lemma 3]{Sun2017capacity}. Specifically, it suffices to show that, for any \( i \in [N] \), the distribution of the query $\QQ_i$ sent to server \( S_i \) is independent of the index \( \theta \) of the desired file. We break this argument into two parts:

\begin{itemize}
    \item \emph{Structural independence:} Ignoring the subscripts of the requested sub-files, the structure of the query to each server depends solely on the pre-determined multiplicity sequence \( \{x_k\}_{k=1}^{N-1} \). As this sequence is independent of \( \theta \), the overall query structure is also independent of \( \theta \).
    \item \emph{Subscript independence:} For each single server, the subscripts of the referred sub-files of some specific file $W_{i,j}$ only need to be pairwise distinct. And the specific subscripts are determined using random permutations \( \sigma_{i,j} \), which are independently and uniformly chosen and do not depend on \( \theta \). Therefore, the joint distribution of subscripts of sub-files of $W_{i,j}$ referred in the query is independent of the identity of the desired file and only determined by the permutation \( \sigma_{i,j} \).
\end{itemize}

To elaborate, for any server \( S_i \), since the selection of subscripts of sub-files of distinct files is independent of each other, the joint distribution of the whole subscripts is only dependent on the permutations and is independent of the desired file. Moreover, the query consists of \( x_kM \) portions of \( k \)-sums of type \( W_{i,j_1} + \cdots + W_{i,j_k} \), for any distinct \( j_1, \dots, j_k \in [N] \setminus \{i\} \). This amount is also independent of the desired file. Hence, the joint distribution of the entire query to \( S_i \) is independent of \( \theta \), establishing privacy.

\paragraph{Reliability:} The reliability of the scheme follows by construction. The queries set \( \mathcal{Q} \) consists of \( L \) valid recovery patterns and some side information, where each pattern allows the recovery of a distinct sub-file of the desired file \( A \). As shown in~\eqref{eq:expression subpacketization}, the subpacketization \( L \) equals the total number of such recovery patterns. Since each pattern recovers a different sub-file and the entire file \( A \) is partitioned into \( L \) sub-files, the user can correctly reconstruct \( A \). This confirms the reliability of the scheme.
\end{proof}

So far, combining~\cref{lemma:zk xk},~\cref{lemma:yk}, and~\cref{lemma:calculate rate}, the scheme $\Pi_N$ is shown to satisfy~\cref{prop:construction of KN}.

\subsubsection{Estimation of the rate}\label{subsubsection:Estimation of the Rate}

In this subsection, we analyze the rate of the PIR scheme \( \Pi_N \) constructed above.

\begin{lemma}\label{lemma:calculate rate}
The rate of the PIR scheme \( \Pi_N \) over \( K_N \) is given by
\begin{equation}
R(\Pi_{N}) = \frac{2\sum_{k=1}^{N-1} \binom{N-2}{k-1} x_k}{N \sum_{k=1}^{N-1} \binom{N-1}{k} x_k} \ge \left(1-\epsilon_N\right)\frac{4}{3N}, \label{eq:expression rate}
\end{equation}
where \( \epsilon_N \to 0 \) as \( N \to \infty \).
\end{lemma}

\begin{proof}[Proof of~\cref{lemma:calculate rate}]
Recall from~\eqref{eq:expression subpacketization} that the subpacketization of our scheme is 
\[
L = 2\sum_{k=1}^{N-1} \binom{N-2}{k-1} x_k,
\]
and that each server sends
\[
\sum_{k=1}^{N-1} \binom{N-1}{k} x_k
\]
messages in response to the user's query. Since the total download is \( N \) times this quantity, the rate is given by
\[
R(\Pi_N) = \frac{L}{N \cdot \sum_{k=1}^{N-1} \binom{N-1}{k} x_k} = \frac{2\sum_{k=1}^{N-1} \binom{N-2}{k-1} x_k}{N \cdot\sum_{k=1}^{N-1} \binom{N-1}{k} x_k}.
\]

Define \( A_k := \binom{N-2}{k-1} x_k \) and \( B_k := \binom{N-1}{k} x_k \), so that
\[
R(\Pi_N) = \frac{2\sum_{k=1}^{N-1} A_k}{N\sum_{k=1}^{N-1} B_k}.
\]

We now estimate the behavior of the ratio. First, by~\eqref{eq:case 1} and~\eqref{eq:case 2}, we analyze the growth rate of the sequence \( \{A_k\} \) as follows.

From the recurrence relations, we have for \( k \le \left\lfloor \frac{N}{2} \right\rfloor + 1 \),
\[
\frac{x_k}{x_{k-1}} \ge \frac{N - k + 1}{2},
\]
which implies that
\[
\frac{A_k}{A_{k-1}} = \frac{\binom{N-2}{k-1} x_k}{\binom{N-2}{k-2} x_{k-1}}\ge \frac{(N-k)(N-k+1)}{2(k-1)}>1.
\]

For \( k \ge \left\lfloor \frac{N}{2} \right\rfloor + 2 \),
\[
\frac{x_k}{x_{k-1}} = \frac{k - 1}{2k - N},
\]
which implies that
\begin{equation}
\frac{A_k}{A_{k-1}} = \frac{\binom{N-2}{k-1} x_k}{\binom{N-2}{k-2} x_{k-1}} = \frac{N - k}{2k - N}. \label{eq:lem:calculate rate 1}
\end{equation}

By analyzing this expression, we find that \( \{A_k\} \) increases for \( k \le \lfloor\frac{2N}{3}\rfloor \) and decreases afterward. In particular, if we assume for simplicity that \( \frac{2N}{3} \in \mathbb{Z} \), then \( A_k \) attains its maximum at \( k = \frac{2N}{3} - 1 \) or \( k = \frac{2N}{3} \). To lower bound the rate, we estimate the total contribution of the dominant terms near the peak of \( A_k \). Let \( \epsilon > 0 \) be an arbitrarily small constant. Then, we claim:

\begin{claim}\label{claim:estimate rate}
    For any \( \epsilon > 0 \),
    \[
    \sum_{k \in [1, (2/3 - \epsilon)N] \cup [(2/3 + \epsilon)N, N-1]} A_k
    \le (N - 1) \cdot e^{-\frac{3}{2} \epsilon^2 N}\cdot A_{2N/3}. 
    \]
\end{claim}

\begin{poc}
Assume \( \epsilon N \in \mathbb{Z} \). For any \( k \in [1, (2/3 - \epsilon)N] \cup [(2/3 + \epsilon)N, N-1] \), we have
\[
A_k \le \max \left\{ A_{(2/3 - \epsilon)N},\, A_{(2/3 + \epsilon)N} \right\}.
\]
Let \( R := A_{2N/3} \). We estimate \( A_{(2/3 - \epsilon)N} \) as follows:
\begin{align}
A_{(2/3 - \epsilon)N} &= R\cdot \prod_{t=0}^{\epsilon N-1}\frac{A_{2N/3-t-1}}{A_{2N/3-t}} \notag \\
&= R \cdot \prod_{t=0}^{\epsilon N - 1} \frac{1 - 6t/N}{1 + 3t/N} \label{eq:lem:calculate rate 2} \\
&\le R \cdot \prod_{t=0}^{\epsilon N - 1} \left(1 - \frac{6t}{N}\right), \notag
\end{align}
where \eqref{eq:lem:calculate rate 2} follows from \eqref{eq:lem:calculate rate 1}. Taking logarithms:
\[
\ln A_{(2/3 - \epsilon)N} \le \ln R - \sum_{t=0}^{\epsilon N - 1} \frac{6t}{N} = \ln R - 3\epsilon(\epsilon N - 1) \le \ln R - 3\epsilon^2 N.
\]
Thus, we have
\[
A_{(2/3 - \epsilon)N} \le R \cdot e^{-3\epsilon^2 N}.
\]
Similarly,
\[
A_{(2/3 + \epsilon)N} \le R \cdot e^{-\frac{3}{2} \epsilon^2 N}.
\]
Therefore,
\[
\sum_{k \in [1, (2/3 - \epsilon)N] \cup [(2/3 + \epsilon)N, N-1]} A_k
\le (N - 1) \cdot R \cdot e^{-c \epsilon^2 N}
\]
for any \( c \le \min\left\{ \frac{3}{2},\, 3 \right\} \), which completes the proof.
\end{poc}
\cref{claim:estimate rate} means that when $N\rightarrow\infty$, then 
\[
\sum_{k \in [1, (2/3 - \epsilon)N] \cup [(2/3 + \epsilon)N, N-1]} A_k=o(A_{2N/3}).
\]
A similar argument shows that when $N\rightarrow\infty$, we have
\[
\sum_{k \in [1, (2/3 - \epsilon)N] \cup [(2/3 + \epsilon)N, N-1]} B_k = o(B_{2N/3}).
\]
Moreover, for all \( k \in \left[(2/3 - \epsilon)N, (2/3 + \epsilon)N\right] \), we have
\[
\frac{A_k}{B_k} = \frac{k}{N - 1} \ge \frac{2}{3} - \epsilon.
\]
Therefore, we can estimate
\begin{align*}
R(\Pi_N) &= \frac{2}{N} \cdot \frac{\sum_{k=1}^{N-1} A_k}{\sum_{k=1}^{N-1} B_k} \\
&= \frac{2}{N} \cdot \frac{\sum_{k=(2/3 - \epsilon)N}^{(2/3 + \epsilon)N} A_k + o(A_{2N/3})}{\sum_{k=(2/3 - \epsilon)N}^{(2/3 + \epsilon)N} B_k + o(B_{2N/3})} \\
&\ge \frac{2}{N} \cdot \left( \frac{(\frac{2}{3} - \epsilon)\sum B_k}{\sum B_k + o(B_{2N/3})} - o(1) \right) \\
&= \left( \frac{2}{3} - \epsilon - o(1) \right) \cdot \frac{2}{N}.
\end{align*}
Since \( \epsilon > 0 \) is arbitrary, we conclude that
\[
R(\Pi_N) \ge (1 - o(1)) \cdot \frac{4}{3N},
\]
as \( N \to \infty \), completing the proof.
\end{proof}

Combining~\cref{prop:construction of KN} and~\cref{lemma:calculate rate}, we construct a PIR scheme $\Pi_N$ with rate at least $(\frac{4}{3}-\epsilon_N)\frac{1}{N}$, which completes the proof of~\cref{thm:LowerBoundMain}.

Overall, we have recursively constructed a PIR scheme \( \Pi_N \) for the complete graph \( K_N \), achieving a rate of \( \frac{4}{3N} \). However, the explicit values of the parameters \( x_k \), \( y_k \), and \( z_k \) have not yet been specified. The following lemma provides a closed-form expression for \( x_k \); its proof is deferred to~\cref{section:specific formula of xk}. Once \( x_k \) is determined, the values of \( y_k \) and \( z_k \) follow immediately from the recursive relations \eqref{eq:case 1}, \eqref{eq:case 2}, and the identity in~\cref{lemma:yk}.

\begin{lemma}\label{lemma:specific formula of xk}
Let \( N \ge 3 \) be an integer, and let \( k_0 = \floor{\frac{N}{2}} + 2 \). Then, for \( 1 \le k \le k_0 - 1 \), we have
\begin{equation}
x_k = \sum_{j=0}^{k-1} \frac{(N+3)!}{2^j (N-j+3)!} (-1)^{k-j-1} \binom{k-1}{j}, \label{eq:explicit expression x_k case 1}
\end{equation}
and for \( k_0 \le k \le N - 1 \), we have
\begin{equation}
x_k = \left( \prod_{i = k_0}^{k} \frac{i - 1}{2i - N} \right) \cdot \left( \sum_{j=0}^{k_0 - 2} \frac{(N+3)!}{2^j (N-j+3)!} (-1)^{k_0 - j - 2} \cdot \frac{k_0 - j + 2}{2} \binom{k_0 - 2}{j} \right). \label{eq:explicit expression x_k case 2}
\end{equation}
\end{lemma}

\begin{rmk}
We now estimate the subpacketization required by the scheme \( \Pi_N \). According to the expression above, we have \( x_{N-1} = (N!)^{O(1)} \). Since the integer \( M \) is defined as the smallest positive number such that all \( x_k M \), \( y_k M \), and \( z_k M \) are integers, it follows from \eqref{eq:case 1}, \eqref{eq:case 2}, and~\cref{lemma:specific formula of xk} that \( M = O(2^{N/2}(N!)^3) = (N!)^{O(1)} \). Thus, using~\eqref{eq:expression subpacketization}, the subpacketization \( L \) satisfies
\begin{align*}
    L &= 2\sum_{k=1}^{N-1} \binom{N-2}{k-1} x_k M \\
      &\le 2^{N - 1} \cdot \max_k x_k \cdot M = (N!)^{O(1)},
\end{align*}
which grows super-exponentially with \( N \), making it impractical for large-scale applications.

To address this limitation, we apply~\cref{thm:transform to probabilistic} to transform the deterministic PIR scheme \( \Pi_N \) over \( K_N \), constructed in~\cref{prop:construction of KN}, into a probabilistic scheme with subpacketization reduced to \( 1 \), while preserving both correctness and privacy.

\end{rmk}

\begin{rmk}
Another interesting phenomenon arises when comparing our construction with previously known schemes for smaller \( N \). Specifically, for \( N = 3 \), the PIR scheme obtained from~\cref{table:specific K4} (by applying~\cref{lemma:graph contain}) is exactly a repetition of the optimal PIR scheme shown in~\cref{table:adjusted scheme over k3}. Moreover, for \( N = 4, 5, 6 \), the PIR scheme over \( K_{N-1} \), obtained from \( \Pi_N \) via~\cref{lemma:graph contain}, coincides precisely with the earlier-constructed scheme \( \Pi_{N-1} \) for \( K_{N-1} \). 

However, this structural self-consistency breaks down when \( N \ge 7 \). The discrepancy arises from a shift in the recursive relation at \( k = \lfloor N/2 \rfloor + 2 \), where \eqref{eq:case 2} replaces \eqref{eq:case 1}, introducing a fundamentally different growth pattern in the sequences \( x_k, y_k, z_k \). Consequently, the PIR scheme over \( K_{N-1} \), obtained by applying~\cref{lemma:graph contain} to \( \Pi_N \), no longer matches the recursive structure of \( \Pi_{N-1} \).

\end{rmk}

\section{Probabilistic PIR schemes}\label{section:probabilistic schemes}

In this section, we present several results concerning \emph{probabilistic} PIR schemes. In~\cref{subsection:transform}, we provide a general transformation that converts any deterministic PIR scheme over graph satisfying the \emph{independence property} into a probabilistic PIR scheme with subpacketization reduced to~$1$. In~\cref{subsection:Probabilistic General Graph}, we propose a universal construction of probabilistic PIR schemes applicable to arbitrary graph-based storage systems.

\subsection{Generating probabilistic PIR schemes from the deterministic one}\label{subsection:transform}

We then prove~\cref{thm:transform to probabilistic} by establishing a transformation from any deterministic PIR scheme to a corresponding probabilistic PIR scheme. The proof can be mainly divided into two parts:

\begin{itemize}
    \item First, we show that for any deterministic PIR scheme satisfying the independence property, its queries can be decomposed into a collection of recovery patterns, possibly accompanied by a small amount of residual side information.
    
    \item Second, we demonstrate that any deterministic PIR scheme admitting such a decomposition can be naturally converted into a probabilistic PIR scheme with subpacketization equal to~$1$.
\end{itemize}

\subsubsection{Generating recovery patterns}
In this subsection, we focus on the procedure for extracting recovery patterns from a given deterministic PIR scheme. The key idea relies on a structural property that is commonly satisfied by many existing PIR schemes, which we formalize in the following definition.

\begin{defn}[Independence Property]\label{def:independence}
Let \( G \) be a (multi-)graph with \( N \) vertices. Let \( \Pi \) be a deterministic PIR scheme over \( G \), and let \( Q_i \) denote the query sent to server \( S_i \), with the full queries \( \QQ = \{ Q_1, \dots, Q_N \} \). We say that the PIR scheme \( \Pi \) satisfies the \emph{independence property} if the following conditions hold:
\begin{itemize}
    \item[\textup{(1)}] For every \( i \in [N] \), each summation in \( Q_i \) is a sum of sub-files from distinct files stored on server \( S_i \).
    
    \item[\textup{(2)}] For every \( i \in [N] \) and every file \( W \) stored on server \( S_i \), each sub-file of \( W \) appears at most once among all summations in \( Q_i \).
    
    \item[\textup{(3)}] For every requested file \( W_\theta \), where \( \theta \in E(G) \), each sub-file of \( W_\theta \) appears exactly once across all queries \( \QQ \).
    
    \item[\textup{(4)}] For every \( \theta \in E(G) \), each sub-file of the desired file \( W_\theta \) can be recovered using at most one requested summation from each server.
\end{itemize}
\end{defn}

\begin{rmk}
    In fact, all existing graph-based PIR schemes satisfy the independence property. Consequently, provided the condition $H(A_i) \le L$ holds for all $i \in [N]$, these schemes can be converted into corresponding probabilistic schemes via~\cref{thm:transform to probabilistic}.
\end{rmk}

Next, we present an algorithm to extract the recovery patterns from a deterministic PIR scheme~\(\Pi\) that satisfies the independence property. Without loss of generality, assume that the subpacketization of \(\Pi\) equals \(L\). To facilitate the description of the algorithm, we begin with the following definition.

\begin{defn}[Support of a Summation]\label{def:Support of Summation}
    Let \(W_{i,j_1}, W_{i,j_2}, \dots, W_{i,j_k}\) be \(k\) distinct files stored on server \(S_i\), where \(k\) is a positive integer and \(j_1, j_2, \dots, j_k \in [N] \setminus \{i\}\) are distinct indices. Suppose \(\Pi\) is a deterministic PIR scheme with subpacketization \(L\). For any integers \(s_1, s_2, \dots, s_k \in [L]\) and any permutations \(\sigma_1,\sigma_2,\dots,\sigma_k\) of $[L]$, define the \emph{support} of the summation
    \[
    \omega := (W_{i,j_1})_{\sigma_1(s_1)} + (W_{i,j_2})_{\sigma_2(s_2)} + \dots + (W_{i,j_k})_{\sigma_k(s_k)}
    \]
    as the set
    \[
    \mathrm{supp}(\omega) := \{ (W_{i,j_1})_{\sigma_1(s_1)}, (W_{i,j_2})_{\sigma_2(s_2)},\dots ,(W_{i,j_k})_{\sigma_k(s_k)} \}.
    \]
\end{defn}
For example, in the PIR scheme shown in~\cref{table:adjusted scheme over k3}, server~\(S_1\) sends the summation \(a_2 + b_2\), whose support is the set \(\{a_2, b_2\}\).

In the deterministic PIR scheme~\(\Pi\), the user privately generates a collection of permutations of \([L]\). In particular, let \(\sigma_{\theta}\) be the permutation associated with the desired file \(W_{\theta}\); then we write \(a_i := (W_{\theta})_{\sigma_{\theta}(i)}\) to denote the \(\sigma_{\theta}(i)\)-th sub-file of \(W_{\theta}\).

Algorithm~\ref{alg:generation of recovery patterns} provides a method to extract recovery patterns from a given deterministic PIR scheme~\(\Pi\) that satisfies the independence property. For simplicity, we write \(\omega \in \QQ\) if a summation~\(\omega\) appears in the query set~\(\QQ\). Furthermore, when referring to a summation~\(\omega\), we consider not only its algebraic expression but also the specific query~\(\QQ_i\) in which it appears.

\begin{algorithm}[h]
\caption{Generation of Recovery Patterns for Deterministic PIR Scheme Satisfying the Independence Property}
\label{alg:generation of recovery patterns}
\KwIn{Query set \(\QQ = \{\QQ_1, \dots, \QQ_N\}\) of a deterministic PIR scheme \(\Pi\) with subpacketization \(L\), satisfying the independence property; index \(\theta\) of the desired file.}
\KwOut{The \(L\) recovery patterns \(\{\PP^1, \dots, \PP^L\}\) and the remaining side information \(\mathcal{D}\).}

\(\mathcal{D} \leftarrow \varnothing\) \tcp*[f]{Set of side information} \\

\For(\tcp*[f]{For each desired sub-file \(a_j\)}){\(j \leftarrow 1\) \KwTo \(L\)}{
    \(\PP^j \leftarrow \varnothing\) \tcp*[f]{Set of recovery pattern} \\
    \(\mathcal{S}^j \leftarrow \{a_j\}\) \tcp*[f]{Set of support of summations in \(\PP^j\)} \\
    \While{there exists \(\omega \in \QQ\) such that \(\mathrm{supp}(\omega) \cap \mathcal{S}^j \ne \varnothing\)}{
        \(\PP^j \leftarrow \PP^j \cup \{\omega\}\) \tcp*[f]{Add summation to recovery pattern} \\
        \(\mathcal{S}^j \leftarrow \mathcal{S}^j \cup \mathrm{supp}(\omega)\)
    }
    Output \(\PP^j\)
}
\(\mathcal{D} \leftarrow \QQ \setminus \bigcup_{j=1}^L \PP^j\) \tcp*[f]{Remaining side information} \\
Output \(\mathcal{D}\)
\end{algorithm}

Before formally proving the correctness of Algorithm~\ref{alg:generation of recovery patterns}, let us first consider a concrete example. Consider a replication system over the 4-star graph consisting of five servers \( S_1, \dots, S_5 \) and four files \( A, B, C, D \).~\cref{fig:4 star} illustrates the relationship between servers and files.

\begin{figure}[hbtp]
    \centering
    \begin{tikzpicture}
        \draw[](0,0)--(2,0);
        \draw[](0,0)--(-2,0);
        \draw[](0,0)--(0,2);
        \draw[](0,0)--(0,-2);
        \fill(0,0)circle(.1);
        \fill(2,0)circle(.1);
        \fill(-2,0)circle(.1);
        \fill(0,2)circle(.1);
        \fill(0,-2)circle(.1);
        \node at (-2,0)[above]{$S_3$};
        \node at (0.3,0)[above]{$S_1$};
        \node at (2,0)[above]{$S_4$};
        \node at (0,2)[right]{$S_2$};
        \node at (0,-2)[right]{$S_5$};
        \node at (0,1)[right]{$A$};
        \node at (-1,0)[above]{$B$};
        \node at (1,0)[above]{$C$};
        \node at (0,-1)[right]{$D$};
    \end{tikzpicture}
    \caption{The replication system based on the 4-star graph.}
    \label{fig:4 star}
\end{figure}
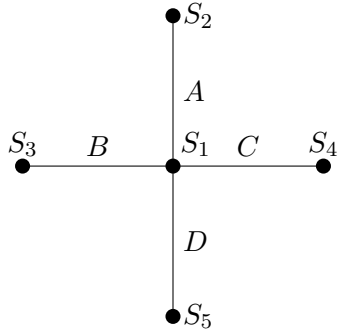

Assume each file is divided into \( L = 5 \) bits, and the user wants to privately retrieve file \( A \). To achieve this, the user privately generates four independent random permutations of \( [L] \), denoted by \( \sigma_i \) for \( i \in [4] \). For each \( i \in [L] \), define \( a_i := A(\sigma_1(i)) \) to be the \( \sigma_1(i) \)-th bit of \( A \). Similarly, define \( b_i := B(\sigma_2(i)) \), \( c_i := C(\sigma_3(i)) \), and \( d_i := D(\sigma_4(i)) \) to be the permuted sub-files of \( B, C \), and \( D \), respectively. An optimal PIR scheme with rate \( \frac{5}{12} \) and subpacketization \( L = 5 \) is given in~\cite{Yao2023star}, as shown in~\cref{table:4 star}.

\begin{table}[htbp]
\centering
\begin{tabular}{|c|c|c|c|c|}
\hline
$S_1$ & $S_2$ & $S_3$ & $S_4$ & $S_5$ \\ \hline
$a_1 + b_1 + c_1$ & $a_4$ & $b_1$ & $c_1$ & $d_1$ \\
$a_2 + b_2 + d_1$ & $a_5$ & $b_2$ & $c_2$ & $d_2$ \\
$a_3 + c_2 + d_2$ &       &       &       &       \\
$b_3 + c_3 + d_3$ &       &       &       &       \\ \hline
\end{tabular}
\caption{Answer table of the optimal PIR scheme over the 4-star graph in~\cite{Yao2023star}, for $\theta = A$.}
\label{table:4 star}
\end{table}

It can be easily verified that the PIR scheme shown in~\cref{table:4 star} satisfies the independence property. Applying Algorithm~\ref{alg:generation of recovery patterns} for \(j = 1\), we initialize \(\mathcal{S}^1 \leftarrow \{a_1\}\). The only summation involving \(a_1\) is \(a_1 + b_1 + c_1\) from server \(S_1\). Thus, we update \(\PP^1 \leftarrow \{a_1 + b_1 + c_1\}\) and \(\mathcal{S}^1 \leftarrow \{a_1, b_1, c_1\}\).

Next, we search for any summations in \(\QQ\) whose supports intersect \(\mathcal{S}^1\). We find \(b_1\) from \(S_3\) and \(c_1\) from \(S_4\). Adding them, we get \(\PP^1 \leftarrow \{a_1 + b_1 + c_1, b_1, c_1\}\). No other summation shares support with \(\mathcal{S}^1\), so the process for \(j=1\) terminates. 

By repeating this process, we can obtain $5$ recovery patterns and remaining side information. By grouping the requested summations into recovery patterns, we obtain an adjusted version of the PIR scheme, shown in~\cref{table:4 star adjusted}.

\begin{table}[htbp]
\centering
\begin{tabular}{|c|c|c|c|c|c|}
\hline
 & $S_1$ & $S_2$ & $S_3$ & $S_4$ & $S_5$ \\ \hline
Recovery Pattern 1 & $a_1 + b_1 + c_1$ &  & $b_1$ & $c_1$ &  \\ \hline
Recovery Pattern 2 & $a_2 + b_2 + d_1$ &  & $b_2$ &       & $d_1$ \\ \hline
Recovery Pattern 3 & $a_3 + c_2 + d_2$ &  &       & $c_2$ & $d_2$ \\ \hline
Recovery Pattern 4 &                   & $a_4$ &       &       &       \\ \hline
Recovery Pattern 5 &                   & $a_5$ &       &       &       \\ \hline
Side Information   & $b_3 + c_3 + d_3$ &       &       &       &       \\ \hline
\end{tabular}
\caption{Adjusted answer table of the PIR scheme over the 4-star graph, for $\theta = A$.}
\label{table:4 star adjusted}
\end{table}

Next, we formally verify the correctness and feasibility of the algorithm.

\begin{lemma}\label{lemma:Behavior of Algorithm 1}
If a deterministic PIR scheme \(\Pi\) satisfies the independence property, then for each \(j \in [L]\), the output \(\PP^j\) generated by Algorithm~\ref{alg:generation of recovery patterns} is a valid recovery pattern for the desired file \(W_\theta\).
\end{lemma}

\begin{proof}[Proof of~\cref{lemma:Behavior of Algorithm 1}]
To verify that \(\PP^j\) is a recovery pattern, we need to show the following:
\begin{itemize}
    \item \(\PP^j\) includes at most one summation from each query set \(\QQ_i\) for all \(i \in [N]\);
    \item The sum of all summations in \(\PP^j\) recovers the sub-file \(a_j\).
\end{itemize}

For the first statement, by property~(3) of the independence property in~\cref{def:independence}, there exists exactly one summation in \( \QQ \), denoted \( \omega_j \), such that \( a_j \in \mathrm{supp}(\omega_j) \). Therefore, due to the reliability of \( \Pi \), the user must be able to recover \( a_j \) using \( \omega_j \) together with some additional summations. Since sub-files behave like distinct free variables, recovering a target sub-file requires canceling out all other interfering sub-files. To cancel any non-desired sub-file that appears in \( \omega_j \), another summation containing that sub-file must exist.

By the structure of the replication system and property (2) of~\cref{def:independence}, each non-desired sub-file appears in at most two summations across all queries (because each file is stored on exactly two servers). Hence, for any sub-file involved in \(\omega_j \setminus \{a_j\}\), there exists at most one additional summation that includes it. Thus, to cancel all interfering sub-files in \(\omega_j\), the algorithm iteratively adds other summations that share these sub-files, updating the support set \(\mathcal{S}^j\) in the process. If the user can indeed recover \(a_j\), then all these added summations are necessary for the recovery, and so all the summations in set \(\PP^j\) must be included in the required recovery pattern. By property (4) in~\cref{def:independence}, this implies that each \(\PP^j\) includes at most one summation from each \(\QQ_i\).

We now verify the second statement: the sum of all summations in \(\PP^j\) equals \(a_j\). Suppose not. Then there must be some sub-file in \(\mathcal{S}^j \setminus \{a_j\}\) that is not canceled out in the total sum of summations in \(\PP^j\). However, by construction (line 5 of Algorithm~\ref{alg:generation of recovery patterns}), there is no other summation in \(\QQ \setminus \PP^j\) whose support intersects \(\mathcal{S}^j\). This would contradict the assumption that \(\Pi\) is a reliable scheme. Hence, the summations in \(\PP^j\) must sum to \(a_j\).
\end{proof}

\begin{lemma}\label{lemma:recovery pattern partition}
Let \(\Pi\) be a deterministic PIR scheme satisfying the independence property. Then, the query set \(\QQ\) of \(\Pi\) can be partitioned into the recovery patterns \(\PP^1, \PP^2, \dots, \PP^L\) and the remaining side information \(\mathcal{D}\).
\end{lemma}

\begin{proof}[Proof of~\cref{lemma:recovery pattern partition}]
Let \( \PP^1, \dots, \PP^L \) and their corresponding support sets \( \mathcal{S}^1, \dots, \mathcal{S}^L \) be the outputs of Algorithm~\ref{alg:generation of recovery patterns}. We need to show that the recovery patterns \( \PP^1, \dots, \PP^L \) are pairwise disjoint.

Assume, for the sake of contradiction, that there exists a summation \( \omega \in \PP^1 \cap \PP^2 \). Then, since \( \omega \) appears in both patterns, we must have \( \supp(\omega) \subseteq \mathcal{S}^1 \cap \mathcal{S}^2 \), implying \( \mathcal{S}^1 \cap \mathcal{S}^2 \ne \varnothing \). We now establish the following claim.

\begin{claim}\label{claim:P1P2}
\(\PP^1 = \PP^2\).
\end{claim}

\begin{poc}
We show that \(\PP^1 \subseteq \PP^2\); the reverse inclusion follows symmetrically.

By the construction in Algorithm~\ref{alg:generation of recovery patterns} (lines 5–8), the pattern \(\PP^1\) is built by successively including all summations whose supports intersect with the growing support set, starting from \(\mathcal{S}^1 = \{a_1\}\). Thus, there exists a chain of summations \(\omega_1, \omega_2, \dots, \omega_m = \omega \in \PP^1\), such that:
\begin{itemize}
    \item \(a_1 \in \supp(\omega_1)\),
    \item \(\supp(\omega_i) \cap \supp(\omega_{i+1}) \ne \varnothing\) for all \(i \in [m-1]\).
\end{itemize}
Since \(\omega = \omega_m \in \PP^2\), and Algorithm~\ref{alg:generation of recovery patterns} applies the same support-propagation rule, it follows inductively that all previous \(\omega_i\) must also be included in \(\PP^2\). In particular, \(\omega_1 \in \PP^2\), so \(a_1 \in \mathcal{S}^2\). As a result, by the construction of the algorithm, all of \(\PP^1\) must be contained in \(\PP^2\). Hence, \(\PP^1 \subseteq \PP^2\). This completes the proof.
\end{poc}

However, \cref{claim:P1P2} leads to a contradiction. Since both \( a_1 \in \mathcal{S}^1 \) and \( a_2 \in \mathcal{S}^2 \), and \( \PP^1 = \PP^2 \), it follows that a single recovery pattern \( \PP^1 \) contains two distinct sub-files of the desired file. This contradicts Property~(3) of the independence property in~\cref{def:independence}, which requires that each sub-file of the desired file appears in exactly one summation in \( \QQ \). As a result, \( a_1 \) cannot be recovered as the sum of all summations in \( \PP^1 \), leading to a contradiction.

Therefore, the recovery patterns \(\PP^1, \dots, \PP^L\) must be pairwise disjoint. It follows that the query set \(\QQ\) can be partitioned as
\[
\QQ = \left(\bigsqcup_{j = 1}^L \PP^j\right) \sqcup \mathcal{D},
\]
where \(\mathcal{D} := \QQ \setminus \left(\bigcup_{j=1}^L \PP^j\right)\) consists of side information that is not involved in the recovery of the desired sub-files. This finishes the proof of~\cref{lemma:recovery pattern partition}.
\end{proof}

\subsubsection{Relation between deterministic PIR schemes and probabilistic PIR schemes}
Recall from~\cref{def:Recovery Pattern} that a recovery pattern is defined as a collection containing at most one summation from each query set \(\QQ_i\). For a recovery pattern \(\mathcal{P}\), if the request to server \(S_i\) is represented by a (possibly empty) set \(q_i\), then we write \(\mathcal{P} = (q_1, q_2, \dots, q_N)\), or more compactly, \(\mathcal{P} = (q_i)_{i \in [N]}\). For instance, in~\cref{table:adjusted scheme over k3}, there exists a recovery pattern \(\mathcal{P}\) where the user requests \(\{a_2 + b_2\}\) from \(S_1\), nothing from \(S_2\), and \(\{b_2\}\) from \(S_3\). This pattern is denoted by \(\mathcal{P} = (\{a_2 + b_2\}, \varnothing, \{b_2\})\).

To formally describe how recovery patterns cover or partition the query sets in a PIR scheme, we introduce the following definitions of union and intersection for recovery patterns.

\begin{defn}
Let \(\Pi\) be a deterministic PIR scheme with query sets \(\QQ = \{\QQ_1, \dots, \QQ_N\}\). Suppose \(\{\mathcal{P}^j = (q_i^j)_{i \in [N]}\}_{j \in \Lambda}\) is a family of recovery patterns indexed by a set \(\Lambda\). The \emph{union} of these recovery patterns is defined as
\[
\bigcup_{j \in \Lambda} \mathcal{P}^j := \left( \bigcup_{j \in \Lambda} q_i^j \right)_{i \in [N]}.
\]
That is, for each server \(S_i\), the union of requests consists of all summations requested from \(S_i\) across the entire family of recovery patterns.

Similarly, the \emph{intersection} of recovery patterns is defined as
\[
\bigcap_{j \in \Lambda} \mathcal{P}^j := \left( \bigcap_{j \in \Lambda} q_i^j \right)_{i \in [N]}.
\]

We say that the query set \(\QQ = (\QQ_1, \dots, \QQ_N)\) is \emph{partitioned} by the recovery patterns \(\{\mathcal{P}^j\}_{j \in \Lambda}\) if, for every \(i \in [N]\), the query set \(\QQ_i\) is the disjoint union of the corresponding requests:
\[
\QQ_i = \bigsqcup_{j \in \Lambda} q_i^j.
\]
\end{defn}

As an example, the queries in the PIR scheme shown in~\cref{table:adjusted scheme over k3} can be partitioned into six disjoint recovery patterns.

The following lemma formalizes the connection between deterministic and probabilistic PIR schemes. Specifically, we consider a deterministic PIR scheme with subpacketization level \(L\), where the queries are partitioned into \(L\) recovery patterns—each enabling recovery of a distinct sub-file of the desired file using at most one summation per server.

\begin{lemma}\label{lemma:transform probabilistic}
Let $\Pi$ be a deterministic PIR scheme over a graph $G$ with subpacketization $L$ and rate $R$, where every summation in the queries involves sub-files from distinct files. Suppose that, for each server $S_i$, the entropy $H(A_i)$ of its response equals the number of summations in $A_i$. If the query set $\QQ$ across all servers can be partitioned into $L$ recovery patterns, then there exists a probabilistic PIR scheme $\Pi'$ over $G$ with subpacketization $1$ and the same rate $R$.
\end{lemma}

\begin{proof}[Proof of~\cref{lemma:transform probabilistic}]
Assume the $L$ recovery patterns are denoted by $\PP^j = (q_i^j)_{i \in [N]}$ for $j \in [L]$. For any file $W$, let $(W)_s$ denote its $s$-th sub-file. For each $i$ and $j$, suppose $q_i^j$ corresponds to a summation of $k(i,j)$ sub-files from $k$ distinct files. Then:
\begin{itemize}
    \item If $k(i,j) = 0$, then $q_i^j = \varnothing$.
    \item If $k(i,j) \ge 1$, there exist $t_1, \dots, t_k \in [N] \setminus \{i\}$, $s_1, \dots, s_k \in [L]$, and permutations $\sigma_1, \dots, \sigma_k$ of $[L]$ such that
    \[
    q_i^j = \left\{ (W_{i,t_1})_{\sigma_1(s_1)} + \cdots + (W_{i,t_k})_{\sigma_k(s_k)} \right\}.
    \]
\end{itemize}

We now construct a probabilistic PIR scheme $\Pi'$ with subpacketization $1$. The user selects an index $j \in [L]$ uniformly at random and sends queries according to the recovery pattern $\PP^j$, with the following simplification:
\begin{itemize}
    \item If $k(i,j) = 0$, the user sends a null query (or requests zero) to server $S_i$.
    \item If $k(i,j) \ge 1$, the user requests $W_{i,t_1} + \cdots + W_{i,t_k}$ from $S_i$.
\end{itemize}

This replaces sub-file summations with corresponding full-file summations, using the structure of $\PP^j$. Let $A_i$ and $A_i'$ denote the responses from $S_i$ in schemes $\Pi$ and $\Pi'$, respectively.

\begin{claim}
$\Pi'$ is a valid probabilistic PIR scheme.
\end{claim}

\begin{poc}
We verify both correctness and privacy.

\paragraph{Correctness.} Each recovery pattern $\PP^j$ in $\Pi$ recovers one sub-file of the desired file. In $\Pi'$, the user retrieves the corresponding full-file summation. Since the summation structure is preserved, the user can recover the desired file by applying the randomly chosen recovery pattern.

\paragraph{Privacy.} For any server $S_i$, the response $A_i'$ is independent of the desired file index $\theta$. Suppose $A_i' = W_{i,t_1} + \cdots + W_{i,t_k}$. Let $N(W_{i,t_1} + \cdots + W_{i,t_k})$ be the number of recovery patterns $\PP^j$ in which $S_i$ returns a $k$-sum of type $W_{i,t_1} + \cdots + W_{i,t_k}$ in the deterministic PIR scheme $\Pi$. Then,
\[
\Pr(A_i' = W_{i,t_1} + \cdots + W_{i,t_k}) = \frac{1}{L} \cdot N(W_{i,t_1} + \cdots + W_{i,t_k}).
\]
Let $M(W_{i,t_1} + \cdots + W_{i,t_k})$ be the number of times this $k$-sum appears in $A_i$ of scheme $\Pi$. Since the recovery patterns partition the queries,
\[
N(W_{i,t_1} + \cdots + W_{i,t_k}) = M(W_{i,t_1} + \cdots + W_{i,t_k}).
\]
Thus,
\[
\Pr(A_i' = W_{i,t_1} + \cdots + W_{i,t_k}) = \frac{1}{L} \cdot M(W_{i,t_1} + \cdots + W_{i,t_k}),
\]
which is independent of $\theta$, as required by privacy.

Now consider the case $A_i' = 0$. This occurs precisely when $q_i^j = \varnothing$ in the chosen recovery pattern $\PP^j$. Since $A_i$ contains $H(A_i)$ summations, the number of patterns in which $S_i$ is not queried is $L - H(A_i)$, yielding
\begin{equation}\label{eq:lemma:transform probabilistic 1}
\Pr(A_i' = 0) = \frac{L - H(A_i)}{L},
\end{equation}
which is also independent of $\theta$.
\end{poc}

\paragraph{Rate.} Finally, we show that $\Pi'$ has the same rate as $\Pi$:
\begin{align}
R(\Pi') &= \frac{1}{\sum_{i=1}^N H(A_i')} \notag \\
&= \frac{1}{\sum_{i=1}^N (1 - \Pr(A_i' = 0))} \label{eq:lemma:transform probabilistic 2} \\
&= \frac{1}{\sum_{i=1}^N \left(1 - \frac{L - H(A_i)}{L} \right)} \label{eq:lemma:transform probabilistic 3} \\
&= \frac{L}{\sum_{i=1}^N H(A_i)} = R(\Pi), \notag
\end{align}
where~\eqref{eq:lemma:transform probabilistic 2} computes the expected answer length, and~\eqref{eq:lemma:transform probabilistic 3} uses~\eqref{eq:lemma:transform probabilistic 1}.
\end{proof}

The PIR scheme over $K_3$ presented in~\cref{table:adjusted scheme over k3} serves as a simple example illustrating the transformation from a deterministic PIR scheme to its probabilistic counterpart. The corresponding probabilistic PIR scheme, obtained by applying~\cref{lemma:transform probabilistic}, is shown in~\cref{table:transform k3}.

\begin{table}[ht]
\centering
\begin{tabular}{|c|c|c|c|}
\hline
Probability & $S_1$ & $S_2$ & $S_3$ \\ \hline
$1/6$ & $A$ & $0$ & $0$ \\ \hline
$1/6$ & $0$ & $A$ & $0$ \\ \hline
$1/6$ & $A + B$ & $0$ & $B$ \\ \hline
$1/6$ & $0$ & $A + C$ & $C$ \\ \hline
$1/6$ & $A + B$ & $C$ & $B + C$ \\ \hline
$1/6$ & $B$ & $A + C$ & $B + C$ \\ \hline
\end{tabular}
\caption{A probabilistic PIR scheme over $K_3$, for $\theta=A$, obtained by transforming the deterministic scheme in~\cref{table:adjusted scheme over k3} using~\cref{lemma:transform probabilistic}.}
\label{table:transform k3}
\end{table}

Furthermore, the following lemma extends~\cref{lemma:transform probabilistic} by incorporating the case where the deterministic PIR scheme includes additional side information.

\begin{lemma}\label{lemma:transform probabilistic side}
Let $\Pi$ be a deterministic PIR scheme over a graph $G$ with subpacketization $L$ and rate $R$, such that every summation in the queries is composed of sub-files from distinct files. Moreover, assume that the entropy $H(A_i)$ of the answer from each server $S_i$ equals the number of summations in $A_i$, for all $i \in [N]$. Suppose the queries $\QQ = (\QQ_1, \dots, \QQ_N)$ can be written as a disjoint union of $L$ recovery patterns and some additional side information. Then, if $H(A_i) \le L$ for all $i \in [N]$, there exists a probabilistic PIR scheme $\Pi'$ over $G$ with subpacketization $1$ and rate $R$.
\end{lemma}

\begin{proof}[Proof of~\cref{lemma:transform probabilistic side}]
We begin by describing the construction of the corresponding probabilistic PIR scheme $\Pi'$. Denote the $L$ recovery patterns as $\PP^j = (q_i^j)_{i \in [N]}$ for $j \in [L]$. As in~\cref{lemma:transform probabilistic}, for a uniformly random index $j \in [L]$, we construct the corresponding queries in $\Pi'$ by transforming $\PP^j$: for each $i$, if $q_i^j$ is non-empty, we ignore sub-file indices and request the corresponding full-file summation.

However, since the original scheme $\Pi$ contains some additional side information, simply discarding it could violate privacy in $\Pi'$. To preserve privacy, we modify the construction as follows.

Let $\QQ_i^j$ denote the query sent to server $S_i$ in round $j$ of $\Pi'$ generated from the recovery pattern $\PP^j$. Note that for some $j$, it may be that $\QQ_i^j = \varnothing$. We replace such empty queries with side information from $\QQ_i$ in the original scheme. Specifically, suppose the side information for $S_i$ includes a $k$-sum of type $W_{i,j_1} + \dots + W_{i,j_k}$, for some $k \ge 1$ and distinct $j_1, \dots, j_k \in [N] \setminus \{i\}$. We randomly select some $j$ such that $\QQ_i^j = \varnothing$ and replace $\QQ_i^j$ with the set $\{W_{i,j_1} + \dots + W_{i,j_k}\}$.

If there are multiple such side information terms, we repeat this process for each, replacing as many empty $\QQ_i^j$ as needed.

To confirm that this replacement is always possible, let $F_i$ denote the number of indices $j$ for which $\QQ_i^j \ne \varnothing$, i.e., the number of recovery pattern contributions to $A_i$. Then:
\begin{itemize}
    \item The number of remaining (empty) queries is $L - F_i$.
    \item The number of summations from side information is $H(A_i) - F_i$.
\end{itemize}
Since $H(A_i) \le L$ by assumption, we have:
\[
L - F_i \ge H(A_i) - F_i,
\]
so the replacement is feasible.

Finally, the reliability, privacy, and rate analysis of the resulting probabilistic scheme $\Pi'$ follow directly from the arguments in~\cref{lemma:transform probabilistic}, since all queries are now drawn from the union of recovery patterns and redistributed side information, preserving both correctness and distributional indistinguishability with respect to the desired file index.
\end{proof}

Next, we illustrate an example of~\cref{lemma:transform probabilistic side}. 

Consider the PIR scheme over 4-star graph shown in~\cref{table:4 star adjusted}, which is a disjoint union of 5 recovery patterns and some side information. Observe that the only side information in~\cref{table:4 star adjusted} is $b_3 + c_3 + d_3$ from server $S_1$. By~\cref{lemma:transform probabilistic side}, we can construct a probabilistic PIR scheme by replacing an empty query in Recovery Patterns 4 or 5 at $S_1$ with this side information. The resulting probabilistic PIR scheme is shown in~\cref{table:4 star probabilistic}.

\begin{table}[ht]
\centering
\begin{tabular}{|c|c|c|c|c|c|}
\hline
Probability & $S_1$ & $S_2$ & $S_3$ & $S_4$ & $S_5$ \\ \hline
$1/5$ & $A + B + C$ & $0$ & $B$ & $C$ & $0$ \\ \hline
$1/5$ & $A + B + D$ & $0$ & $B$ & $0$ & $D$ \\ \hline
$1/5$ & $A + C + D$ & $0$ & $0$ & $C$ & $D$ \\ \hline
$1/5$ & $B + C + D$ & $A$ & $0$ & $0$ & $0$ \\ \hline
$1/5$ & $0$         & $A$ & $0$ & $0$ & $0$ \\ \hline
\end{tabular}
\caption{Answer table of the probabilistic PIR scheme over the 4-star graph, transformed from the scheme in~\cref{table:4 star adjusted}, for $\theta = A$.}
\label{table:4 star probabilistic}
\end{table}

Combining~\cref{lemma:transform probabilistic side} and~\cref{lemma:Behavior of Algorithm 1}, we complete the proof of~\cref{thm:transform to probabilistic}.

\subsection{Probabilistic PIR schemes over general graphs}
\subsubsection{Capacities for probabilistic PIR schemes: Proof of~\cref{thm:general}}\label{subsection:Probabilistic General Graph}
We now present a probabilistic PIR scheme that achieves the retrieval rate stated in~\cref{thm:general}. This scheme applies to any (multi-)graph, requires at most one summation from each server, and is implemented through a simple randomized strategy. Moreover, we identify scenarios in which it provides improved performance over existing constructions.

Consider a (multi-)graph-based replication system with \( N \) servers \( S_1, S_2, \dots, S_N \), and \( K \) files \( W_1, W_2, \dots, W_K \). Let \( G = (V, E) \) be an undirected graph representing this replication system, where \( |V| = N \) and \( |E| = K \). Each edge in \( G \) corresponds to a file, so we label the edges as \( E = \{e_1, e_2, \dots, e_K\} \), where \( e_k \) corresponds to file \( W_k \). For a vertex (server) \( S_i \in V \), with \( i \in [N] \), define \( N_i \) as the set of indices of all edges incident to \( S_i \). The degree of server \( S_i \) is then \( d(i) = |N_i| \).

Suppose the user wishes to retrieve file \( W_\theta \) for some \( \theta \in [K] \), which is stored on servers \( S_{\theta_1} \) and \( S_{\theta_2} \), where \( \theta_1, \theta_2 \in [N] \). Without loss of generality, we assume \( \theta_1 < \theta_2 \).

We begin by outlining the core strategy of our PIR scheme. The user generates two independent binary vectors of length \( K \):
\[
\boldsymbol{\mu} = (\mu_1, \mu_2, \dots, \mu_K) \quad \text{and} \quad \boldsymbol{\lambda} = (\lambda_1, \lambda_2, \dots, \lambda_K),
\]
where each entry is drawn independently and uniformly at random from \(\{0,1\}\). Intuitively, the vector \(\boldsymbol{\mu}\) specifies which files participate in the queries, while \(\boldsymbol{\lambda}\) determines the sign (positive or negative) assigned to each file within the queries.

For each \( i \in [N] \) and each \( j \in [K] \), if file \( W_j \) is stored on server \( S_i \) (i.e., \( j \in N_i \)), we define a sign value \( \epsilon_{i,j} \in \{-1, +1\} \) based on the edge \( e_j \) and a uniformly random bit \( \lambda_j \in \{0,1\} \) associated with file \( W_j \).

Fix any \( j \in [K] \), and suppose file \( W_j \) is stored on servers \( S_{i_1} \) and \( S_{i_2} \), where \( i_1 < i_2 \). We define
\[
\epsilon_{i_1,j} = (-1)^{\lambda_j}, \quad \text{and} \quad \epsilon_{i_2,j} = (-1)^{\lambda_j + 1}.
\]
This ensures that the contributions from the two servers cancel out:
\begin{equation}
    \sum_{i \in e_j} \epsilon_{i,j} = \epsilon_{i_1,j} + \epsilon_{i_2,j} = 0.
    \label{eq:epsilon_summation_general_PIR_scheme}
\end{equation}

Thus, the values \( \epsilon_{i,j} \) are now fully defined for all pairs \( (i, j) \) with \( j \in N_i \), including the case \( j = \theta \). Since each \( \lambda_j \) is independently and uniformly chosen from \( \{0,1\} \), the signs \( \epsilon_{i,j} \) are also independent across different \( j \). In particular, for any fixed server \( i \), the collection \( \{\epsilon_{i,j} : j \in N_i\} \) consists of mutually independent random variables, each uniformly distributed over \( \{-1, +1\} \).

Based on the vectors \( \boldsymbol{\mu} \), \( \boldsymbol{\lambda} \), and the sign values \( \epsilon_{i,j} \), the user generates the queries according to the following strategy:

\begin{itemize}
    \item For each server \( S_i \) with \( i \in [N] \setminus \{\theta_1, \theta_2\} \), the user sends the linear combination:
    \[
        \sum_{j \in N_i} \mu_j \epsilon_{i,j} \cdot W_j.
    \]
    
    \item For server \( S_{\theta_1} \), the user sends:
    \[
        \mu_{\theta} \epsilon_{\theta_1,\theta} \cdot W_{\theta} + \sum_{j \in N_{\theta_1} \setminus \{\theta\}} \mu_j \epsilon_{\theta_1,j} \cdot W_j.
    \]
    
    \item For server \( S_{\theta_2} \), the user sends:
    \[
        (1 - \mu_{\theta}) \epsilon_{\theta_2,\theta} \cdot W_{\theta} + \sum_{j \in N_{\theta_2} \setminus \{\theta\}} \mu_j \epsilon_{\theta_2,j} \cdot W_j.
    \]
\end{itemize}
In this construction, the value \( \mu_j \) determines whether a non-desired file \( W_j \) is included in the queries. If \( \mu_j = 1 \), then \( W_j \) appears in the queries to both of the servers storing it, but with opposite signs due to the definition of \( \epsilon_{i,j} \), ensuring cancellation in the overall response. For the desired file \( W_{\theta} \), only one of the two servers holding it includes it in its response. The bit \( \mu_{\theta} \) determines which server it is: if \( \mu_{\theta} = 1 \), then \( W_{\theta} \) appears only in the query to \( S_{\theta_1} \); otherwise, it appears only in the query to \( S_{\theta_2} \).

Next, we verify that the constructed queries form a valid PIR scheme and compute the corresponding retrieval rate.

\paragraph{Reliability:} Consider the sum of all server responses:
\begin{align}
    \sum_{i \in [N]} A_i 
    &= \mu_{\theta} \epsilon_{\theta_1,\theta} \cdot W_{\theta} 
    + \sum_{j \in N_{\theta_1} \setminus \{\theta\}} \mu_j \epsilon_{\theta_1,j} \cdot W_j 
    + (1 - \mu_{\theta}) \epsilon_{\theta_2,\theta} \cdot W_{\theta} 
    + \sum_{j \in N_{\theta_2} \setminus \{\theta\}} \mu_j \epsilon_{\theta_2,j} \cdot W_j \notag \\
    &\quad + \sum_{i \in [N] \setminus \{\theta_1, \theta_2\}} \sum_{j \in N_i} \mu_j \epsilon_{i,j} \cdot W_j \notag \\
    &= \mu_{\theta} \epsilon_{\theta_1,\theta} \cdot W_{\theta} 
    + (1 - \mu_{\theta}) \epsilon_{\theta_2,\theta} \cdot W_{\theta} 
    + \sum_{j \in [K] \setminus \{\theta\}} \sum_{i \in e_j} \mu_j \epsilon_{i,j} \cdot W_j \label{eq:general_scheme_reliability_1} \\
    &= (1 - 2\mu_{\theta}) \epsilon_{\theta_2,\theta} \cdot W_{\theta} 
    + \mu_{\theta} (\epsilon_{\theta_1,\theta} + \epsilon_{\theta_2,\theta}) \cdot W_{\theta} 
    + \sum_{j \in [K] \setminus \{\theta\}} \mu_j W_j \sum_{i \in e_j} \epsilon_{i,j} \notag \\
    &= (1 - 2\mu_{\theta}) \epsilon_{\theta_2,\theta} \cdot W_{\theta}, \label{eq:general_scheme_reliability_2}
\end{align}
where~\eqref{eq:general_scheme_reliability_1} follows by regrouping terms, and~\eqref{eq:general_scheme_reliability_2} uses the identity from~\eqref{eq:epsilon_summation_general_PIR_scheme}, which states that \( \sum_{i \in e_j} \epsilon_{i,j} = 0 \) for every \( j \in [K] \).

Since both \( 1 - 2\mu_{\theta} \) and \( \epsilon_{\theta_2,\theta} \) take values in \( \{-1, +1\} \), the user can recover \( W_{\theta} \) up to a sign. As the user knows both values, the recovery is exact, thus ensuring the reliability of the PIR scheme.

\paragraph{Privacy:} Fix some \( i \in [N] \). We need to show that the distribution of the response \( A_i \) is independent of the choice of the desired file \( \theta \in [K] \). Without loss of generality, suppose there exists a subset \( \Lambda \subseteq N_i \) such that
\[
A_i = \sum_{j \in \Lambda} \hat{\epsilon}_{i,j} W_j,
\]
for some signs \( \hat{\epsilon}_{i,j} \in \{-1, +1\} \). We now compute the probability of this event.

First, consider the case where \( \theta \notin N_i \). Then, $\Pr\left(A_i = \sum_{j \in \Lambda} \hat{\epsilon}_{i,j} W_j \right)$ is equal to
\begin{align}
    &\Pr\left(
        (\mu_j = 0 \text{ for all } j \in N_i \setminus \Lambda) ~\wedge~
        (\mu_j = 1 \text{ for all } j \in \Lambda) ~\wedge~
        (\epsilon_{i,j} = \hat{\epsilon}_{i,j} \text{ for all } j \in \Lambda)
    \right) \notag \\
    &= \prod_{j \in N_i \setminus \Lambda} \Pr(\mu_j = 0) 
    \cdot\prod_{j \in \Lambda} \Pr(\mu_j = 1) 
    \cdot\prod_{j \in \Lambda} \Pr(\epsilon_{i,j} = \hat{\epsilon}_{i,j}) \label{eq:general scheme privacy 1} \\
    &= \left(\frac{1}{2}\right)^{|N_i \setminus \Lambda|} 
    \cdot \left(\frac{1}{2}\right)^{|\Lambda|} 
    \cdot \left(\frac{1}{2}\right)^{|\Lambda|} \label{eq:general scheme privacy 2} \\
    &= \frac{1}{2^{d(i) + |\Lambda|}}, \notag
\end{align}
where~\eqref{eq:general scheme privacy 1} follows from the independence of the bits \( \mu_j \) and the fact that each \( \epsilon_{i,j} \) depends only on the independent random bit \( \lambda_j \), which is independent of \( \mu_j \). Equation~\eqref{eq:general scheme privacy 2} uses that both \( \mu_j \) and \( \lambda_j \) are sampled uniformly and independently from \( \{0,1\} \).

Now consider the case \( \theta \in N_i \). The only difference is that the query involving \( W_\theta \) may use either \( \mu_\theta \) or \( 1 - \mu_\theta \), depending on which server sends the desired term. However, since \( \mu_\theta \sim \text{Unif}(\{0,1\}) \), this does not affect the distribution of \( A_i \). Hence, for any fixed server \( S_i \), the distribution of \( A_i \) is independent of the identity of \( \theta \), which completes the proof of privacy.

\paragraph{Rate:} For any \( i \in [N] \setminus \{\theta_1, \theta_2\} \), we have
\begin{align}
    \Pr(A_i = 0) &= \Pr(\mu_j = 0 \text{ for all } j \in N_i) \notag \\
    &= \prod_{j \in N_i} \Pr(\mu_j = 0) \notag \\
    &= \frac{1}{2^{d(i)}},
\end{align}
where \( d(i) = |N_i| \) is the degree of server \( S_i \) in \( G \).

For \( i \in \{\theta_1, \theta_2\} \), the same argument applies since \( 1 - \mu_\theta \) is also uniformly random and independent of \( \mu_j \) for \( j \in N_i \setminus \{\theta\} \). Moreover, observe that whenever \( A_i \ne 0 \), the number of bits of $A_i$ is \( 1 \), which equals the subpacketization of scheme. Therefore, the rate of the PIR scheme is given by
\begin{align*}
    R(\Pi) &= \frac{1}{\sum_{i \in [N]} H(A_i)} \\
    &= \frac{L}{\sum_{i \in [N]} \Pr(A_i \ne 0) \cdot 1} \\
    &= \frac{1}{\sum_{i \in [N]} \left(1 - \frac{1}{2^{d(i)}}\right)}.
\end{align*}

\subsubsection{Examples and comparison}
We now illustrate the proposed scheme with a simple example to highlight its construction and performance. For simplicity, we use $\theta$ to represent the target file.

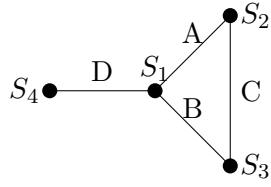
\begin{figure}
    \centering
\begin{tikzpicture}
    \draw[](-1.4,0)--(0,0);
    \draw[](0,0)--(1,1);
    \draw[](0,0)--(1,-1);
    \draw[](1,1)--(1,-1);
    \fill(0,0)circle(.1);
    \fill(-1.4,0)circle(.1);
    \fill(1,1)circle(.1);    
    \fill(1,-1)circle(.1);    
    \node at (0,0)[above]{$S_1$};
    \node at (-1.4,0)[left]{$S_4$};
    \node at (1,1)[right]{$S_2$};
    \node at (1,-1)[right]{$S_3$};
    \node at (-0.7,0)[above]{D};
    \node at (0.5,0.5)[above]{A};
    \node at (0.5,-0.5)[above]{B};
    \node at (1,0)[right]{C};
\end{tikzpicture}
    \caption{An example for general graph}
    \label{fig:genral}
\end{figure}

For example, consider the graph in~\cref{fig:genral}, if $\theta=A,\boldsymbol{\mu}=(1,0,1,0),\boldsymbol{\lambda}=(0,1,0,1)$, the queries of four servers are $\mathcal{Q}_1=A,\mathcal{Q}_2=C,\mathcal{Q}_3=-C,\mathcal{Q}_4=0$; if $\theta=C,\boldsymbol{\mu}=(1,1,0,1),\boldsymbol{\lambda}=(0,0,0,1)$, the queries of four servers are $\mathcal{Q}_1=A+B-D,\mathcal{Q}_2=-A,\mathcal{Q}_3=-B-C,\mathcal{Q}_4=D$.

Next, we will respectively fix the values of $\boldsymbol{\mu}$ and $\boldsymbol{\lambda}$ to demonstrate the roles of these two random vectors in the scheme. For simplicity, we represent the queries to servers \(S_1, S_2, S_3, S_4\) as a vector \((\mathcal{Q}_1, \mathcal{Q}_2, \mathcal{Q}_3, \mathcal{Q}_4)\). Fix the vector \(\boldsymbol{\lambda} = (0, 0, 0, 0)\). In~\cref{table:random demand}, we list all possible choices of \(\boldsymbol{\mu}\) along with the corresponding queries sent to the servers in the PIR scheme over the graph in~\cref{fig:genral}. 

\begin{table}[htbp]
\scriptsize
\centering  
\begin{tabular}{|c|c|c|c|c|}
\hline

$\boldsymbol{\mu}$ & $(0,0,0,0)$ & $(1,0,0,0)$ & $(0,1,0,0)$ & $(0,0,1,0)$ \\ \hline
query & $(0,-A,0,0)$ & $(A,0,0,0)$ & $(B,-A,-B,0)$ & $(0,-A+C,-C,0)$ \\ \hline
$\boldsymbol{\mu} $ & $(0,0,0,1)$ & $(1,1,0,0)$ & $(1,0,1,0)$& $(1,0,0,1)$ \\ \hline
query& $(D,-A,0,-D)$ & $(A+B,0,-B,0)$ & $(A,C,-C,0)$ & $(A+D,0,0,-D)$ \\ \hline
$\boldsymbol{\mu}$ & $(0,1,1,0)$ & $(0,1,0,1)$ & $(0,0,1,1)$ & $(1,1,1,0)$\\ \hline
query & $(B,-A+C,-B-C,0)$ & $(B+D,-A,-B,-D)$ & $(D,-A+C,-C,-D)$ & $(A+B,C,-B-C,0)$ \\ \hline
$\boldsymbol{\mu}$ & $(1,1,0,1)$ & $(1,0,1,1)$ & $(0,1,1,1)$ & $(1,1,1,1)$\\ \hline
query & $(A+B+D,0,-B,-D)$ & $(A+D,C,-C,-D)$ & $(B+D,-A+C,-B-C,-D)$ & $(A+B+D,C,-B-C,-D)$ \\ \hline
\end{tabular}

\caption{Table of queries of probabilistic scheme when $\boldsymbol{\lambda}=(0,0,0,0)$, for $\theta=A$.}
\label{table:random demand}
\end{table}

\cref{table:random demand} illustrates the impact of the vector \( \boldsymbol{\mu} \), which determines the set of files involved in the queries. Moreover, for any value of \( \boldsymbol{\mu} \), the user can recover the desired file by summing the responses from all servers, thereby ensuring the reliability of the PIR scheme.

Next, we fix the vector $\boldsymbol{\mu}=(1,1,1,1)$. In~\cref{tabel:mu=1111}, we present all possible choices of $\boldsymbol{\lambda}$ and the corresponding queries when $\theta=A$. 

\begin{table}[htbp]
\tiny
\centering
\begin{tabular}{|c|c|c|c|c|}
\hline
$\boldsymbol{\lambda}$ & $(0,0,0,0)$ & $(1,0,0,0)$ & $(0,1,0,0)$ & $(0,0,1,0)$ \\ \hline
query & $(A+B+D,C,-B-C,-D)$ & $(-A+B+D,C,-B-C,-D)$ & $(A-B+D,C,-B-C,-D)$ & $(A+B+D,-C,C-D,-D)$ \\ \hline
$\boldsymbol{\lambda}$ & $(0,0,0,1)$ & $(1,1,0,0)$ & $(1,0,1,0)$ & $(1,0,0,1)$ \\ \hline
query & $(A+B-D,C,-B-C,D)$ & $(-A-B+D,C,B-C,-D)$ & $(-A+B+D,-C,-B+C,-D)$ & $(-A+B-D,C,-B-C,D)$ \\ \hline
$\boldsymbol{\lambda}$ & $(0,1,1,0)$ & $(0,1,0,1)$ & $(0,0,1,1)$ & $(1,1,1,0)$ \\ \hline
query & $(A-B+D,-C,B+C,-D)$ & $(A-B-D,C,B-C,D)$ & $(A+B-D,-C,-B+C,D)$ & $(-A-B+D,-C,B+C,-D)$ \\ \hline
$\boldsymbol{\lambda}$ & $(1,1,0,1)$ & $(1,0,1,1)$ & $(0,1,1,1)$ & $(1,1,1,1)$ \\ \hline
query & $(-A-B-D,C,B-C,D)$ & $(-A+B-D,-C,-B+C,D)$ & $(A-B-D,-C,B+C,D)$ & $(-A-B-D,-C,B+C,D)$ \\ \hline
\end{tabular}
\caption{Table of queries of probabilistic scheme when $\boldsymbol{\mu}=(1,1,1,1)$, for $\theta=A$.}
\label{tabel:mu=1111}
\end{table}

\cref{tabel:mu=1111} actually demonstrates the effect of vector $\boldsymbol{\lambda}$, which controls the positivity or negativity of the corresponding file in the queries through the corresponding components in the vector, ensuring that the probability of requesting positive or negative values from the server is the same.

Finally, from the perspective of the servers, we have a more intuitive understanding of how this scheme achieves privacy.

If the query $\mathcal{Q}_1$ received by server $S_1$ is $A+B+D$, then the possible values of $\theta,\boldsymbol{\mu} ,  \boldsymbol{\lambda}$ and all queries are as shown in~\cref{table:a+b+d}. \cref{table:a+b+d} clearly shows that when the server $S_1$ receives query $\mathcal{Q}_1=A+B+D$, the target file $W_{\theta}$ has an equal probability being any one of the four available files. Therefore, for each server, the received query does not reflect the actual needs of the users.
\begin{table}[htbp]
\footnotesize
\centering
\begin{tabular}{|c|cc|}
\hline
           & \multicolumn{2}{c|}{$\mathcal{Q}_1=A+B+D$}                                                                                                     \\ \hline
vector     & \multicolumn{1}{c|}{$\boldsymbol{\mu}=(1,1,0,1),\boldsymbol{\lambda}=(0,0,0,0)$} & $\boldsymbol{\mu}=(1,1,0,1),\boldsymbol{\lambda}=(0,0,1,0)$ \\ \hline
$\theta=A$ & \multicolumn{1}{c|}{$(A+B+D,0,-B,-D)$}                                           & $(A+B+D,0,-B,-D)$                                           \\ \hline
$\theta=B$ & \multicolumn{1}{c|}{$(A+B+D,-A,0,-D)$}                                           & $(A+B+D,-A,0,-D)$                                           \\ \hline
$\theta=C$ & \multicolumn{1}{c|}{$(A+B+D,-A,-B-C,-D)$}                                        & $(A+B+D,-A,-B+C,-D)$                                         \\ \hline
$\theta=D$ & \multicolumn{1}{c|}{$(A+B+D,-A,-B,0)$}                                           & $(A+B+D,-A,-B,0)$                                           \\ \hline
vector     & \multicolumn{1}{c|}{$\boldsymbol{\mu}=(1,1,1,1),\boldsymbol{\lambda}=(0,0,0,0)$} & $\boldsymbol{\mu}=(1,1,1,1),\boldsymbol{\lambda}=(0,0,1,0)$ \\ \hline
$\theta=A$ & \multicolumn{1}{c|}{$(A+B+D,C,-B-C,-D)$}                                         & $(A+B+D,-C,-B+C,-D)$                                        \\ \hline
$\theta=B$ & \multicolumn{1}{c|}{$(A+B+D,-A+C,-C,-D)$}                                         & $(A+B+D,-A-C,C,-D)$                                         \\ \hline
$\theta=C$ & \multicolumn{1}{c|}{$(A+B+D,-A+C,-B,-D)$}                                         & $(A+B+D,-A-C,-B,-D)$                                        \\ \hline
$\theta=D$ & \multicolumn{1}{c|}{$(A+B+D,-A+C,-B-C,0)$}                                        & $(A+B+D,-A-C,-B+C,0)$                                        \\ \hline
\end{tabular}
\caption{Possible value of $\theta,\boldsymbol{\mu} \ and\  \boldsymbol{\lambda}$ when $\mathcal{Q}_1=A+B+D$.}
\label{table:a+b+d}
\end{table}

For a general graph, compared to the baseline rate \( \frac{2^{N-1}}{2^{N-1}-1} \cdot \frac{1}{N} \) from~\cite{Sadeh2023bound}, our scheme achieves a strictly better rate:
\[
R(\Pi) = \frac{1}{\sum_{i \in [N]} \left(1 - \frac{1}{2^{d(i)}}\right)},
\]
with subpacketization 1. Furthermore, our construction does not require generating a PIR scheme over complete graph as a subroutine; the scheme can be directly implemented on any general graph.

Compared to the scheme derived from the complete-graph-based construction, our scheme achieves better performance when the average degree of the graph is below 2. While the complete-graph-based approach may yield a higher retrieval rate when the average degree exceeds 2, this construction offers a practical advantage by avoiding the need to solve auxiliary derived schemes. This trade-off between rate and implementation complexity may make our approach more favorable in many practical scenarios.

Finally, in the case of \( r \)-multigraphs ($r \ge 2$), for a general graph $G$, the \emph{r-multigraph extension} of $G$ is $G^{(r)}$, the lower bound of $\mathcal{C}(G)$ is given by \cref{cor:anygraph}, and the lower bound of $\mathcal{C}(G^{(r)})$ is derived from~\cref{thm:mutigraph}, which gives
\[
\mathcal{C}(G^{(r)}) \ge \left(\frac{1}{2 - \frac{1}{2^{r-1}}} \right) \cdot R(\Pi) 
= \left(\frac{1}{2 - \frac{1}{2^{r-1}}} \right) \cdot \left( \frac{4}{3} - \epsilon_N \right) \cdot \frac{1}{N}.
\]
Our probabilistic construction yields the lower bound
\[
\mathcal{C}(G^{(r)}) \ge \frac{1}{N - \sum_{i=1}^{N} \frac{1}{2^{r \cdot d(i)}}},
\]
which is strictly better. In particular, we have:
\[
\left ( \frac{1}{2-\frac{1}{2^{r-1}} }  \right ) \cdot \frac{4}{3} \cdot \frac{1}{N} \le \frac{8}{9} \cdot\frac{1}{N} \le \frac{1}{N} \le \frac{1}{N-\sum_{i=1}^{N}\frac{1}{2^{r\cdot d(i)} }  }.
\]
This confirms that, for general multi-graphs, the probabilistic scheme presented in this section is strictly superior.

\section{Concluding Remarks}\label{section:conclusion}
In this work, we investigated the PIR capacity of graph-based replication systems with non-colluding servers, focusing particularly on complete graphs. We introduced a novel approach to derive improved upper bounds that exploit the symmetry of the underlying graph structure, yielding results that are substantially tighter than existing bounds. From an algorithmic perspective, we developed an explicit recursive construction of PIR schemes that provides enhanced lower bounds on the capacity. Consequently, we reduced the asymptotic gap between the upper and lower bounds for the PIR capacity of complete graphs to approximately~1.0444, thereby making significant progress toward resolving this fundamental problem.

Beyond these concrete advances, our results reveal deeper structural insights into how graph topology constrains information flow in PIR protocols. Notably, our techniques extend naturally to other highly symmetric graphs such as complete bipartite graphs and certain multigraphs. Moreover, the connection we established between deterministic and probabilistic PIR schemes offers new perspectives for designing low-subpacketization protocols, which are crucial for practical deployment. Despite these developments, many fundamental questions remain open:
\begin{itemize}
    \item \textbf{Exact capacity of small graphs.} Can the remaining gap between upper and lower bounds for $\mathcal{C}(K_4)$ be closed? Determining the precise capacity of small graphs remains a challenging but potentially tractable problem.

    \item \textbf{Beyond complete graphs.} Can the recursive construction framework developed for $K_N$ be generalized to yield bounds for other graph classes such as trees, expanders, or bounded-degree graphs?

    \item \textbf{Tight bounds for multigraphs.} The probabilistic PIR scheme we propose for the multigraph $K_N^{(r)}$ is relatively simple and likely far from optimal. Is it possible to establish tighter upper or lower bounds on $\mathcal{C}(K_N^{(r)})$, especially as $r$ increases?

    \item \textbf{Beyond 2-replication.} Our current model assumes each file is replicated on exactly two servers. Can our techniques be extended to hypergraph-based systems, where each file may be stored on more than two servers?

    \item \textbf{Graph density vs PIR capacity.} In~\cref{cor:graph contain}, we showed that $\mathcal{C}(G_1) \ge \mathcal{C}(G_2)$ whenever $G_1 \subseteq G_2$. This raises a broader question: is there a general monotonic relationship between the density or degree distribution of a graph and its PIR capacity?
\end{itemize}

We believe that the methods introduced in this paper, from entropy-based upper bounds to structure-aware PIR constructions, lay the groundwork for further progress in understanding PIR over structured storage systems. In particular, exploring connections with extremal graph theory, information inequalities, and coding-theoretic tools may yield deeper insights into the combinatorial limits of private information retrieval.

\bibliographystyle{abbrv}
\bibliography{PIRoverGraph}

\appendix

\section{Proof of~\cref{lem:symmetric complete}}\label{section:proof of symmetric of complete graph}
   Consider the automorphism group \( \operatorname{Aut}(K_N) \), which coincides with the symmetric group on \( N \) vertices. For an arbitrary PIR scheme \( \Pi \) over \( K_N \), and a desired file indexed by \( \theta \), corresponding to an edge \( e \in E(K_N) \), let \( \mathcal{Q}_i^{\Pi}(\theta) \) denote the query sent to the \( i \)-th server, and let \( A_i^{\Pi} \) denote the corresponding response. Let \( \mathcal{Q}^{\Pi} \) denote the complete collection of queries generated by the user.

Now, for any achievable rate \( R \), let \( \Pi \) be a PIR scheme over \( K_N \) that achieves rate \( R \). For any automorphism \( \sigma \in \operatorname{Aut}(K_N) \), define \( \sigma(\theta) \) to be the index of the file corresponding to the edge \( \sigma(e) \), where \( e \in E(K_N) \) is the edge associated with the original index \( \theta \). 

Given the scheme \( \Pi \) and the automorphism \( \sigma \), we construct a new PIR scheme \( \Pi_{\sigma} \) for retrieving the file indexed by \( \theta \) as follows: the user sends to server \( i \) the query
\begin{equation}
    \mathcal{Q}_i^{\Pi_{\sigma}} := \mathcal{Q}_{\sigma(i)}^{\Pi}(\sigma(\theta)),
\end{equation}
and server \( i \) responds with
\begin{equation}
    A_i^{\Pi_{\sigma}} := A_{\sigma(i)}^{\Pi}.
\end{equation}
    In other words, the scheme \( \Pi_{\sigma} \) can be viewed as applying the original scheme \( \Pi \) to the relabeled graph \( \sigma(K_N) \), with the desired file index transformed to \( \sigma(\theta) \). Since \( \sigma(K_N) \) is merely a permutation of the vertices of \( K_N \), the reliability of \( \Pi \) guarantees that the user can correctly recover the file corresponding to \( \sigma(\theta) \). The privacy of \( \Pi_{\sigma} \) is also preserved, as the scheme effectively executes \( \Pi \) under a relabeling of the servers, which does not reveal any information about the identity of the requested file. Furthermore, the rate of \( \Pi_{\sigma} \) is identical to that of \( \Pi \).

   Next, to construct a PIR scheme that satisfies~\eqref{eq:lem:symmetric complete 1} and~\eqref{eq:lem:symmetric complete 2} while preserving the rate of \( \Pi \), we define a new scheme \( \Pi' \) as follows. The scheme \( \Pi' \) randomly selects one of the schemes \( \Pi_{\sigma} \), where \( \sigma \in \operatorname{Aut}(K_N) \), and uses it to retrieve the desired file. The procedure is described below:
\begin{itemize}
    \item The user selects a desired file \( \theta \in E(K_N) \) and independently samples an automorphism \( \sigma \in \operatorname{Aut}(K_N) \) uniformly at random.
    \item The user generates the queries \( \mathcal{Q}^{\Pi'} := \mathcal{Q}^{\Pi_{\sigma}} \) and sends them to the corresponding servers.
    \item Each server \( i \) responds with the answer \( A_i^{\Pi'} := A_i^{\Pi_{\sigma}} \).
\end{itemize}
Clearly, \( \Pi' \) achieves the same rate as \( \Pi \), since it is equivalent to taking a uniform average over all schemes \( \Pi_{\sigma} \) for \( \sigma \in \operatorname{Aut}(K_N) \). 

To prove \eqref{eq:lem:symmetric complete 1} for any \( k \in [N-2] \), let \( G_k \) denote the induced subgraph of \( K_N \) on the vertex set \( S_{[k+1, N]} \). Consider two servers \( S_i \) and \( S_j \) with \( i \neq j > k \). By the vertex-transitivity of \( K_N \), there exists an automorphism \( \tau \in \operatorname{Aut}(K_N) \) such that \( \tau(S_i) = S_j \), \( \tau(S_j) = S_i \), and \( \tau \) fixes all other vertices.
    
    Then, we have
    \begin{align}
        H(A_i^{\Pi'} \mid \WW_{[k]}, \QQ^{\Pi'}) 
        &= \frac{1}{|\operatorname{Aut}(K_N)|} \sum_{\sigma \in \operatorname{Aut}(K_N)} H(A_i^{\Pi_{\sigma}} \mid \WW_{[k]}, \QQ^{\Pi_{\sigma}}) \label{eq:lem:symmetric complete 3} \\
        &= \frac{1}{|\operatorname{Aut}(K_N)|} \sum_{\sigma \in \operatorname{Aut}(K_N)} H(A_{\sigma(i)}^{\Pi} \mid \sigma(\WW_{[k]}), \QQ^{\Pi}) \notag \\
        &= \frac{1}{|\operatorname{Aut}(K_N)|} \sum_{\sigma \in \operatorname{Aut}(K_N)} H(A_{\sigma(\tau(i))}^{\Pi} \mid \sigma(\tau(\WW_{[k]})), \QQ^{\Pi}) \label{eq:lem:symmetric complete 4} \\
        &= \frac{1}{|\operatorname{Aut}(K_N)|} \sum_{\sigma \in \operatorname{Aut}(K_N)} H(A_{\sigma(j)}^{\Pi} \mid \sigma(\WW_{[k]}), \QQ^{\Pi}) \label{eq:lem:symmetric complete 5} \\
        &= \frac{1}{|\operatorname{Aut}(K_N)|} \sum_{\sigma \in \operatorname{Aut}(K_N)} H(A_j^{\Pi_{\sigma}} \mid \WW_{[k]}, \QQ^{\Pi_{\sigma}}) \notag \\
        &= H(A_j^{\Pi'} \mid \WW_{[k]}, \QQ^{\Pi'}), \notag
    \end{align}
    where \eqref{eq:lem:symmetric complete 3} follows from the uniform randomness of the choice of \( \sigma \); \eqref{eq:lem:symmetric complete 4} uses the change of variables \( \sigma \mapsto \sigma \circ \tau \) (noting that \( \sigma \circ \tau \) runs over all automorphisms as \( \sigma \) does); and \eqref{eq:lem:symmetric complete 5} follows from \( \tau(i) = j \) and \( \tau(\WW_{[k]}) = \WW_{[k]} \) (since \( \tau \) only swaps \( i \) and \( j \) and fixes the rest). Thus, \eqref{eq:lem:symmetric complete 1} is established.
    
   An almost identical argument establishes~\eqref{eq:lem:symmetric complete 2}, and we omit the repeated details here.

\section{Table for PIR scheme over $K_4$}\label{section:Table for PIR scheme over k4}
In~\cref{subsection:PIR scheme K4}, we construct a deterministic PIR scheme over \( K_4 \) using a greedy approach. The specific queries of this PIR scheme are presented in~\cref{table:specific K4}.

\begin{table}[hbtp]
\centering
\begin{tabular}{|c|c|c|c|c|c|}
\hline
Step & Patterns & $S_1$ & $S_2$ & $S_3$ & $S_4$ \\ \hline
\multirow{8}{*}{1} & \multirow{8}{*}{1} & $a_1$ &  &  &  \\ \cline{3-6}
 &  & $a_2$ &  &  &  \\ \cline{3-6}
 &  & $a_3$ &  &  &  \\ \cline{3-6}
 &  & $a_4$ &  &  &  \\ \cline{3-6}
 &  &  & $a_5$ &  &  \\ \cline{3-6}
 &  &  & $a_6$ &  &  \\ \cline{3-6}
 &  &  & $a_7$ &  &  \\ \cline{3-6}
 &  &  & $a_8$ &  &  \\ \hline
\multirow{40}{*}{2} & \multirow{16}{*}{2} & $a_9+b_1$ &  & $b_1$ &  \\ \cline{3-6}
 &  & $a_{10}+b_2$ &  & $b_2$ &  \\ \cline{3-6}
 &  & $a_{11}+b_3$ &  & $b_3$ &  \\ \cline{3-6}
 &  & $a_{12}+b_4$ &  & $b_4$ &  \\ \cline{3-6}
 &  & $a_{13}+c_1$ &  &  & $c_1$ \\ \cline{3-6}
 &  & $a_{14}+c_2$ &  &  & $c_2$ \\ \cline{3-6}
 &  & $a_{15}+c_3$ &  &  & $c_3$ \\ \cline{3-6}
 &  & $a_{16}+c_4$ &  &  & $c_4$ \\ \cline{3-6}
 &  &  & $a_{17}+d_1$ & $d_1$ &  \\ \cline{3-6}
 &  &  & $a_{18}+d_2$ & $d_2$ &  \\ \cline{3-6}
 &  &  & $a_{19}+d_3$ & $d_3$ &  \\ \cline{3-6}
 &  &  & $a_{20}+d_4$ & $d_4$ &  \\ \cline{3-6}
 &  &  & $a_{21}+e_1$ &  & $e_1$ \\ \cline{3-6}
 &  &  & $a_{22}+e_2$ &  & $e_2$ \\ \cline{3-6}
 &  &  & $a_{23}+e_3$ &  & $e_3$ \\ \cline{3-6}
 &  &  & $a_{24}+e_4$ &  & $e_4$ \\ \cline{2-6}
 & \multirow{16}{*}{3} & $a_{25}+b_{5}$ & $d_5$ & $b_5+d_5$ &  \\ \cline{3-6}
 &  & $a_{26}+b_{6}$ & $d_6$ & $b_6+d_6$ &  \\ \cline{3-6}
 &  & $a_{27}+b_{7}$ & $d_7$ & $b_7+d_7$ &  \\ \cline{3-6}
 &  & $a_{28}+b_{8}$ & $d_8$ & $b_8+d_8$ &  \\ \cline{3-6}
 &  & $a_{29}+c_{5}$ & $e_5$ &  & $c_5+e_5$ \\ \cline{3-6}
 &  & $a_{30}+c_{6}$ & $e_6$ &  & $c_6+e_6$ \\ \cline{3-6}
 &  & $a_{31}+c_{7}$ & $e_7$ &  & $c_7+e_7$ \\ \cline{3-6}
 &  & $a_{32}+c_{8}$ & $e_8$ &  & $c_8+e_8$ \\ \cline{3-6}
 &  & $b_9$ & $a_{33}+d_9$ & $b_9+d_9$ &  \\ \cline{3-6}
 &  & $b_{10}$ & $a_{34}+d_{10}$ & $b_{10}+d_{10}$ &  \\ \cline{3-6}
 &  & $b_{11}$ & $a_{35}+d_{11}$ & $b_{11}+d_{11}$ &  \\ \cline{3-6}
 &  & $b_{12}$ & $a_{36}+d_{12}$ & $b_{12}+d_{12}$ &  \\ \cline{3-6}
 &  & $c_{9}$ & $a_{37}+e_{9}$ &  & $c_{9}+e_{9}$ \\ \cline{3-6}
 &  & $c_{10}$ & $a_{38}+e_{10}$ &  & $c_{10}+e_{10}$ \\ \cline{3-6}
 &  & $c_{11}$ & $a_{39}+e_{11}$ &  & $c_{11}+e_{11}$ \\ \cline{3-6}
 &  & $c_{12}$ & $a_{40}+e_{12}$ &  & $c_{12}+e_{12}$ \\ \cline{2-6}
 & \multirow{8}{*}{4} & $a_{41}+b_{13}$ &  & $b_{13}+f_1$ & $f_1$ \\ \cline{3-6}
 &  & $a_{42}+b_{14}$ &  & $b_{14}+f_2$ & $f_2$ \\ \cline{3-6}
 &  & $a_{43}+c_{13}$ &  & $f_3$ & $c_{13}+f_3$ \\ \cline{3-6}
 &  & $a_{44}+c_{14}$ &  & $f_4$ & $c_{14}+f_4$ \\ \cline{3-6}
 &  &  & $a_{45}+d_{13}$ & $d_{13}+f_5$ & $f_5$ \\ \cline{3-6}
 &  &  & $a_{46}+d_{14}$ & $d_{14}+f_6$ & $f_6$ \\ \cline{3-6}
 &  &  & $a_{47}+e_{13}$ & $f_7$ & $e_{13}+f_7$ \\ \cline{3-6}
 &  &  & $a_{48}+e_{14}$ & $f_8$ & $e_{14}+f_8$ \\ \hline
\end{tabular}
\end{table}

\begin{table}[hbtp]
\centering
\begin{tabular}{|c|c|c|c|c|c|}
\hline
Step & Patterns & $S_1$ & $S_2$ & $S_3$ & $S_4$ \\ \hline
\multirow{36}{*}{3} & \multirow{16}{*}{5} & $a_{49}+b_{15}+c_{15}$ &  & $b_{15}+f_{9}$ & $c_{15}+f_{9}$ \\ \cline{3-6}
 &  & $a_{50}+b_{16}+c_{16}$ &  & $b_{16}+f_{10}$ & $c_{16}+f_{10}$ \\ \cline{3-6}
 &  & $a_{51}+b_{17}+c_{16}$ &  & $b_{17}+f_{11}$ & $c_{17}+f_{11}$ \\ \cline{3-6}
 &  & $a_{52}+b_{18}+c_{16}$ &  & $b_{18}+f_{12}$ & $c_{18}+f_{12}$ \\ \cline{3-6}
 &  & $a_{53}+b_{19}+c_{16}$ &  & $b_{19}+f_{13}$ & $c_{19}+f_{13}$ \\ \cline{3-6}
 &  & $a_{54}+b_{20}+c_{16}$ &  & $b_{20}+f_{14}$ & $c_{20}+f_{14}$ \\ \cline{3-6}
 &  & $a_{55}+b_{21}+c_{16}$ &  & $b_{21}+f_{15}$ & $c_{21}+f_{15}$ \\ \cline{3-6}
 &  & $a_{56}+b_{22}+c_{16}$ &  & $b_{22}+f_{16}$ & $c_{22}+f_{16}$ \\ \cline{3-6}
 &  &  & $a_{57}+d_{15}+e_{15}$ & $d_{15}+f_{17}$ & $e_{15}+f_{17}$ \\ \cline{3-6}
 &  &  & $a_{58}+d_{16}+e_{16}$ & $d_{16}+f_{18}$ & $e_{16}+f_{18}$ \\ \cline{3-6}
 &  &  & $a_{59}+d_{17}+e_{17}$ & $d_{17}+f_{19}$ & $e_{17}+f_{19}$ \\ \cline{3-6}
 &  &  & $a_{60}+d_{18}+e_{18}$ & $d_{18}+f_{20}$ & $e_{18}+f_{20}$ \\ \cline{3-6}
 &  &  & $a_{61}+d_{19}+e_{19}$ & $d_{19}+f_{21}$ & $e_{19}+f_{21}$ \\ \cline{3-6}
 &  &  & $a_{62}+d_{20}+e_{20}$ & $d_{20}+f_{22}$ & $e_{20}+f_{22}$ \\ \cline{3-6}
 &  &  & $a_{63}+d_{21}+e_{21}$ & $d_{21}+f_{23}$ & $e_{21}+f_{23}$ \\ \cline{3-6}
 &  &  & $a_{64}+d_{22}+e_{22}$ & $d_{22}+f_{24}$ & $e_{22}+f_{24}$ \\ \cline{2-6}
 & \multirow{2}{*}{6} & $a_{65}+b_{23}+c_{23}$ & $d_{23}+e_{23}$ & $b_{23}+d_{23}$ & $c_{23}+e_{23}$ \\ \cline{3-6}
 &  & $b_{24}+c_{24}$ & $a_{66}+d_{24}+e_{24}$ & $b_{24}+d_{24}$ & $c_{24}+e_{24}$ \\ \cline{2-6}
 & \multirow{18}{*}{7} & $a_{67}+b_{25}+c_{25}$ & $d_{25}+e_{25}$ & $b_{25}+d_{25}+f_{25}$ & $c_{25}+e_{25}+f_{25}$ \\ \cline{3-6}
 &  & $a_{68}+b_{26}+c_{26}$ & $d_{26}+e_{26}$ & $b_{26}+d_{26}+f_{26}$ & $c_{26}+e_{26}+f_{26}$ \\ \cline{3-6}
 &  & $a_{69}+b_{27}+c_{27}$ & $d_{27}+e_{27}$ & $b_{27}+d_{27}+f_{27}$ & $c_{27}+e_{27}+f_{27}$ \\ \cline{3-6}
 &  & $a_{70}+b_{28}+c_{28}$ & $d_{28}+e_{28}$ & $b_{28}+d_{28}+f_{28}$ & $c_{28}+e_{28}+f_{28}$ \\ \cline{3-6}
 &  & $a_{71}+b_{29}+c_{29}$ & $d_{29}+e_{29}$ & $b_{29}+d_{29}+f_{29}$ & $c_{29}+e_{29}+f_{29}$ \\ \cline{3-6}
 &  & $a_{72}+b_{30}+c_{30}$ & $d_{30}+e_{30}$ & $b_{30}+d_{30}+f_{30}$ & $c_{30}+e_{30}+f_{30}$ \\ \cline{3-6}
 &  & $a_{73}+b_{31}+c_{31}$ & $d_{31}+e_{31}$ & $b_{31}+d_{31}+f_{31}$ & $c_{31}+e_{31}+f_{31}$ \\ \cline{3-6}
 &  & $a_{74}+b_{32}+c_{32}$ & $d_{32}+e_{32}$ & $b_{32}+d_{32}+f_{32}$ & $c_{32}+e_{32}+f_{32}$ \\ \cline{3-6}
 &  & $a_{75}+b_{33}+c_{33}$ & $d_{33}+e_{33}$ & $b_{33}+d_{33}+f_{33}$ & $c_{33}+e_{33}+f_{33}$ \\ \cline{3-6}
 &  & $b_{34}+c_{34}$ & $a_{76}+d_{34}+e_{34}$ & $b_{34}+d_{34}+f_{34}$ & $c_{34}+e_{34}+f_{34}$ \\ \cline{3-6}
 &  & $b_{35}+c_{35}$ & $a_{77}+d_{35}+e_{35}$ & $b_{35}+d_{35}+f_{35}$ & $c_{35}+e_{35}+f_{35}$ \\ \cline{3-6}
 &  & $b_{36}+c_{36}$ & $a_{78}+d_{36}+e_{36}$ & $b_{36}+d_{36}+f_{36}$ & $c_{36}+e_{36}+f_{36}$ \\ \cline{3-6}
 &  & $b_{37}+c_{37}$ & $a_{79}+d_{37}+e_{37}$ & $b_{37}+d_{37}+f_{37}$ & $c_{37}+e_{37}+f_{37}$ \\ \cline{3-6}
 &  & $b_{38}+c_{38}$ & $a_{80}+d_{38}+e_{38}$ & $b_{38}+d_{38}+f_{38}$ & $c_{38}+e_{38}+f_{38}$ \\ \cline{3-6}
 &  & $b_{39}+c_{39}$ & $a_{81}+d_{39}+e_{39}$ & $b_{39}+d_{39}+f_{39}$ & $c_{39}+e_{39}+f_{39}$ \\ \cline{3-6}
 &  & $b_{40}+c_{40}$ & $a_{82}+d_{40}+e_{40}$ & $b_{40}+d_{40}+f_{40}$ & $c_{40}+e_{40}+f_{40}$ \\ \cline{3-6}
 &  & $b_{41}+c_{41}$ & $a_{83}+d_{41}+e_{41}$ & $b_{41}+d_{41}+f_{41}$ & $c_{41}+e_{41}+f_{41}$ \\ \cline{3-6}
 &  & $b_{42}+c_{42}$ & $a_{84}+d_{42}+e_{42}$ & $b_{42}+d_{42}+f_{42}$ & $c_{42}+e_{42}+f_{42}$ \\ \hline
\end{tabular}
\caption{The table of specific PIR scheme over $K_4$ with rate $\frac{7}{20}$ and sub-packetization $84$.}
\label{table:specific K4}
\end{table}

\section{Proof of~\cref{lemma:specific formula of xk}}\label{section:specific formula of xk}
We first focus on deriving a closed-form expression for \( x_k \) when \( 1 \le k \le \left\lfloor \frac{N}{2} \right\rfloor + 1 \). Combining the recursive relations for \( x_k \) and \( z_k \) given in~\eqref{eq:case 1}, we obtain:
\begin{equation}
    x_{k+2} = \frac{N - k + 1}{2} x_{k+1} - \frac{k}{2} x_k. \label{eq:recursive relation x_k}
\end{equation}

Fix an integer \( N \ge 3 \). Define a new sequence \( \{g_k\}_{k \ge 0} \) satisfying the same recurrence:
\[
g_{k+2} = \frac{N - k + 1}{2} g_{k+1} - \frac{k}{2} g_k, \quad \text{with } g_0 = 0,\, g_1 = 1.
\]
Then, by uniqueness of solutions to the recurrence, we have \( x_k = g_k \) for all \( k \le \left\lfloor \frac{N}{2} \right\rfloor + 1 \).

Let \( G_N(t) := \sum_{k=0}^\infty g_k t^k \) denote the generating function of the sequence \( \{g_k\} \). We now derive a differential equation for \( G_N(t) \). The recurrence implies:
\begin{align*}
\sum_{k=0}^\infty g_{k+2} t^k &= \frac{G_N(t) - t}{t^2}, \\
\sum_{k=0}^\infty \frac{N - k + 1}{2} g_{k+1} t^k &= -\frac{1}{2} G_N'(t) + \frac{N + 2}{2} \cdot \frac{G_N(t)}{t}, \\
\sum_{k=0}^\infty \frac{kg_k}{2} t^k &= \frac{1}{2} t G_N'(t).
\end{align*}

Substituting these into the recurrence yields the following first-order linear differential equation:
\begin{equation}
    (t^2 + t^3) G_N'(t) + (2 - (N + 2)t) G_N(t) - 2t = 0. \label{eq:original differential equation}
\end{equation}

We first solve the associated homogeneous equation:
\begin{equation}
    (t^2 + t^3) G_N'(t) + (2 - (N + 2)t) G_N(t) = 0. \label{eq:homogeneous equation}
\end{equation}
Rewriting this as a separable ODE gives:
\[
\frac{dG_N(t)}{G_N(t)} = \frac{(N + 2)t - 2}{t^2(1 + t)} \, dt.
\]
Integrating both sides, we find that the general solution to~\eqref{eq:homogeneous equation} is:
\[
G_N(t) = C \left( \frac{t}{1 + t} \right)^{N + 4} e^{\frac{2}{t}},
\]
for some constant \( C \).

To solve the original inhomogeneous equation~\eqref{eq:original differential equation}, we apply the method of variation of parameters. We consider a solution of the form
\begin{equation}
    G_N(t) = C(t) \left( \frac{t}{1 + t} \right)^{N + 4} e^{\frac{2}{t}}. \label{eq:variation ansatz}
\end{equation}
Substituting~\eqref{eq:variation ansatz} into~\eqref{eq:original differential equation} and simplifying yields the following differential equation for \( C(t) \):
\begin{equation}
    C'(t) = \frac{2}{t^2} \left( \frac{1 + t}{t} \right)^{N + 3} e^{- \frac{2}{t}}. \label{eq:C differential equation}
\end{equation}

For each integer \( n \ge 0 \), define
\[
J_n(t) := \int \frac{1}{t^2} \left( \frac{1 + t}{t} \right)^n e^{- \frac{2}{t}} \, dt.
\]
Then the general solution to~\eqref{eq:original differential equation} is
\[
G_N(t) = \left( 2 J_{N + 3}(t) + C_0 \right) \left( \frac{t}{1 + t} \right)^{N + 4} e^{\frac{2}{t}},
\]
for some constant \( C_0 \).

We now compute \( J_n(t) \). For \( n \ge 1 \), we have:
\begin{align*}
    2J_n(t) &= \int \frac{2}{t^2} \left( \frac{1 + t}{t} \right)^n e^{- \frac{2}{t}} \, dt \\
    &= \int \left( \frac{1 + t}{t} \right)^n \, d \left( e^{- \frac{2}{t}} \right) \\
    &= \left( \frac{1 + t}{t} \right)^n e^{- \frac{2}{t}} + n \int \frac{1}{t^2} \left( \frac{1 + t}{t} \right)^{n - 1} e^{- \frac{2}{t}} \, dt \\
    &= \left( \frac{1 + t}{t} \right)^n e^{- \frac{2}{t}} + n J_{n - 1}(t).
\end{align*}
In particular,
\[
2J_0(t) = \int \frac{2}{t^2} e^{- \frac{2}{t}} \, dt = e^{- \frac{2}{t}} + C.
\]
By induction, we obtain:
\[
J_n(t) = \sum_{i = 0}^{n} \frac{n!}{2^{n - i + 1} i!} \left( \frac{1 + t}{t} \right)^i e^{- \frac{2}{t}} + C,
\]
for some constant \( C \). Hence,
\begin{equation}
    C(t) = \sum_{i = 0}^{N + 3} \frac{(N + 3)!}{2^{N - i + 3} i!} \left( \frac{1 + t}{t} \right)^i e^{- \frac{2}{t}} + C. \label{eq:Ct expression}
\end{equation}

Substituting into~\eqref{eq:variation ansatz} and $G_N(0)=0$, we obtain
\begin{equation}
    G_N(t) = \sum_{i = 0}^{N + 3} \frac{(N + 3)!}{2^{N - i + 3} i!} \left( \frac{t}{1 + t} \right)^{N + 4 - i}. \label{eq:explicit expression G_N}
\end{equation}

To extract \( g_k \), we compute \( G_N^{(k)}(0) \). Let \( f_i(t) := \left( \frac{t}{1 + t} \right)^i \). Then:
\[
f_i^{(k)}(0) =
\begin{cases}
0, & \text{if } k < i, \\
(-1)^{k - i} k! \binom{k - 1}{i - 1}, & \text{if } k \ge i.
\end{cases}
\]
Hence,
\begin{align*}
    k! g_k &= G_N^{(k)}(0) = \sum_{i = 0}^{N + 3} \frac{(N + 3)!}{2^{N - i + 3} i!} f_{N - i + 4}^{(k)}(0) \\
    &= \sum_{i = 0}^{N + 3} \frac{(N + 3)!}{2^{N - i + 3} i!} (-1)^{k - N + i - 4} k! \binom{k - 1}{N - i + 3}.
\end{align*}
Therefore, for \( k \le \left\lfloor \frac{N}{2} \right\rfloor + 1 \), we have:
\begin{align*}
    x_k = g_k 
    &= \sum_{i = N - k + 4}^{N + 3} \frac{(N + 3)!}{2^{N - i + 3} i!} (-1)^{k - N + i - 4} \binom{k - 1}{N - i + 3} \\
    &= \sum_{j = 0}^{k - 1} \frac{(N + 3)!}{2^j (N - j + 3)!} (-1)^{k - j - 1} \binom{k - 1}{j}.
\end{align*}

For \( k_0 := \left\lfloor \frac{N}{2} \right\rfloor + 2 \), we have:
\begin{align}
    x_{k_0} &= \frac{k_0 - 1}{2k_0 - N} (x_{k_0 - 1} + z_{k_0 - 1}) \label{eq:xk0 step 1} \\
    &= \frac{k_0 - 1}{2k_0 - N} \left( 2 x_{k_0 - 1} - \frac{k_0 - 2}{2} x_{k_0 - 2} \right). \label{eq:xk0 step 2}
\end{align}
Substituting the closed forms of \( x_{k_0 - 1} \) and \( x_{k_0 - 2} \), we obtain:
\begin{align*}
    x_{k_0} = \frac{k_0 - 1}{2k_0 - N}\cdot \sum_{j = 0}^{k_0 - 2} \left((-1)^{k_0 - j - 2}\cdot\frac{(N + 3)!}{2^j (N - j + 3)!} \cdot \frac{k_0 - j + 2}{2} \cdot\binom{k_0 - 2}{j}\right).
\end{align*}

For \( k > k_0 \), by iterating the recurrence in~\eqref{eq:case 2}, we get:
\begin{align*}
    x_k &= x_{k_0} \prod_{i = k_0 + 1}^{k} \frac{i - 1}{2i - N} \\
    &= \left( \prod_{i = k_0}^{k} \frac{i - 1}{2i - N} \right) \cdot \sum_{j = 0}^{k_0 - 2} \left((-1)^{k_0 - j - 2}\cdot\frac{(N + 3)!}{2^j (N - j + 3)!} \cdot \frac{k_0 - j + 2}{2} \cdot\binom{k_0 - 2}{j}\right).
\end{align*}
This completes the proof of~\cref{lemma:specific formula of xk}.

\end{document}